\documentclass[11pt]{article}

\newcommand{\bydef}{\coloneqq}

\usepackage[margin=1in]{geometry}
\usepackage{amsmath,amssymb,amsthm,mathtools, mathrsfs, enumitem}
\usepackage{commath}
\usepackage{comment}
\usepackage{mathtools}
\usepackage{bm}
\usepackage{tocloft}
\usepackage{authblk}
\mathtoolsset{showonlyrefs}
\usepackage[dvipsnames]{xcolor}
\usepackage{subcaption, hyperref}
\usepackage{graphicx}
\usepackage{placeins}
\usepackage{float}
\graphicspath{{figures/}}


\theoremstyle{plain}
\newtheorem{theorem}{Theorem}[section]
\newtheorem{lemma}[theorem]{Lemma}
\newtheorem{corollary}[theorem]{Corollary}
\newtheorem{proposition}[theorem]{Proposition}

\newtheorem{assumption}[theorem]{Assumption}

\theoremstyle{remark}
\newtheorem{remark}[theorem]{Remark}

\newcommand{\Sym}{\textsf{Sym}_d}

\newcommand{\CN}{\mathcal{N}}
\newcommand{\CA}{\mathcal{A}}
\newcommand{\CR}{\mathcal R}

\newcommand{\CH}{\mathcal H}
\newcommand{\CE}{\mathcal E}

\newcommand{\CB}{\mathcal B}
\renewcommand{\P}{\mathbb P}
\newcommand{\CK}{\mathcal K}

\newcommand{\FT}{\mathfrak T}
\newcommand{\FB}{\mathfrak B}
\newcommand{\BFB}{\mathbf{B}}
\newcommand{\BFT}{\mathbf{T}}
\newcommand{\FA}{\mathfrak A}
\newcommand{\FC}{\mathfrak C}
\newcommand{\FQ}{\mathfrak{Q}}
\newcommand{\SA}{\mathscr{A}}
\newcommand{\CP}{\mathcal P}
\newcommand{\CU}{\mathcal U}

\newcommand{\mmat}{\text{mat}}

\newcommand{\bone}{\mathbb I}
\newcommand{\E}{\mathbb{E}}
\newcommand{\Tr}{\mathrm{Tr}}

\newcommand{\smi}{{\setminus i}}
\newcommand{\R}{\mathbb R}

\newcommand{\I}{\mathrm{I}}

\newcommand{\svecop}{\operatorname{vec}}

\newcommand{\mat}{\mathrm{mat}}

\newcommand{\CM}{\mathcal M}

\newcommand{\iids}{\stackrel{i.i.d}{\sim}}
\newcommand{\Var}{\text{Var}}
\newcommand{\op}{\mathrm{op}}

\newcommand{\lin}{\mathrm{lin}}
\newcommand{\inner}[2]{\left\langle #1,#2\right\rangle}
\renewcommand{\norm}[1]{\left\lVert #1\right\rVert}
\renewcommand{\abs}[1]{\left\lvert #1\right\rvert}
\renewcommand{\I}{\mathrm{I}}

\newcommand{\HS}{\mathrm{HS}}

\newcommand{\quadt}{\mathrm{quad}}

\newcommand{\Frob}{\mathrm{F}}

\newcommand{\CO}{\mathcal{O}}
\newcommand{\BS}{\mathbb{S}}

\newcommand{\SC}{\mathscr{C}}

\newcommand{\CG}{\mathcal{G}}
\newcommand{\bom}{\bar{\Omega}}
\newcommand{\tom}{\tilde{\Omega}}
\newcommand{\veps}{\varepsilon}

\newcommand{\Cov}{\text{Cov}}
\newcommand{\spec}{\mathrm{spec}}
\newcommand{\diag}{\mathrm{diag}}
\newcommand{\main}{\mathrm{main}}

\newcommand{\fb}{\mathrm{fb}}

\renewcommand{\bf}{\mathbf}
\newcommand{\hS}{\widehat{S}}
\newcommand{\ro}{\mathrm{RO}}
\newcommand{\sig}{\mathrm{sig}}
\newcommand{\cross}{\mathrm{cross}}
\renewcommand{\sf}{\mathsf}
\newcommand{\CD}{\mathcal D}
\renewcommand{\bydef}{\mathrel{\mathop{:}}=}

\renewcommand{\Re}{\text{Re }}
\newcommand{\sfe}{\mathsf{E}}
\newcommand{\hbs}{\widehat{\bf s}}

\usepackage{hyperref}
 \hypersetup{
     colorlinks=true,
     linkcolor=blue,
     filecolor=blue,
     citecolor = blue,
     urlcolor=cyan,
     pdftitle={Quantitative Universality of Approximate Message Passing for Rank-One Quadratic Sensing},
     pdfauthor={Duy Thuc Nguyen, Yue M. Lu, Subhabrata Sen},
     pdfkeywords={approximate message passing, universality, quadratic sensing, state evolution},
 }

\title{Quantitative Universality of Approximate Message Passing for Rank-One Quadratic Sensing}
\author[1]{Duy Thuc Nguyen}
\author[1,2]{Yue M. Lu}
\author[1]{Subhabrata Sen}
\affil[1]{Department of Statistics, Harvard University}
\affil[2]{Applied Mathematics, Harvard University}
\date{October 2026}

\begin{document}

\maketitle

\begin{abstract}
Approximate Message Passing (AMP) algorithms are attractive as they are computationally efficient and simultaneously admit a precise characterization in terms of the low-dimensional `state-evolution' recursion. In this work, we establish quantitative universality for AMP under rank-one sensing: its main observables agree asymptotically with those under a covariance-matched isotropic Gaussian ensemble. Specifically, we consider rank-one sensing matrices defined as
\[
Z_i=\frac{x_i x_i^\top-I_d}{\sqrt d},
\qquad
x_i\stackrel{\mathrm{i.i.d.}}{\sim}\mathcal N(0,I_d).
\]
These matrices arise in quadratic regression and in learning quadratic neural networks. Lifting the inputs gives a linear regression problem over symmetric matrices. Their normalized vectorizations have covariance $I_p/p$, where $p=d(d+1)/2$. The $p$ coordinates are strongly dependent, since they come from only $d$ Gaussian variables.

We consider observations $y_i=\operatorname{Tr}(Z_iS^\star)+\epsilon_i$ with Gaussian noise. The spectral denoisers are prescribed functions with uniformly bounded values and first four derivatives. We work with $n\leq Cp$ and a fixed number of iterations, under uniform bounds on the normalized signal energy, initialization, and algorithm coefficients. We show that replacing rank-one sensing by Gaussian sensing gives asymptotically the same normalized signal overlaps, overlaps between iterates, and smooth linear spectral statistics. For each normalized observable, the difference is $\CO(d^{-1/2})$ in every fixed $L^m$. Consequently, the normalized mean-squared error and fixed spectral moments are universal, and the empirical spectral distributions agree asymptotically. Furthermore, whenever these observables admit a joint state-evolution limit under Gaussian sensing, the same predictions hold under rank-one sensing. The proof uses a columnwise Lindeberg argument, leave-one-out estimates, smooth spectral expansions, and high-moment concentration. For this class of denoisers, our result establishes the AMP universality conjectured by Maillard et al. (2024) for learning quadratic neural networks.
\end{abstract}

\clearpage
\tableofcontents
\setcounter{tocdepth}{2}
\clearpage

\section{Introduction}\label{sec:introduction}

Approximate message passing (AMP) is a family of low-complexity iterative algorithms for high-dimensional inference \cite{DonohoDavidL2009Mafc,Rangan2011}. For sensing matrices with i.i.d.\ Gaussian entries, its asymptotic dynamics are described by \emph{state evolution}, a low-dimensional recursion that gives precise predictions for quantities such as the mean-squared error \cite{BayatiM2011TDoM,JavanmardMontanari2013,BerthierRaphael2020Sefa}. State evolution has also been established for algorithms adapted to rotationally invariant matrices \cite{Fan2022}.

AMP universality has been proved for sensing matrices with independent entries, using Gaussian ensembles with matching moments \cite{BayatiLelargeMontanari2015,ChenLam2021,Han2025,LovigWangFan2025}. For semi-random and delocalized matrices, the comparison is with rotationally invariant ensembles with matching spectral data \cite{DudejaLuSen2023,DudejaSenLu2024,WangZhongFan2024}. We consider sensing columns formed by a centered quadratic lift of independent Gaussian vectors. The columns are independent, but the entries within each column are strongly dependent.

Let $x_1,\ldots,x_n\iids\CN(0,I_d)$ be Gaussian inputs and define the lifted features
\begin{equation}\label{eq:rank-one-features}
   Z_i=\frac{x_i x_i^\top-I_d}{\sqrt d}, \qquad f_i=\frac{\svecop(Z_i)}{\sqrt{d+1}}\in\R^p,
    \qquad p=\frac{d(d+1)}2,
\end{equation}
where $\svecop$ maps a symmetric $d\times d$ matrix to the vector of its $p$ nonredundant entries, normalized so that inner products and norms are preserved; see \eqref{eq:svec-isometry}. For every $i\in[n]$ and every symmetric matrix $S$, these features satisfy the lifting identity
\begin{equation}\label{eq:lifting-identity}
    \frac{x_i^\top Sx_i-\Tr (S)}{\sqrt d}
    =\Tr(Z_iS)
    =\sqrt{d+1}\,\big\langle f_i,\svecop(S)\big\rangle.
\end{equation}
A quadratic function of $x_i$ is therefore linear in $f_i$. The lifted vector is non-Gaussian, with $p=\Theta(d^2)$ coordinates generated from only $d$ Gaussian variables.

Our main motivation is learning quadratic neural networks \cite{MaillardAntoine2024Bloa}. With one hidden layer, a quadratic activation, and weights $W=[w_1,\ldots,w_m]\in\R^{d\times m}$, such a network computes $x\mapsto x^\top S(W)x$, up to normalization. Here $S(W)=WW^\top=\sum_{k\le m}w_kw_k^\top\succeq0$. By \eqref{eq:lifting-identity}, its centered output at $x_i$ is $\Tr(Z_iS(W))$. With noisy samples, we estimate a symmetric matrix $S^\star$ from the observations
\begin{equation}\label{eq:intro-observation-model}
    y_i=\Tr(Z_iS^\star)+\epsilon_i=\sqrt{d+1}\,\big\langle f_i,\svecop(S^\star)\big\rangle+\epsilon_i,
    \qquad 1\leq i \leq n,
\end{equation}
where $\epsilon_i \sim \mathcal{N}(0, \Delta)$ i.i.d. This is linear regression in the lifted space. When the width $m$ is proportional to $d$, the parameter count $md$ and the lifted dimension $p$ are both of order $d^2$. This motivates the sample scale $n\asymp p$. The same rank-one measurements arise in matrix sensing, including phase retrieval and low-rank matrix recovery \cite{KuengRauhutTerstiege2017}.

Maillard et al.\ \cite{MaillardAntoine2024Bloa} derived a Bayesian AMP algorithm for this lifted problem, with sensing matrix $F=[f_1,\ldots,f_n]\in\R^{p\times n}$. We write a schematic AMP recursion in feature coordinates as \cite{JavanmardMontanari2013,BerthierRaphael2020Sefa}:
\begin{equation}\label{eq:general-amp}
    m^t=\mathcal M_t\big(y,F^\top s^t,F^\top s^{t-1},\ldots,F^\top s^0\big),
    \qquad
    s^{t+1}=\mathcal D_t\big(s^t+Fm^t\big).
\end{equation}
Here $s^t\in\R^p$ is the $t$-th iterate and $m^t\in\R^n$ is a vector of messages, one per sensing feature. The vector $F^\top s^t=(\langle f_i,s^t\rangle)_{i\le n}$ lists the projections of the current iterate onto the features. The vector $Fm^t=\sum_{i\le n} m_i^t f_i$ weights each feature column by its message. The dependence on earlier projections includes the memory, or Onsager, correction.

Each vector iterate $s^t$ corresponds to a symmetric $d\times d$ matrix $\widehat S^t$. The sensing operations use vector coordinates, while the denoiser $\mathcal D_t$ acts on the eigenvalues of the matrix corresponding to $s^t+Fm^t$. For $M=U\Lambda U^\top$, with $U$ orthogonal and $\Lambda=\diag(\lambda_1,\ldots,\lambda_d)$, a spectral denoiser applies a scalar function $\sigma:\R\to\R$ to the eigenvalues while keeping the eigenvectors:
\begin{equation}\label{eq:intro-spectral-denoiser}
    \sigma(M)=U\diag\big(\sigma(\lambda_1),\ldots,\sigma(\lambda_d)\big)U^\top.
\end{equation}
A nonlinear spectral map is generally nonseparable in vector coordinates: an output coordinate may depend on the entire input matrix.

We compare the rank-one features with Gaussian features having the same mean and covariance. Let
\begin{equation}\label{eq:gaussian-surrogate-features}
    \widetilde f_i\iids\CN(0,I_p/p),
    \qquad i\in[n];
\end{equation}
equivalently, $\widetilde f_i=\svecop(\widetilde Z_i)/\sqrt{d+1}$, where $\widetilde Z_i$ is drawn from the Gaussian orthogonal ensemble (GOE) with the same entrywise variances as $Z_i$; Section~\ref{sec:model-setup} gives the precise normalization. The two ensembles satisfy
\begin{equation}\label{eq:column-covariance}
    \E[f_i]=\E[\widetilde f_i]=0,
    \qquad
    \E[f_i f_i^\top]=\E[\widetilde f_i\widetilde f_i^\top]=\frac{I_p}{p}.
\end{equation}
Maillard et al.\ describe state evolution for AMP with these Gaussian features in \cite[Appendix~D.9]{MaillardAntoine2024Bloa}, using the theory of Berthier, Montanari, and Nguyen for nonseparable nonlinearities \cite{BerthierRaphael2020Sefa}. They conjectured that the same predictions hold for the original rank-one features, whose coordinates are dependent.

We prove quantitative universality for AMP in the linear observation model \eqref{eq:intro-observation-model} (Theorem~\ref{thm:final-universality}). The result establishes the conjectured comparison for prescribed spectral denoisers in $C_b^4$, with bounds uniform in $d$. Fix an iteration horizon $T$, iterations $s,t\le T$, and a test function $\mu$ with three continuous derivatives. We compare the normalized observables
\begin{equation}\label{eq:universality-observables}
    \frac1d\Tr(\widehat S^t S^\star),
    \qquad
    \frac1d\Tr(\widehat S^t\widehat S^s),
    \qquad
    \frac1d\Tr\big(\mu(\widehat S^t)\big).
\end{equation}
For each observable, the difference between its values along the rank-one and Gaussian trajectories is $\CO(d^{-1/2})$ in every fixed $L^m$ norm.

The first observable measures the alignment of $\widehat S^t$ with $S^\star$. The second measures the overlap between iterates; for $s=t$, it is the normalized squared norm $d^{-1}\|\widehat S^t\|_\Frob^2$. Since
\[
    \frac1d\|\widehat S^t-S^\star\|_\Frob^2
    =\frac1d\|\widehat S^t\|_\Frob^2
    -\frac2d\Tr(\widehat S^tS^\star)
    +\frac1d\|S^\star\|_\Frob^2,
\]
the first two comparisons yield universality of the normalized mean-squared error. The third observable is a smooth linear statistic of the spectrum: if $\lambda_1,\ldots,\lambda_d$ are the eigenvalues of $\widehat S^t$, then
\[
    \frac1d\Tr\big(\mu(\widehat S^t)\big)
    =\frac1d\sum_{j=1}^d \mu(\lambda_j),
\]
so varying $\mu$ tests the empirical eigenvalue distribution of the iterate. The uniform operator-norm bound on the iterates also gives universality of every fixed spectral moment $d^{-1}\Tr\big((\widehat S^t)^k\big)$ (Corollary~\ref{cor:power-from-spectral}). Whenever the Gaussian recursion admits a state-evolution limit for one of these observables, the same limit holds for the rank-one recursion (Section~\ref{subsec:se-transfer}).

The proof combines a columnwise Lindeberg interpolation \cite{Chatterjee2006,KoradaMontanari2011} with a leave-one-out analysis of the AMP dynamics \cite{BaoHanXu2025}. We replace whole sensing columns because their coordinates are dependent. At each replacement, we remove the active column and run the background recursion, then derive an exact recursion for the change caused by reinserting it. The inserted column is independent of the background trajectory, so we can compare its contributions under the two ensembles.

The spectral denoiser \eqref{eq:intro-spectral-denoiser} couples the matrix entries. Its response to an inserted column feeds into later iterations, where the same sensing matrix is used again. A single-column change affects a normalized observable at the $\CO(d^{-3/2})$ scale. With $n=\Theta(d^2)$ replacements, the Lindeberg argument requires the difference between the expected contributions of the two candidate columns to be $\CO(d^{-5/2})$.

We linearize the response recursion and expand the spectral map with a fourth-order Fr\'echet remainder. Contributions determined by the first two column moments agree between the ensembles. For the remaining terms, we use dimension-free Frobenius bounds on spectral derivatives and Hilbert--Schmidt estimates for their scalar contractions. We also use a refined bound for centered quadratic forms in a rank-one column that retains the Frobenius norm of the coefficient matrix. These estimates compare the conditional means. A high-moment Efron--Stein argument \cite{BoucheronStephane2005MIfF}, based on the same leave-one-out construction, controls fluctuations at the $d^{-1/2}$ scale.

\subsection{Related Work}\label{sec:related-work}

\paragraph{Quadratic sensing.}
Maillard et al.\ analyzed GAMP for the extensive-width quadratic neural network after replacing the rank-one features by the Gaussian surrogate \eqref{eq:gaussian-surrogate-features}; they left universality of the AMP dynamics under rank-one sensing as a conjecture \cite{MaillardAntoine2024Bloa}. Xu et al.\ later proved the corresponding universality of the limiting free entropy and Bayes-optimal MMSE \cite{XuMaillardZdeborovaKrzakala2025}. Their result is at the level of the free entropy and does not compare algorithmic iterates. Our theorem supplies the complementary dynamical statement, comparing AMP observables after any fixed number of iterations under rank-one and Gaussian sensing, with a quantitative rate.

\paragraph{Polynomial-chaos sensing.}
The rank-one feature $f_i=\svecop(x_i x_i^\top-I_d)/\sqrt{d(d+1)}$ is the normalized degree-two Hermite chaos of $x_i$. It is therefore the degree-two member of a broader hierarchy of polynomial-chaos sensing vectors formed by collecting all Hermite features of a fixed degree \cite{NourdinPeccatiReinert2010}. Gaussian equivalence for this hierarchy is, by now, well understood as it governs both the spectra of the associated Gram matrices \cite{LuYau2025,PanditWangZhu2025} and the training and generalization behavior of high-dimensional empirical risk minimization with random features \cite{HuLu2023}. These results compare a random matrix or the solution of a static optimization problem. Here, by contrast, the same chaos columns are reused adaptively at every AMP iteration, and their effect is repeatedly fed back into the spectral denoiser. To our knowledge, this is the first universality theorem for polynomial-chaos sensing at the level of iterative dynamics, proved here for the degree-two member of the hierarchy.

\paragraph{Universality of AMP.}
For matrices with independent entries, Bayati, Lelarge, and Montanari proved universality for polynomial nonlinearities under moment and tail assumptions \cite{BayatiLelargeMontanari2015}, Chen and Lam developed an interpolation approach for Lipschitz separable nonlinearities and data-dependent spectral initialization \cite{ChenLam2021}, and Han established entrywise universality for a broader family of first-order algorithms with independent sub-Gaussian entries \cite{Han2025}. A complementary line studies semi-random ensembles, in which limited randomness is added to an otherwise structured matrix. Dudeja, Lu, and Sen established state-evolution universality for memory-free AMP under semi-random matrix structure \cite{DudejaLuSen2023}. Dudeja, Sen, and Lu developed spectral universality classes for regularized linear regression, covering proximal-gradient dynamics and estimator performance for sensing matrices that may be highly structured or nearly deterministic \cite{DudejaSenLu2024}. Wang, Zhong, and Fan used tensor-network methods to prove AMP universality for generalized Wigner matrices and delocalized sign-and-permutation-invariant ensembles \cite{WangZhongFan2024}. The universality results have been extended to Generalized First Order Methods (GFOMs) in \cite{gorini2026universality}. These results substantially broaden the admissible sensing ensembles, but none covers the rank-one dependence structure studied here.

\paragraph{Nonseparable nonlinearities.}
For Gaussian sensing, Berthier, Montanari, and Nguyen established state evolution for nonseparable Lipschitz nonlinearities \cite{BerthierRaphael2020Sefa}. Lovig, Wang, and Fan proved universality for nonseparable AMP with independent non-Gaussian entries under a bounded-composition condition, including spectral maps for signals with generic singular vectors \cite{LovigWangFan2025}. With rank-one sensing, changing one feature column perturbs the matrix entering the spectral denoiser. The column appears again at later iterations, so its effect must be controlled through the recursion. Our argument compares whole dependent columns and controls the difference in their expected contributions at order $\CO(d^{-5/2})$.

\paragraph{Lindeberg and leave-one-out methods.}
Lindeberg replacement is a classical route to universality \cite{Chatterjee2006,KoradaMontanari2011}. Panahi and Hassibi used it to establish universality for regularized least-squares estimators \cite{PanahiHassibi2017}. Our proof is closest in spirit to the Lindeberg and leave-one-out framework developed by Hu and Lu for shallow random-feature models \cite{HuLu2023} and subsequently extended to deep random-feature models by Bosch, Panahi, and Hassibi \cite{BoschPanahiHassibi2023}. Leave-one-out methods have also yielded precise nonasymptotic descriptions of AMP under Gaussian independent-entry sensing \cite{BaoHanXu2025}. In our setting, however, the independent units are entire sensing columns rather than individual entries. The comparison must therefore follow a rank-one column through a matrix-valued recursion while preserving the dependence within that column.

\subsection{Organization}\label{sec:intro-organization}

Section~\ref{sec:model-setup} defines the two sensing models, states the universality theorem and its corollaries, and presents Gaussian state evolution and its transfer to rank-one sensing. Section~\ref{sec:numerics} presents numerical experiments on state evolution, spectral distributions, convergence rates, and a nonlinear output channel. Section~\ref{sec:proof-main} outlines the proof. Sections~\ref{sec:aux} and \ref{sec:linearized} derive the leave-one-out recursion and its linearization. Sections~\ref{sec:proof-quadratic} and \ref{sec:proof-linear} prove the columnwise comparison estimates; Section~\ref{sec:concentration} proves concentration. Section~\ref{sec:discussion} discusses the scope of the theorem and possible extensions. The appendices contain auxiliary results and an additional experiment on the effect of a spectral spike.

    \qquad

\section{Model and Main Results}\label{sec:model-setup}
This section compares two statistical experiments for estimating a symmetric matrix $S^\star$ from $n$ noisy linear measurements. Both use the same AMP recursion, initialization, coefficients, and spectral denoisers, but one is driven by the centered rank-one matrices $Z_i$ of \eqref{eq:rank-one-features} and the other by their Gaussian surrogates $\widetilde Z_i$. We formulate the experiments at the matrix level, state the assumptions and quantitative universality theorem, and give its consequences for spectral moments and mean-squared error. The vector coordinates used in the proofs are deferred to Section~\ref{sec:proof-main}.\\

\noindent We first introduce some notation that will be used thoroughout the subsequent discussion.
\paragraph{Notation.}
Let $\Sym$ be the space of real symmetric $d\times d$ matrices and let $p\bydef\dim(\Sym)=d(d+1)/2$. Fix an ordering of the index pairs $1\le j\le k\le d$. For $A\in\Sym$, the vector $\svecop(A)\in\R^p$ has coordinates
\begin{equation}\label{eq:svec-coordinate-def}
    [\svecop(A)]_{(j,k)}=
    \begin{cases}
        A_{jj}, & j=k,\\
        \sqrt{2}\,A_{jk}, & j<k,
    \end{cases}
\end{equation}
so that each diagonal entry is recorded once and each nonredundant off-diagonal entry with a factor $\sqrt2$; the inverse map $\mmat\bydef\svecop^{-1}$ restores the symmetric matrix. This normalization makes $\svecop$ a linear isometry: for $A,B\in\Sym$,
\begin{equation}\label{eq:svec-isometry}
    \langle\svecop(A),\svecop(B)\rangle=\Tr(AB).
\end{equation}
For $M=U\Lambda U^\top\in\Sym$ and a scalar function $\varphi:\R\to\R$, the associated spectral map is
\begin{equation}\label{eq:spectral-map-def}
    \varphi(M)=U\varphi(\Lambda)U^\top.
\end{equation}
We write $\|\cdot\|_\op$ and $\|\cdot\|_\Frob$ for the operator and Frobenius norms of a matrix, $\|\cdot\|$ for the Euclidean norm of a vector, $\|M\|_\HS$ for the Hilbert--Schmidt norm of a multilinear form $M$, and $[n]\bydef\{1,\ldots,n\}$. For an integer $k\geq0$, $\SC_b^k(\R)$ denotes the class of $k$-times continuously differentiable functions that are bounded together with their first $k$ derivatives.

For positive sequences, $a_d\lesssim b_d$ means $a_d\le Cb_d$ for a constant $C$ independent of $d$, and $a_d\asymp b_d$ means $a_d\lesssim b_d\lesssim a_d$. We use the stochastic-domination notation of \cite[Definition~2.1]{ErdosKnowlesYauYin2013}: for nonnegative random variables $(Y_d)_{d\ge1}$ and a deterministic positive scale $(a_d)_{d\ge1}$, we write $Y_d\prec a_d$ if $Y_d$ is bounded by $a_d$, up to an arbitrarily small polynomial factor, with failure probability smaller than any fixed inverse power of $d$; formally, for every $\veps,D>0$ there exists $d_0(\veps,D)$ such that
\begin{equation}\label{eq:stochastic-domination}
    \P\big(Y_d>d^\veps a_d\big)\leq d^{-D},
    \qquad \forall d\geq d_0(\veps,D).
\end{equation}
For a possibly signed random variable $X_d$, we write $X_d=\CO_\prec(a_d)$ when $|X_d|\prec a_d$. For $m\ge1$ we write $\norm{X}_{L^m}\bydef(\E\abs{X}^m)^{1/m}$, and $X_d=\CO_{L^m}(a_d)$ if $\norm{X_d}_{L^m}\le Ca_d$ for a constant $C$ independent of $d$. If $X_d=\CO_{L^m}(a_d)$ for every fixed $m\ge1$, with the constant allowed to depend on $m$, then Markov's inequality gives $X_d=\CO_\prec(a_d)$: given $\veps,D>0$, it suffices to take $m>D/\veps$. For a random vector $X$, $\norm{X}_{L^m}$ denotes $\norm{\norm{X}}_{L^m}$, and for a random matrix $X$ we use
\begin{equation}\label{eq:matrix-Lm-norms}
    \norm{X}_{L^m(F)}\bydef\norm{\norm{X}_\Frob}_{L^m},
    \qquad
    \norm{X}_{L^m(\op)}\bydef\norm{\norm{X}_\op}_{L^m}.
\end{equation}
Throughout, $C,C_m,C_\mu,\ldots$ denote finite positive constants that do not depend on $d$ or on the column index $i$, but may depend on the iteration horizon $T$, on the constants in the assumptions, and on the indicated parameters; their values may change from line to line.

\subsection{Observation Model}
Throughout this section, $\sfe\in\{\ro,\mathrm G\}$ labels the rank-one and Gaussian sensing ensembles respectively. Recall the centered rank-one matrices $Z_i$ from \eqref{eq:rank-one-features}. For the Gaussian experiment, let $\widetilde Z_1,\ldots,\widetilde Z_n$ be independent symmetric matrices whose upper-triangular entries are mutually independent and satisfy
\begin{equation}\label{eq:goe-entry-definition}
    \widetilde Z_{i,jk}\sim\CN\left(0,\frac{1+\mathbf 1_{\{j=k\}}}{d}\right),
    \qquad 1\,\leq\, j\,\leq\, k\,\leq\, d.
\end{equation}
Thus every coordinate of $\svecop(\widetilde Z_i)$ has variance $2/d$, and therefore
\begin{equation}\label{eq:goe-feature-verification}
    \frac{\svecop(\widetilde Z_i)}{\sqrt{d+1}}
    \sim\CN\left(0,\frac{2}{d(d+1)}I_p\right)
    =\CN(0,I_p/p),
\end{equation}
which is exactly the Gaussian surrogate in \eqref{eq:gaussian-surrogate-features}. Write $Z_{\ro,i}\bydef Z_i$ and $Z_{\mathrm G,i}\bydef\widetilde Z_i$. The sensing operator in ensemble $\sfe$ and its adjoint are
\begin{equation}\label{eq:sensing-map-def}
    \CA_{\sfe}(S)\bydef\big(\Tr(Z_{\sfe,i}S)\big)_{i=1}^n,
    \qquad
    \CA_{\sfe}^*(r)\bydef\sum_{i=1}^n r_i Z_{\sfe,i}.
\end{equation}

Fix a noise level $\Delta\geq0$. Let $S^\star\in\Sym$ be the signal and let $\epsilon\sim\CN(0,\Delta I_n)$ be the observation noise. The two sensing ensembles are independent of one another and of $(S^\star,\epsilon)$, and the two experiments use the same signal and noise. The observations in ensemble $\sfe$ are
\begin{equation}\label{eq:observation-model}
    y_{\sfe}=\CA_{\sfe}(S^\star)+\epsilon,
    \qquad \sfe\in\{\ro,\mathrm G\}.
\end{equation}
Thus, the only difference between the two experiments is whether the sensing matrices $Z_{\sfe,i}$ are centered rank-one or GOE.

\subsection{The AMP Recursion and Its Observables}
In either ensemble, AMP alternates between two steps. In observation space, it forms the residual between the observations and the current prediction, corrected by a one-step memory term. In matrix space, it back-projects that residual, adds it to the current iterate, and applies a spectral denoiser. Under the coefficient prescription of Section~\ref{subsec:se-transfer}, the memory term is the Onsager correction.

Fix an integer horizon $T\geq1$ and a common initialization $\widehat S^0\in\Sym$, and set $\CB_0\bydef\sigma(S^\star,\widehat S^0)$. For $0\leq t<T$, let $b_t,c_t\in\R$ and $\gamma_t>0$ be common $\CB_0$-measurable coefficients, chosen independently of the sensing matrices and observation noise, and let $\sigma_t:\R\to\R$ be scalar functions that act on symmetric matrices through the spectral map \eqref{eq:spectral-map-def}. Deterministic coefficients are included as a special case. Starting from $\widehat S_{\sfe}^0=\widehat S^0$ and $r_{\sfe}^{-1}=0$, define
\begin{equation}\label{eq:AMP-recursion}
    \begin{aligned}
        r_{\sfe}^t
        &=\frac{y_{\sfe}-\CA_{\sfe}(\widehat S_{\sfe}^t)+b_t r_{\sfe}^{t-1}}{\gamma_t},\\
        R_{\sfe}^t
        &=\widehat S_{\sfe}^t+c_t\CA_{\sfe}^*(r_{\sfe}^t),\\
        \widehat S_{\sfe}^{t+1}
        &=\sigma_t(R_{\sfe}^t),
        \qquad \sfe\in\{\ro,\mathrm G\}.
    \end{aligned}
\end{equation}
We call $r_{\sfe}^t$ the residual, scaled by $\gamma_t^{-1}$, and $R_{\sfe}^t$ the pre-denoising matrix. The first line includes a memory term that corrects the leading self-interaction caused by reusing the same sensing matrices at every iteration. The second line back-projects the residual to the $d\times d$ signal space, and the last line transforms the eigenvalues of the resulting matrix while preserving its eigenvectors. The factor $\gamma_t$ is a residual normalization; the effective Gaussian variance arising in state evolution is denoted by $\tau_t^2$ in Section~\ref{subsec:se-transfer}.

At the proportional sample scale $n\asymp d^2$, the natural back-projection scale is $c_t\asymp d^{-1}$: an entrywise order-one residual typically gives $\|\CA_{\sfe}^*(r)\|_\Frob$ of order $d^{3/2}$, while a matrix with bounded operator norm has Frobenius norm of order $d^{1/2}$. Section~\ref{subsec:se-transfer} gives the particular coefficients under which the Gaussian recursion takes the standard AMP form; the universality theorem itself requires only the bounds in Assumption~\ref{ass:basic}.

The recursion produces a matrix-valued trajectory $\widehat S_{\sfe}^0,\ldots,\widehat S_{\sfe}^T$. In ensemble $\sfe$, the three observables introduced in Section~\ref{sec:introduction} are
\begin{equation}\label{eq:observable-notation}
    \Phi_{\sfe,t}^{\sig}\bydef\frac1d\Tr(\widehat S_{\sfe}^tS^\star),
    \qquad
    \Phi_{\sfe,s,t}^{\cross}\bydef\frac1d\Tr(\widehat S_{\sfe}^s\widehat S_{\sfe}^t),
    \qquad
    \Phi_{\sfe,t,\mu}^{\spec}\bydef\frac1d\Tr\big(\mu(\widehat S_{\sfe}^t)\big).
\end{equation}
The signal overlap measures the alignment of an iterate with $S^\star$. The cross-time overlap measures the alignment of two iterates; for $s=t$, it is the normalized squared Frobenius norm. The spectral statistic tests the empirical eigenvalue distribution against $\mu$. The normalization by $d$ keeps all three at order one under the assumptions below.

\subsection{Main Universality Result}
We impose two standing assumptions. They control the sample size (as a function of the dimension), the coefficients $b_t,c_t,\gamma_t$ in the AMP recursion \eqref{eq:AMP-recursion}, and the regularity of the denoisers.
\begin{assumption}\label{ass:basic}
    The coefficients are $\CB_0$-measurable, and there exists a deterministic $C_0<\infty$, independent of $d$, such that, almost surely,
    \begin{equation}\label{eq:basic-coefficient-assumption}
        n\leq C_0p,
        \qquad
        \max_{0\,\leq\, t\,<\,T}\big(|b_t|+\gamma_t+d|c_t|+\gamma_t^{-1}\big)\,\leq\, C_0.
    \end{equation}
    Moreover, $\sigma_t\in\SC_b^4(\R)$ for every $0\leq t<T$, with the functions and their derivatives through order four bounded uniformly in $d$.
\end{assumption}
The second assumption keeps the normalized Frobenius norm of the signal and the eigenvalues of the initialization bounded.
\begin{assumption}\label{ass:signal-init}
    The pair $(S^\star,\widehat S^0)$ is independent of the two sensing ensembles and the observation noise. There exists $C_1<\infty$, independent of $d$, such that
    \begin{equation}\label{eq:signal-init-bound}
        d^{-1/2}\|S^\star\|_\Frob+\|\widehat S^0\|_\op\leq C_1
        \qquad\text{almost surely}.
    \end{equation}
\end{assumption}
Consequently,
\begin{equation}\label{eq:signal-init-frobenius-bound}
    \frac1d\|S^\star\|_\Frob^2+\frac1d\|\widehat S^0\|_\Frob^2\leq C_1^2
    \qquad\text{almost surely}.
\end{equation}

Let $C_\ast$ be a deterministic uniform bound for $C_1$ and for $\|\sigma_t\|_\infty$ over all $d$ and $0\leq t<T$, and set $I_\ast\bydef[-C_\ast-1,C_\ast+1]$. Every iterate in either experiment has spectrum in $[-C_\ast,C_\ast]$.

We can now state the main result. For a fixed choice among the signal overlap at time $t$, the cross-time overlap at $(s,t)$, and the spectral statistic at $(t,\mu)$, let $\Phi_{\sfe}$ denote its value in ensemble $\sfe$.
\begin{theorem}\label{thm:final-universality}
    Let Assumptions \ref{ass:basic} and \ref{ass:signal-init} hold, and let $0\leq s,t\leq T$, $\mu\in\SC^3(\R)$, and an integer $m\geq2$ be fixed. There exists $C_{m,\mu}<\infty$, independent of $d$, such that, for each of the three observables in \eqref{eq:observable-notation},
    \begin{equation}\label{eq:main-universality-moments}
        \norm{\Phi_\ro-\Phi_{\mathrm G}}_{L^m}
        \leq C_{m,\mu}d^{-1/2}.
    \end{equation}
    Consequently,
    \begin{equation}\label{eq:main-universality-stochastic}
        \Phi_\ro-\Phi_{\mathrm G}=\CO_\prec(d^{-1/2}).
    \end{equation}
\end{theorem}
For the signal and cross-time overlaps, the dependence on $\mu$ is absent and we write the constant simply as $C_m$. The constant is uniform over $s,t\leq T$ and over all coefficient sequences satisfying \eqref{eq:basic-coefficient-assumption}. It depends only on $m,T,C_0,C_1,\Delta$, the uniform bounds on $\sigma_t^{(r)}$ for $0\leq r\leq4$, and, for the spectral statistic, the bounds on $\mu^{(r)}$ over $I_\ast$ for $0\leq r\leq3$. Enlarging the constant absorbs the finitely many small dimensions excluded by the asymptotic estimates. In particular, Jensen's inequality shows that the means of $\Phi$ in the two ensembles differ by $\CO(d^{-1/2})$.

The allowance for $\CB_0$-measurable coefficients follows by conditioning on $(S^\star,\widehat S^0)$: the coefficients are then fixed, and all comparison and concentration estimates are uniform under \eqref{eq:basic-coefficient-assumption}.

\begin{remark}
    The common coupling makes \eqref{eq:main-universality-moments} a direct comparison of the random observables. It is not essential: the proof separately compares their conditional means given $(S^\star,\widehat S^0)$ and controls their fluctuations, so the same rate follows for independent noise realizations provided the signal and initialization remain shared. The proof is assembled in Section~\ref{sec:proof-main}, using the comparison and concentration estimates developed in the subsequent sections.
\end{remark}

The theorem immediately yields universality of two familiar performance measures, the spectral moments and the normalized mean-squared error, as well as convergence of the empirical spectral distributions under the two experiments.
\begin{corollary}\label{cor:power-from-spectral}
    Under Assumptions \ref{ass:basic} and \ref{ass:signal-init}, for every fixed integer $k\geq1$, every fixed $t\leq T$, and every integer $m\geq2$,
    \begin{equation}\label{eq:power-mse-universality-moments}
        \begin{aligned}
            \norm{\frac{1}{d}\Tr((\widehat S^t_\ro)^k)-\frac{1}{d}\Tr((\widehat S^t_{\mathrm G})^k)}_{L^m}
            &\leq C_{m,k}d^{-1/2},\\
            \norm{\frac{1}{d}\|\widehat S^t_\ro-S^\star\|_\Frob^2-\frac{1}{d}\|\widehat S^t_{\mathrm G}-S^\star\|_\Frob^2}_{L^m}
            &\leq C_md^{-1/2}.
        \end{aligned}
    \end{equation}
    Both differences in \eqref{eq:power-mse-universality-moments} are also $\CO_\prec(d^{-1/2})$. Moreover, if
    \begin{equation*}
        \widehat\nu_{\sfe,t}\bydef\frac1d\sum_{j=1}^d\delta_{\lambda_j(\widehat S_{\sfe}^t)},
    \end{equation*}
    then
    \begin{equation}\label{eq:esd-bl-universality}
        d_{\mathrm{BL}}(\widehat\nu_{\ro,t},\widehat\nu_{\mathrm G,t})\xrightarrow[d\to\infty]{\P}0,
    \end{equation}
    where the bounded-Lipschitz distance is
    \begin{equation*}
        d_{\mathrm{BL}}(\nu,\nu')
        \bydef\,\sup_{\substack{\|\varphi\|_\infty\,\leq\,1\\ \operatorname{Lip}(\varphi)\,\leq\,1}}
        \left|\int\varphi\,\mathrm d\nu-\int\varphi\,\mathrm d\nu'\right|.
    \end{equation*}
\end{corollary}
\begin{proof}[Proof of Corollary \ref{cor:power-from-spectral}]
    The first bound is Theorem~\ref{thm:final-universality} with $\mu(x)=x^k$. For $\sfe\in\{\ro,\mathrm G\}$,
    \begin{equation*}
        \frac1d\|\widehat S^t_{\sfe}-S^\star\|_\Frob^2
        =\frac1d\Tr((\widehat S^t_{\sfe})^2)-\frac2d\Tr(\widehat S^t_{\sfe}S^\star)+\frac1d\|S^\star\|_\Frob^2.
    \end{equation*}
    Subtracting the two identities cancels the signal-norm term, so the second bound follows from the triangle inequality, the first bound with $k=2$, and the signal-overlap comparison in Theorem~\ref{thm:final-universality}. The stochastic-domination claims follow from the all-moment bounds and Markov's inequality.

    For \eqref{eq:esd-bl-universality}, both empirical measures are supported on $[-C_\ast,C_\ast]$. The unit ball of bounded Lipschitz functions on this interval is compact in the uniform norm and admits, for every fixed accuracy, a finite net of smooth functions. Theorem~\ref{thm:final-universality} gives convergence in probability for every function in this finite net; approximation in the uniform norm then gives \eqref{eq:esd-bl-universality}.
\end{proof}

\subsection{State Evolution and Its Transfer to Rank-One Sensing}\label{subsec:se-transfer}
Theorem~\ref{thm:final-universality} turns any asymptotic description of the Gaussian trajectory into a description of the rank-one trajectory. For any one of the three observables and any random variable $\phi_d$ on the common probability space, \eqref{eq:main-universality-stochastic} gives the exact decomposition
\begin{equation}\label{eq:transfer-principle}
    \Phi_\ro-\phi_d
    =\Phi_{\mathrm G}-\phi_d+\CO_\prec(d^{-1/2}).
\end{equation}
Thus, if $\Phi_{\mathrm G}-\phi_d\to0$ in probability, then $\Phi_\ro-\phi_d\to0$ in probability as well.

To apply this transfer principle, it remains to identify the Gaussian prediction $\phi_d$. Maillard et al.\ sketch an informal state evolution for the Bayes-optimal GAMP algorithm with the rotationally invariant estimator, after replacing the rank-one sensing matrices by GOE matrices \cite[Appendix~D.9]{MaillardAntoine2024Bloa}. Their recursion tracks scalar order parameters describing the marginal performance at each iteration, including the signal overlap and the squared norm of the current iterate. It does not specify the joint law of iterates at different times, and hence does not provide the cross-time overlaps in \eqref{eq:observable-notation}, nor does it state convergence for the smooth spectral statistics considered here. Moreover, the replacement of rank-one sensing by GOE sensing is assumed in their derivation.

For the joint Gaussian prediction, we use the state-evolution theorem of Berthier, Montanari, and Nguyen for nonseparable Lipschitz AMP \cite[Theorem~1]{BerthierRaphael2020Sefa}. Their theorem is stated in vector coordinates. Appendix~\ref{app:gaussian-se-verification} translates it to our matrix normalization, identifies the correlated GOE perturbations and the Onsager coefficient, and checks the three observable test functions. Proposition~\ref{prop:gaussian-se} states the resulting Gaussian prediction, and Theorem~\ref{thm:se-transfer} transfers it to rank-one sensing.

At each fixed iteration, Gaussian state evolution replaces the pre-denoising matrix by the signal plus an effective GOE perturbation. These perturbations are correlated across iterations: the array $C$ below records the overlaps of the estimation errors, while $Q=\Delta+C$ records the total effective noise entering the denoisers.

Fix deterministic sequences $S^\star$ and $\widehat S^0$ and set $\alpha_d\bydef n/p$. If the signal or initialization is random, the statements below hold conditionally on $(S^\star,\widehat S^0)$ whenever the limits in Assumption~\ref{ass:se-regularity} exist almost surely; the coefficients specified below are then $\CB_0$-measurable, as permitted by Theorem~\ref{thm:final-universality}. The noise matrices and iterates below form triangular arrays indexed by $d$; to keep the notation light, we suppress this dependence. Given a positive semidefinite array $(Q_{s,r})_{0\leq s,r\leq t}$, let $W^0,\ldots,W^t\in\Sym$ be jointly Gaussian symmetric matrices whose upper-triangular entries are independent across index pairs $(j,k)$ and satisfy
\begin{equation}\label{eq:se-noise-covariance}
    \E\big[W^s_{jk}W^r_{jk}\big]
    =\frac{\big(1+\mathbf 1_{\{j=k\}}\big)\,Q_{s,r}}{2(d+1)\alpha_d},
    \qquad 1\,\leq\, j\,\leq\, k\,\leq\, d,\quad 0\,\leq\, s,r\,\leq\, t.
\end{equation}
Each $W^r$ is a GOE matrix: in the normalization \eqref{eq:goe-entry-definition}, $W^r$ has the law of $\sqrt{Q_{r,r}\,d/(2(d+1)\alpha_d)}\;\widetilde Z_1$. The state-evolution iterates are
\begin{equation}\label{eq:se-recursion-matrix}
    \overline S^0\bydef\widehat S^0,
    \qquad
    \overline R^{t}\bydef S^\star+W^{t},
    \qquad
    \overline S^{t+1}\bydef\sigma_t\big(\overline R^{t}\big),
    \qquad 0\leq t<T,
\end{equation}
and the state-evolution arrays are
\begin{equation}\label{eq:se-arrays}
    C_{s,t}\bydef\lim_{d\to\infty}\frac2d\,\E\Tr\big((\overline S^s-S^\star)(\overline S^t-S^\star)\big),
    \qquad
    Q_{s,t}\bydef\Delta+C_{s,t},
    \qquad 0\leq s,t\leq T,
\end{equation}
where the expectation is over the noise matrices $W$. The definition is recursive in $t$. First, $C_{0,0}$ is determined by $S^\star$ and $\widehat S^0$, and $Q_{0,0}=\Delta+C_{0,0}$. Given the entries of $Q$ through time $t$, equation \eqref{eq:se-noise-covariance} generates $W^0,\ldots,W^t$; applying the denoisers in \eqref{eq:se-recursion-matrix} then determines the next row and column of $C$, and hence of $Q$. Each finite principal subarray of $C$ is a limit of Gram matrices and is therefore positive semidefinite. Consequently, $Q=C+\Delta\mathbf 1\mathbf 1^\top$ is positive semidefinite as well, so the required jointly Gaussian matrices exist at every step. The factor $2/d=(d+1)/p$ in \eqref{eq:se-arrays} converts the trace normalization $d^{-1}\Tr$ used here into the normalization $n^{-1}\langle\cdot,\cdot\rangle$ of the vectorized recursion in \cite{BerthierRaphael2020Sefa}; see Appendix~\ref{app:gaussian-se-verification}.

The recursion requires limits of the error overlaps. We assume these limits exist for every Gaussian covariance pair, so that the assumption can be checked before the state-evolution covariances are known.
\begin{assumption}[State-evolution regularity]\label{ass:se-regularity}
    Suppose that $\alpha_d=n/p\to\alpha\in(0,\infty)$, and define the error maps
    \begin{equation*}
        \mathcal E_{0,d}(W)\bydef\widehat S^0-S^\star,
        \qquad
        \mathcal E_{t,d}(W)\bydef\sigma_{t-1}(S^\star+W)-S^\star,
        \quad 1\leq t\leq T.
    \end{equation*}
    For every $0\leq s,t\leq T$ and every positive semidefinite $K\in\R^{2\times2}$, let $(W_1,W_2)$ be jointly Gaussian symmetric matrices with upper-triangular entries independent across index pairs and
    \begin{equation*}
        \E[(W_a)_{jk}(W_b)_{jk}]
        =\frac{(1+\mathbf 1_{\{j=k\}})K_{ab}}{2(d+1)\alpha_d},
        \qquad a,b\in\{1,2\}.
    \end{equation*}
    Assume that the limit
    \begin{equation}\label{eq:se-regularity-limit}
        \lim_{d\to\infty}\frac2d\,\E\Tr\big(\mathcal E_{s,d}(W_1)\mathcal E_{t,d}(W_2)\big)
    \end{equation}
    exists and is finite.
\end{assumption}
The convergence $n/p\to\alpha$ is used only here; Theorem~\ref{thm:final-universality} requires no more than $n\leq C_0p$.

We now specify the coefficients that place this matrix recursion in the standard AMP normalization. The back-projection coefficient $c_t$ fixes the scale, while the memory coefficient $b_t$ is the Onsager term that cancels the leading self-interaction created by reusing the same sensing matrix. To define the latter, let $\CD_t(M)$ be the divergence of the spectral map $\sigma_t$ at $M$, namely the trace of its Fr\'echet derivative $D\sigma_t(M)$ as a linear operator on $\Sym$. By Lemma~\ref{lem:dk-bounds} with $m=1$, if $M=U\diag(\lambda_1,\ldots,\lambda_d)U^\top$, then
\begin{equation}\label{eq:se-divergence-formula}
    \CD_t(M)\,=\,\sum_{1\,\leq\, j\,\leq\, k\,\leq\, d}\sigma_t^{[1]}(\lambda_j,\lambda_k),
    \qquad
    \sigma_t^{[1]}(x,y)\bydef
    \begin{cases}
        \dfrac{\sigma_t(x)-\sigma_t(y)}{x-y}, & x\neq y,\\[1ex]
        \sigma_t'(x), & x=y.
    \end{cases}
\end{equation}
We run the recursion \eqref{eq:AMP-recursion} with
\begin{equation}\label{eq:se-coefficients}
    c_t\bydef\frac{\gamma_t\,p}{n(d+1)}=\frac{\gamma_t}{\alpha_d(d+1)},
    \qquad
    b_0\bydef0,
    \qquad
    b_t\bydef\frac{\gamma_{t-1}}{n}\,\E\big[\CD_{t-1}(\overline R^{t-1})\big],
    \quad 1\leq t<T.
\end{equation}
Since $\abs{\CD_t(M)}\leq p\norm{\sigma_t'}_\infty$ and $n\asymp p$, these coefficients satisfy the bounds in Assumption~\ref{ass:basic}. With $\mathfrak r^t\bydef\gamma_tr^t_{\sfe}$, the recursion \eqref{eq:AMP-recursion} reads $\mathfrak r^t=y_{\sfe}-\CA_{\sfe}(\widehat S^t_{\sfe})+(b_t/\gamma_{t-1})\mathfrak r^{t-1}$ and $R^t_{\sfe}=\widehat S^t_{\sfe}+(c_t/\gamma_t)\CA_{\sfe}^*(\mathfrak r^t)$. Under \eqref{eq:se-coefficients}, the ratios $b_t/\gamma_{t-1}$ and $c_t/\gamma_t$ are independent of $\gamma$, so $\gamma_t$ changes only the normalization assigned to the residual, not the matrix iterates. A different back-projection ratio $c_t/\gamma_t$, however, changes the scale of the matrix presented to the fixed denoiser $\sigma_t$. Thus $c_t/\gamma_t$ fixes a normalization convention, whereas the Onsager ratio $b_t/\gamma_{t-1}$ is prescribed by the normalized expected divergence at the state-evolution input.

For these coefficients, the state-evolution predictions for the three observables are
\begin{equation}\label{eq:se-predictions}
    \phi^{\sig}_{d,t}\bydef\frac1d\,\E\Tr\big(\overline S^tS^\star\big),
    \qquad
    \phi^{\cross}_{d,s,t}\bydef\frac1d\,\E\Tr\big(\overline S^s\overline S^t\big),
    \qquad
    \phi^{\spec}_{d,t,\mu}\bydef\frac1d\,\E\Tr\big(\mu(\overline S^t)\big).
\end{equation}
\begin{proposition}[Gaussian state evolution, \cite{BerthierRaphael2020Sefa}]\label{prop:gaussian-se}
    Let $S^\star$ and $\widehat S^0$ be deterministic sequences satisfying Assumptions \ref{ass:signal-init} and \ref{ass:se-regularity}, let Assumption~\ref{ass:basic} hold, and run the Gaussian recursion \eqref{eq:AMP-recursion} with the coefficients \eqref{eq:se-coefficients}. Then, for every $0\leq s,t\leq T$ and every $\mu\in\SC^3(\R)$,
    \begin{equation}\label{eq:gaussian-se-observables}
        \Phi^{\sig}_{\mathrm G,t}-\phi^{\sig}_{d,t}\xrightarrow{\ \P\ }0,
        \qquad
        \Phi^{\cross}_{\mathrm G,s,t}-\phi^{\cross}_{d,s,t}\xrightarrow{\ \P\ }0,
        \qquad
        \Phi^{\spec}_{\mathrm G,t,\mu}-\phi^{\spec}_{d,t,\mu}\xrightarrow{\ \P\ }0.
    \end{equation}
\end{proposition}
Appendix~\ref{app:gaussian-se-verification} derives Proposition~\ref{prop:gaussian-se} from \cite[Theorem~1]{BerthierRaphael2020Sefa}. Only the Lipschitz continuity of the $\sigma_t$ enters this reduction; the $C^4$ regularity in Assumption~\ref{ass:basic} is needed for universality, not for Gaussian state evolution.

\begin{theorem}[State evolution for rank-one sensing]\label{thm:se-transfer}
    Under the assumptions of Proposition~\ref{prop:gaussian-se}, the convergences \eqref{eq:gaussian-se-observables} hold with $\mathrm G$ replaced by $\ro$. In particular, for both ensembles and every $0\leq t\leq T$,
    \begin{equation}\label{eq:se-mse-limit}
        \frac1d\big\|\widehat S^t_{\sfe}-S^\star\big\|_\Frob^2
        \xrightarrow{\ \P\ }\frac12\,C_{t,t},
        \qquad \sfe\in\{\ro,\mathrm G\}.
    \end{equation}
\end{theorem}
\begin{proof}
    The first claim follows from Proposition~\ref{prop:gaussian-se} and \eqref{eq:transfer-principle} with $\phi_d$ equal to the corresponding prediction in \eqref{eq:se-predictions}. For \eqref{eq:se-mse-limit}, expand
    \begin{equation*}
        \frac1d\big\|\widehat S^t_{\sfe}-S^\star\big\|_\Frob^2
        =\Phi^{\cross}_{\sfe,t,t}-2\Phi^{\sig}_{\sfe,t}+\frac1d\|S^\star\|_\Frob^2,
        \qquad
        \phi^{\cross}_{d,t,t}-2\phi^{\sig}_{d,t}+\frac1d\|S^\star\|_\Frob^2
        =\frac1d\,\E\big\|\overline S^t-S^\star\big\|_\Frob^2,
    \end{equation*}
    and note that the right-hand side of the second identity converges to $\frac12C_{t,t}$ by \eqref{eq:se-arrays}.
\end{proof}

Under the conditions above, the joint matrix state evolution therefore gives a rigorous prediction for the mean-squared error of AMP under rank-one sensing.

\begin{remark}[Interpretation]\label{rem:se-interpretation}
    Write $\tau_t^2\bydef Q_{t,t}=\Delta+C_{t,t}$. By \eqref{eq:se-noise-covariance}, the denoiser at step $t$ sees the signal corrupted by a GOE matrix, $\overline R^t=S^\star+W^t$, with $W^t$ distributed as $\sqrt{\tau_t^2\,d/(2(d+1)\alpha_d)}\,\widetilde Z_1$. Thus, in the normalization \eqref{eq:goe-entry-definition}, $W^t$ is asymptotically $\sqrt{\tau_t^2/(2\alpha)}\,\widetilde Z_1$. In the Bayes-optimal algorithm of \cite{MaillardAntoine2024Bloa}, $\sigma_t$ is the rotationally invariant estimator of \cite{BunAllezBouchaudPotters2016} for denoising $S^\star$ at this GOE noise level, and \eqref{eq:se-arrays} becomes a recursion for the matrix-denoising mean-squared error. Theorem~\ref{thm:se-transfer} does not cover that denoiser, which is neither bounded nor $C^4$; Section~\ref{subsec:numerics-spectra} examines it empirically.
\end{remark}

\begin{remark}[A sufficient condition]\label{rem:se-regularity-sufficient}
    Assumption~\ref{ass:se-regularity} holds, for instance, whenever $\sup_d\|S^\star\|_\op<\infty$, the empirical spectral distribution of $S^\star$ converges weakly, and $\widehat S^0$ is a fixed continuous spectral function of $S^\star$ (for example, $\widehat S^0=0$). Represent a correlated Gaussian pair $(W_1,W_2)$ as linear combinations of independent GOE matrices. These GOE components and $S^\star$ are asymptotically free \cite[Section~5.4]{AndersonGuionnetZeitouni2010}, which gives convergence for normalized traces of polynomials. With probability at least $1-e^{-cd}$, the spectra of $S^\star+W_1$ and $S^\star+W_2$ lie in a fixed compact interval, on which the spectral maps can be approximated uniformly by polynomials; on the complementary event, boundedness of the spectral maps makes the contribution negligible. The same argument applies when $\widehat S^0$ is independent of $S^\star$, rotationally invariant, has uniformly bounded operator norm, and has an almost surely convergent empirical spectral distribution. The additional operator-norm bound on $S^\star$ is used only for these convenient sufficient conditions. The result itself assumes \eqref{eq:se-regularity-limit} directly.
\end{remark}

\section{Numerical Experiments}\label{sec:numerics}

We compare the rank-one and Gaussian AMP trajectories with the state-evolution predictions of Proposition~\ref{prop:gaussian-se} and Theorem~\ref{thm:se-transfer}, and measure how their discrepancies decrease as $n$ and $d$ grow at a fixed sampling ratio. We examine the mean-squared error, cross-time overlaps, and empirical spectra. The final experiment uses a nonlinear output channel to test what happens beyond the linear observation model of Theorem~\ref{thm:final-universality}.

Throughout this section, $p=d(d+1)/2$ and $n=\lfloor\alpha p\rfloor$ for a fixed sampling ratio $\alpha$, as in the state-evolution formulation of Section~\ref{subsec:se-transfer}. Unless stated otherwise, we run the linear-channel AMP recursion \eqref{eq:AMP-recursion} from $\widehat S^0=0$, take $\gamma_t=1$, and use the coefficients prescribed by \eqref{eq:se-coefficients}. Thus, $c_t=d/(2n)$ in the present normalization, while $b_t$ is obtained from the expected divergence at the state-evolution input. We compute by Monte Carlo the finite-dimensional expectations entering the state-evolution arrays $C$ and $Q$ in \eqref{eq:se-arrays}. In each paired trial, the two recursions use the same signal, initialization, and observation noise from \eqref{eq:observation-model} but independently sampled rank-one and GOE sensing matrices. Unless stated otherwise, confidence intervals are $95\%$ percentile-bootstrap intervals over these paired trials.

The Gaussian sensing matrix has $pn=\Theta(d^4)$ entries, since both $p$ and $n$ are of order $d^2$. We implement its forward and adjoint actions using the deferred-sampling idea of Householder Dice \cite{Lu2021HouseholderDice}. Each new matrix--vector product is sampled conditional on those already computed. The resulting products have the same joint law as repeated queries to a single Gaussian sensing matrix, while avoiding storage of all $pn$ entries.

\subsection{State Evolution}\label{subsec:numerics-state-evolution}

We begin with a bounded-spectrum choice of the signal $S^\star$ in \eqref{eq:observation-model}: $3/4$ of its eigenvalues are $+1$ and the rest are $-1$, up to rounding. For this signal, we use the spectral denoiser
\begin{equation}\label{eq:numerical-tanh-denoiser}
    \sigma_t(x)=\tanh(x/v_t+h).
\end{equation}
In the state-evolution recursion \eqref{eq:se-recursion-matrix}, the denoiser at iteration $t$ receives $S^\star+W^t$. In the large-$d$ limit, \eqref{eq:se-noise-covariance} gives $W^t$ the law of $\sqrt{v_t}\,\widetilde Z_1$ for the GOE matrix $\widetilde Z_1$ normalized in \eqref{eq:goe-entry-definition}, with
\begin{equation*}
    v_t=\frac{Q_{t,t}}{2\alpha}.
\end{equation*}
Thus, $v_t$ is the effective GOE noise variance at iteration $t$, and we use it to set the transition scale of the spectral map. Writing $q$ for the proportion of $+1$ eigenvalues after rounding, we take $h=\frac12\log(q/(1-q))$, a positive offset that reflects the asymmetry of the signal spectrum. This choice is not meant to be the posterior mean for the matrix channel: the fixed signal $S^\star$ and the GOE perturbation $W^t$ need not have the same eigenvectors, so the eigenvalues presented to the denoiser are not scalar Gaussian observations of the eigenvalues of $S^\star$. Rather, \eqref{eq:numerical-tanh-denoiser} provides a convenient deterministic, bounded, $C^\infty$ spectral map whose scale follows state evolution. It satisfies Assumption~\ref{ass:basic} and is therefore covered by Theorem~\ref{thm:final-universality}.

In the first experiment, we verify numerically that the rank-one and Gaussian recursions follow the common state-evolution prediction in Figure~\ref{fig:numerical-state-evolution}. For this experiment, $d=500$, $\alpha=1$, and the observation-noise variance in \eqref{eq:observation-model} is $\Delta=0.5$. We run 50 paired trials for eight iterations and approximate the state-evolution expectations with 400 Gaussian draws. The normalized mean-squared error in \eqref{eq:se-mse-limit} tests the diagonal entries $C_{t,t}$ of the array $C$ defined in \eqref{eq:se-arrays}. To test the joint prediction across different iterations, we also record
\begin{equation*}
    \widehat C_{\sfe,s,t}\bydef\frac1d\Tr\big((\widehat S^s_{\sfe}-S^\star)
    (\widehat S^t_{\sfe}-S^\star)\big),
    \qquad \sfe\in\{\ro,\mathrm G\},
\end{equation*}
whose state-evolution prediction is $C_{s,t}/2$. These cross-time overlaps measure how strongly the errors at two different iterations remain correlated.

\begin{figure}
    \centering
    \includegraphics[width=\textwidth]{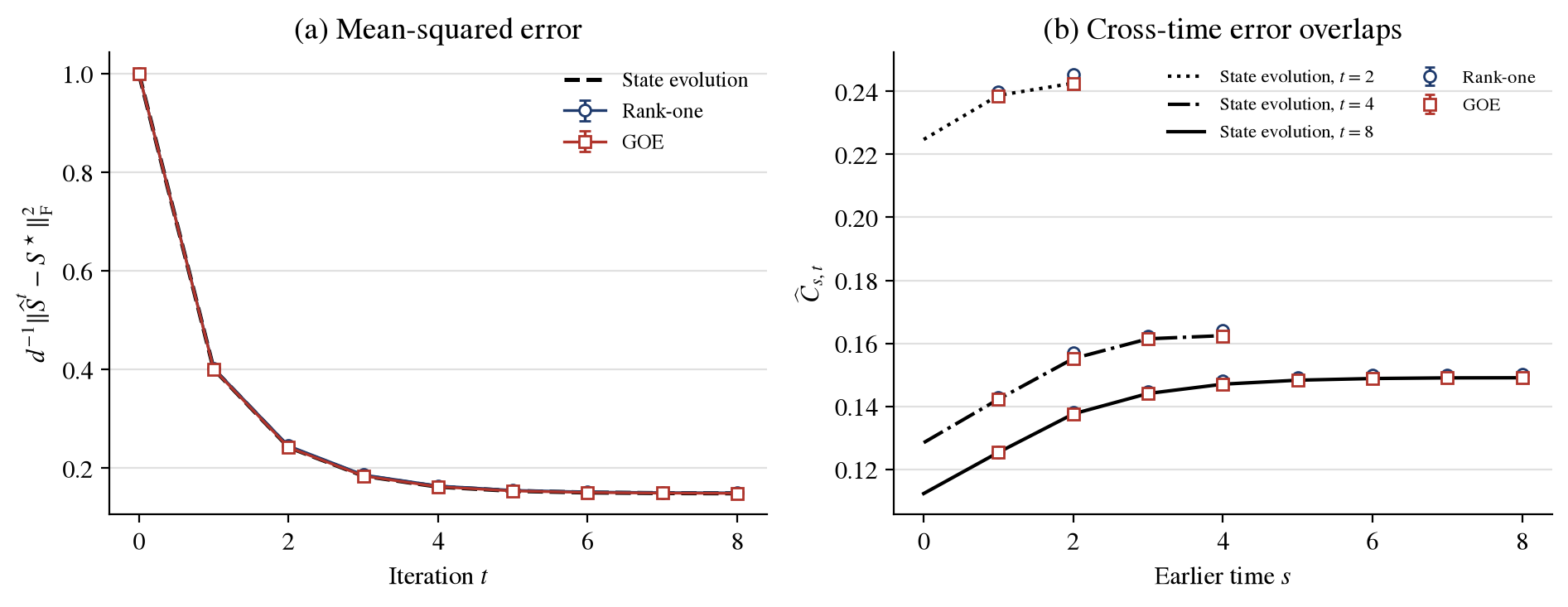}
    \caption{State evolution for the bounded-spectrum signal at $d=500$, $\alpha=1$, and $\Delta=0.5$. Symbols are means over 50 paired trials, the bars are $95\%$ bootstrap confidence intervals, and the black curves are the finite-$d$ state-evolution predictions computed from 400 draws. Left: normalized mean-squared error. Right: cross-time error overlaps for $t\in\{2,4,8\}$. The rank-one and GOE recursions follow not only the same marginal error curve but also the joint correlations between different iterations.}
    \label{fig:numerical-state-evolution}
\end{figure}

The three curves in the left panel of Figure~\ref{fig:numerical-state-evolution} are nearly indistinguishable. For scale, $1/d=2.0\times10^{-3}$ in this experiment. The right panel shows that the agreement is not limited to marginal errors: the empirical cross-time overlaps also follow the off-diagonal entries $C_{s,t}$ defined in \eqref{eq:se-arrays}, as predicted by Proposition~\ref{prop:gaussian-se} for GOE sensing and by Theorem~\ref{thm:se-transfer} for rank-one sensing.

\FloatBarrier
\subsection{Spectral Distributions}\label{subsec:numerics-spectra}

Figure~\ref{fig:spectra} compares the empirical spectral distributions for a Wishart signal. We draw $S^\star=WW^\top/d$ and an independent initialization $\widehat S^0=W_0W_0^\top/d$, where $W,W_0\in\R^{d\times d}$ have independent standard Gaussian entries. We run 20 paired trials with $d=500$, $\Delta=1$, and $\alpha=0.5$ for eight iterations. We compare the plain spectral map $\sigma_t(x)=\tanh(x)$ with the Wishart RIE denoiser of \cite{MaillardAntoine2024Bloa},
\begin{equation*}
    \sigma_t^{\mathrm{RIE}}(x)=x-2v_t\Re G_{v_t}(x+i0).
\end{equation*}
Here $\rho_v$ is the density of the free convolution of the aspect-ratio-one Marchenko--Pastur law and a semicircle law of variance $v$, and $G_v(z)=\int (z-u)^{-1}\rho_v(u)\,du$. We fix the coefficients in advance by numerical integration:
\begin{equation*}
    v_t=\frac{\Delta+2m_t}{2\alpha},\qquad b_0=0,\qquad
    b_{t+1}=\frac{2}{\alpha}\int\sigma_t(x)\Re G_{v_t}(x+i0)\rho_{v_t}(x)\,dx,
\end{equation*}
where $m_t$ is the limiting normalized mean-squared error of the Gaussian state-evolution iterate. The RIE denoiser lies outside the bounded $C_b^4$ class in our theorem; the agreement in its spectral distributions provides numerical evidence beyond that class.
\begin{figure}
    \centering
    \includegraphics[width=\textwidth]{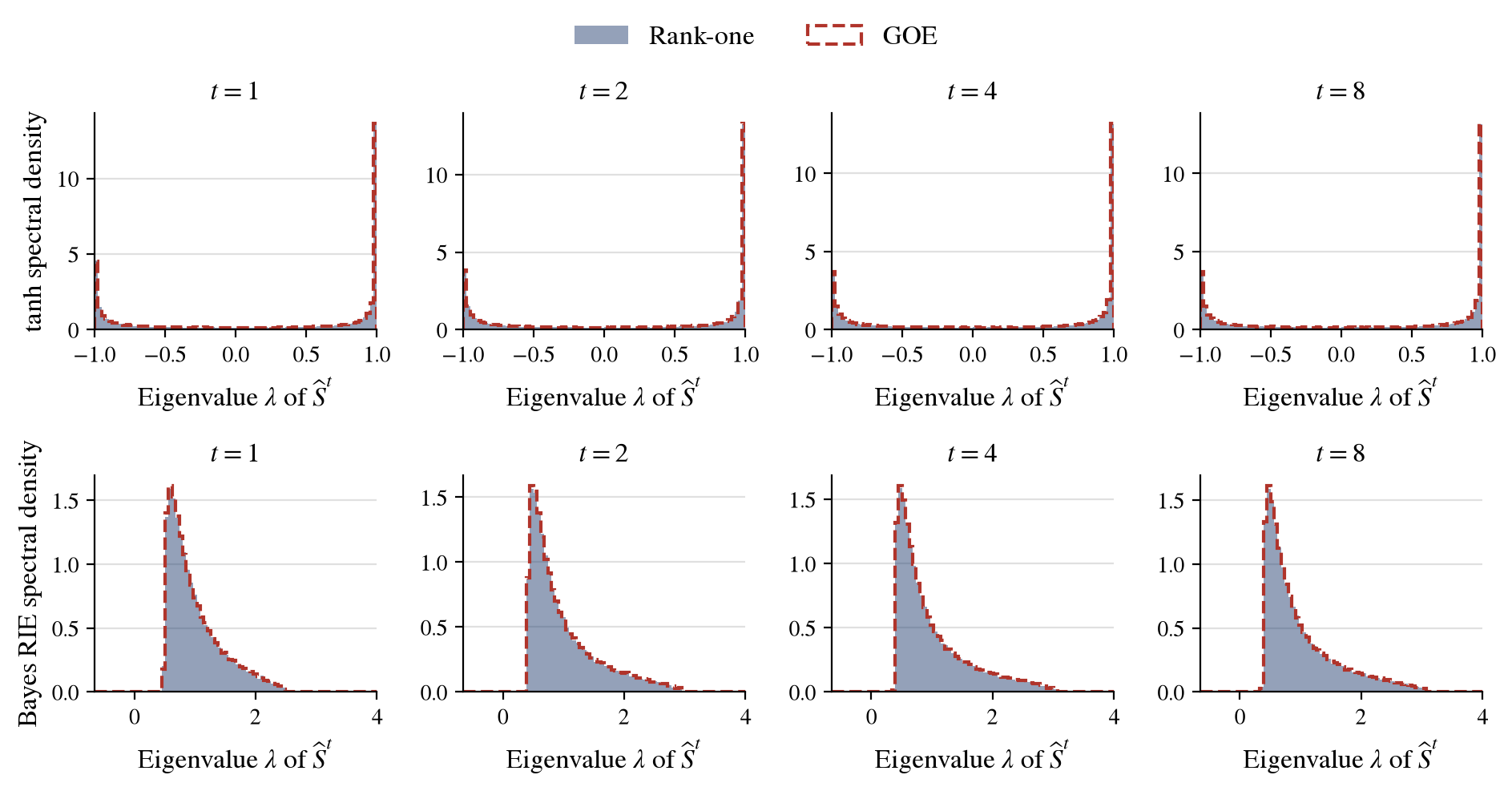}
    \caption{Empirical spectral distributions of the AMP iterates at $t\in\{1,2,4,8\}$ for independent square Wishart signal and initialization, with $d=500$, $\Delta=1$, and $\alpha=0.5$. Histograms pool the eigenvalues from 20 paired trials. Top: the plain $\tanh$ denoiser. Bottom: the Wishart RIE denoiser.}
    \label{fig:spectra}
\end{figure}

\FloatBarrier
\subsection{Convergence Rates}\label{subsec:numerics-rates}

We next measure how the discrepancy between the two ensembles decreases with dimension. For any one of the observables $\Phi$ in \eqref{eq:observable-notation}, set
\begin{equation}\label{eq:numerical-rms-discrepancy}
    \mathcal R_\Phi(d)\bydef
    \big\|\Phi_{\ro}-\Phi_{\mathrm G}\big\|_{L^2}
\end{equation}
\nopagebreak
where the $L^2$ norm is taken under the coupling of Theorem~\ref{thm:final-universality}, and estimate this quantity by the root mean square over the paired trials. Thus, $\mathcal R_\Phi(d)$ reflects both a difference between the two ensemble means and the fluctuations of the paired difference around its mean. Theorem~\ref{thm:final-universality} guarantees $\mathcal R_\Phi(d)=\CO(d^{-1/2})$.

We use the dimension grid
\begin{equation*}
    d\in\{100,150,200,300,400,500,800,1200,1600,2000\},
\end{equation*}
with 200 paired trials for each signal at every dimension, $\alpha=0.25$, and $\Delta=0.5$. We report the first iteration, where the source of the discrepancy is easiest to identify. At this iteration, the prescribed coefficients $b_0=0$ and $v_0=(\Delta+2)d^2/(4n)$ are known exactly. No Monte Carlo approximation of the coefficients is needed for this comparison.

We compare the bounded-spectrum signal above with a signal containing a spectral spike and a bounded bulk. Consider the family
\begin{equation}\label{eq:numerical-signal-interpolation}
    S^\star(a)
    =a\sqrt d\,e_1e_1^\top
    +\sqrt{\frac{(1-a^2)d}{d-1}}\,
      \diag(0,\eta_2,\ldots,\eta_d),
    \qquad 0\leq a\leq1,
\end{equation}
where the signs $\eta_j\in\{-1,+1\}$ are as evenly balanced as possible. Every member has normalized signal energy $d^{-1}\|S^\star(a)\|_\Frob^2=1$. For the spiked signal, we take $a=0.75$ and use the denoiser \eqref{eq:numerical-tanh-denoiser} with $h=0$. Both signals satisfy Assumption~\ref{ass:signal-init}.

\begin{figure}
    \centering
    \includegraphics[width=\textwidth]{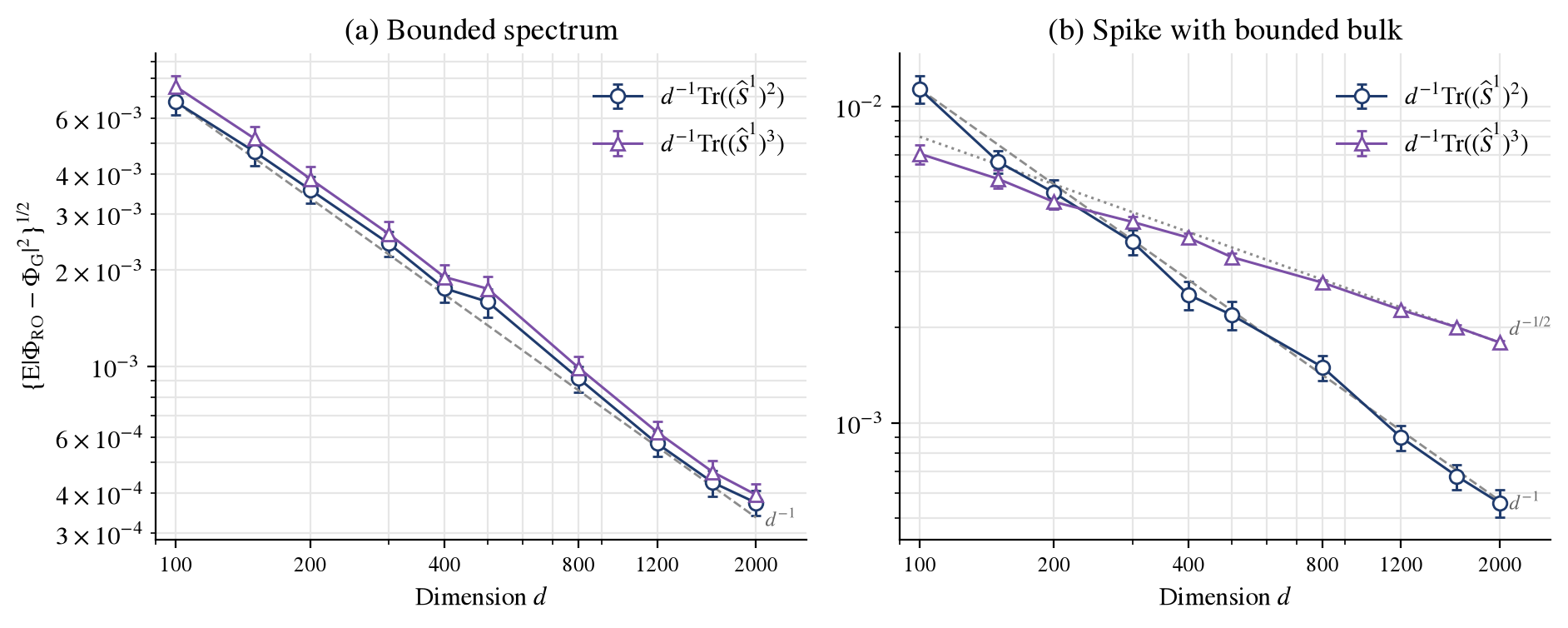}
    \caption{Empirical $L^2$ discrepancy \eqref{eq:numerical-rms-discrepancy} at $t=1$, with $\alpha=0.25$ and $\Delta=0.5$. Left: the bounded-spectrum signal. Right: the signal \eqref{eq:numerical-signal-interpolation} with $a=0.75$. The bars are $95\%$ bootstrap confidence intervals for the root mean square over 200 paired trials at each dimension. The dashed reference lines have slope $-1$; the dotted line in the right panel has slope $-1/2$. The cubic statistic decays more slowly in the spiked setting.}
    \label{fig:numerical-l2-rate}
\end{figure}

Figure~\ref{fig:numerical-l2-rate} plots the discrepancies of the quadratic and cubic spectral statistics $d^{-1}\Tr((\widehat S^1_{\sfe})^k)$, $k\in\{2,3\}$, which are instances of the spectral observable in \eqref{eq:observable-notation}. Table~\ref{tab:numerical-effective-exponents} gives the fitted exponents for these statistics and the signal overlap. Each exponent is the slope of a log--log fit to the RMS discrepancy \eqref{eq:numerical-rms-discrepancy}.

\begin{table}[H]
    \centering
    \small
    \begin{tabular}{llc}
        \hline
        Signal & Observable at $t=1$ & RMS exponent\\
        \hline
        bounded & signal overlap & $-1.01$\\
        bounded & $d^{-1}\Tr((\widehat S^1)^2)$ & $-0.99$\\
        bounded & $d^{-1}\Tr((\widehat S^1)^3)$ & $-1.00$\\
        spike + bulk & signal overlap & $-1.10$\\
        spike + bulk & $d^{-1}\Tr((\widehat S^1)^2)$ & $-0.99$\\
        spike + bulk & $d^{-1}\Tr((\widehat S^1)^3)$ & $-0.45$\\
        \hline
    \end{tabular}
    \caption{Effective log--log slopes of the RMS discrepancy over $100\leq d\leq2000$, with 200 paired trials at every dimension.}
    \label{tab:numerical-effective-exponents}
\end{table}
For the bounded-spectrum signal, all three RMS discrepancies decay approximately as $d^{-1}$ over the tested range. For the spiked signal, the cubic statistic has fitted slopes of $-0.45$ over the full grid and $-0.46$ over $400\leq d\leq2000$. The estimated mean difference, multiplied by $\sqrt d$, is also close to $0.08$ at the four largest dimensions. These observations suggest that the $d^{-1/2}$ decay rate for the universality theorem is tight for the class of observables considered. To help explain this observation, Appendix~\ref{app:numerical-spike} examines how the skewness of rank-one measurements affects the bulk spectrum before denoising.

\FloatBarrier
\subsection{A Nonlinear Output Channel}\label{subsec:numerics-nonlinear}

Our theorem assumes linear observations. Here we compare the rank-one and Gaussian dynamics under a nonlinear output channel, using a bounded-spectrum target and a spiked target.

\begin{figure}
    \centering
    \includegraphics[width=\textwidth]{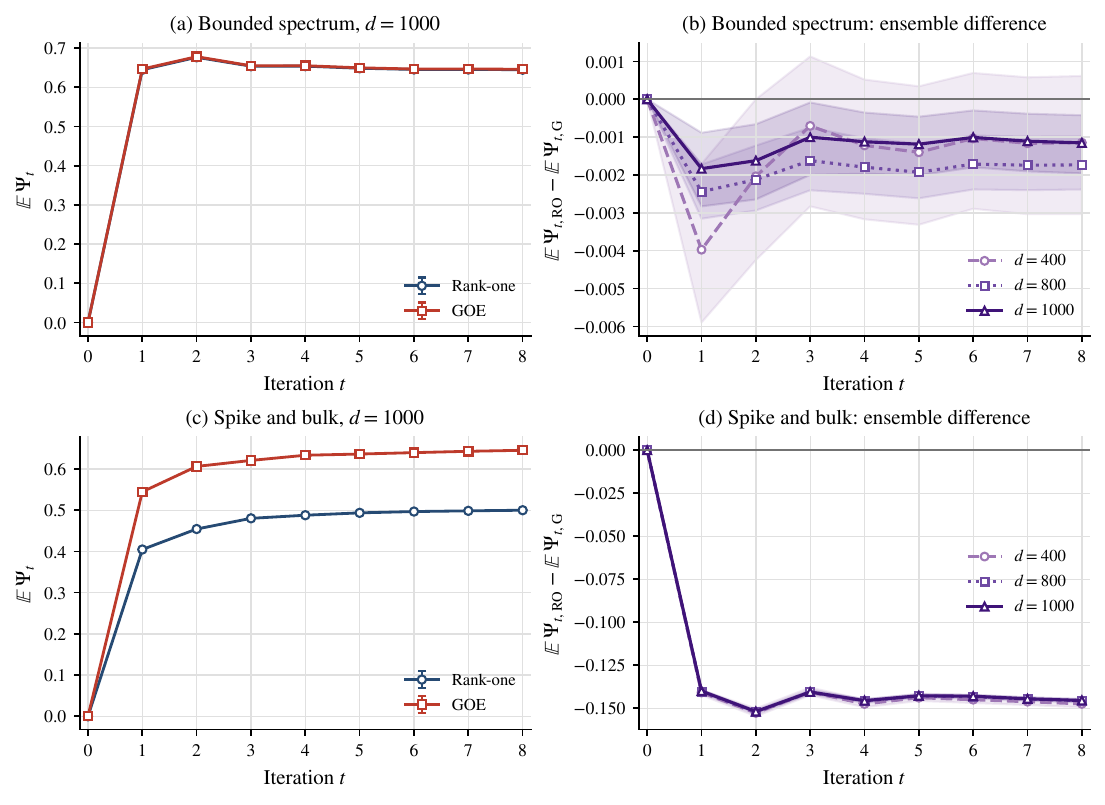}
    \caption{Nonlinear observations with bounded and spiked target spectra. Top: $a=0$ in \eqref{eq:numerical-signal-interpolation}; bottom: $a=0.95$. Both use \eqref{eq:numerical-nonlinear-channel} with $\theta=1.75$, $\tau=0.2$, and the same recursion \eqref{eq:numerical-nonlinear-gamp} with $\alpha=1$. Left: mean trajectories at $d=1000$; bars are $95\%$ bootstrap confidence intervals over ten paired trials. Right: differences between ensemble means at $d\in\{400,800,1000\}$, with $95\%$ bootstrap confidence bands and ten paired trials at every dimension. The right panels use different vertical scales.}
    \label{fig:numerical-nonlinear-trajectory}
\end{figure}

We compare two signals from the family $S^\star(a)$ in \eqref{eq:numerical-signal-interpolation}. The choice $a=0$ gives a balanced signal with bounded eigenvalues. For $a=0.95$, the spike carries $90.25\%$ of the normalized Frobenius energy, while the remaining $9.75\%$ is spread over a bounded bulk. For $z_{\sfe,i}=\Tr(Z_{\sfe,i}S^\star(a))$, the observation is the noisy ternary quantization
\begin{equation}\label{eq:numerical-nonlinear-channel}
    y_{\sfe,i}=q_\theta(z_{\sfe,i}+\tau\epsilon_i),
    \qquad
    q_\theta(u)=
    \begin{cases}
        1, & u>\theta,\\
        0, & |u|\leq\theta,\\
        -1, & u<-\theta,
    \end{cases}
\end{equation}
where $\epsilon_i\sim\CN(0,1)$. We take $\theta=1.75$ and $\tau=0.2$. This threshold probes the tails of the latent response, where the skewness induced by a spike has a visible effect.

For this experiment we use the output-function form of GAMP in \cite[Algorithm~1]{MaillardAntoine2024Bloa}, but keep its coefficients fixed and identical in the two ensembles. This removes adaptive tuning as a possible source of the difference. Starting from $\widehat S^0_{\sfe}=0$ and $g^{-1}_{\sfe}=0$, we iterate
\begin{equation}\label{eq:numerical-nonlinear-gamp}
    \begin{aligned}
        \omega^t_{\sfe}
        &=\CA_{\sfe}(\widehat S^t_{\sfe})-Vg^{t-1}_{\sfe},\\
        g^t_{\sfe}
        &=g_{\mathrm{out}}(y_{\sfe},\omega^t_{\sfe},V),\\
        R^t_{\sfe}
        &=\widehat S^t_{\sfe}+\frac{1}{dA_0}\CA_{\sfe}^*(g^t_{\sfe}),\\
        \widehat S^{t+1}_{\sfe}
        &=\tanh(R^t_{\sfe}),
    \end{aligned}
    \qquad \sfe\in\{\ro,\mathrm G\}.
\end{equation}
Here $g_{\mathrm{out}}(y,\omega,V)$ is the derivative with respect to $\omega$ of the log-likelihood obtained by passing a $\CN(\omega,V)$ variable through \eqref{eq:numerical-nonlinear-channel}, applied coordinatewise. We set $V=2$ and hold
\begin{equation*}
    A_0=\frac{2n}{d^2}\,
    \E\big[g_{\mathrm{out}}(Y,0,2)^2\big]
\end{equation*}
fixed at every iteration, with the expectation evaluated under the corresponding Gaussian scalar channel. The same deterministic $A_0$ and $V$ are used for rank-one and GOE sensing. The final line of \eqref{eq:numerical-nonlinear-gamp} is the ordinary hyperbolic tangent applied spectrally. It is bounded and smooth, so any failure of universality here cannot be attributed to the unbounded or data-dependent denoisers discussed in Section~\ref{sec:discussion}.

We follow the normalized squared norm
\begin{equation*}
    \Psi_{\sfe,t}\bydef\frac1d\Tr\big((\widehat S^t_{\sfe})^2\big)
    =\Phi^{\cross}_{\sfe,t,t},
\end{equation*}
which is one of the iterate observables in \eqref{eq:observable-notation}. For each signal, we run eight iterations with $\alpha=1$ and ten paired trials at each $d\in\{400,800,1000\}$. The channel, initialization, coefficients, and denoiser are identical in the two signal settings. Within a trial, the sensing operator and observations are generated once and reused throughout the trajectory.

For the bounded-spectrum signal, the rank-one and GOE trajectories in Figure~\ref{fig:numerical-nonlinear-trajectory} closely agree. At $d=1000$, the difference between the ensemble means stays below $0.002$ in absolute value at every displayed iteration. These observations support our expectation that universality extends to nonlinear output channels for signals with uniformly bounded eigenvalues, although we do not prove this extension here. For the spiked signal, the difference stays near $-0.14$ to $-0.15$ across all three dimensions and all eight iterations. The contrast suggests that universality under nonlinear observations needs additional restrictions on the target spectrum.

This distinction is visible already in the latent response. With bounded eigenvalues, the rank-one response is a sum of many small independent centered quadratic terms and approaches a Gaussian law. A spike contributes the term $a(G^2-1)$, whose non-Gaussian shape persists as $d$ grows. The nonlinear channel can distinguish this response from its covariance-matched Gaussian counterpart before AMP begins.

When universality fails, the Gaussian surrogate may need a correction that preserves these non-Gaussian directions. Wen et al.\ \cite{WenHuLuFanMisiakiewicz2025} develop such a correction for random-feature models in the quadratic scaling regime: their conditional Gaussian equivalent model retains a low-dimensional non-Gaussian component and Gaussianizes the remaining features. We conjecture that an analogous conditional Gaussian comparison applies to AMP with nonlinear observations and spiked signals.

\FloatBarrier
\section{Proof Roadmap and Proof of the Main Theorem}\label{sec:proof-main}

\subsection{Vector Form of the AMP Recursion}
For the leave-one-out analysis, we write the matrix recursion \eqref{eq:AMP-recursion} in vector coordinates.

The rank-one and Gaussian feature matrices used in the proof are
\begin{equation}\label{eq:global-sensing-matrices}
    F\bydef[f_1,\ldots,f_n],
    \qquad
    \widetilde F\bydef[\widetilde f_1,\ldots,\widetilde f_n]
    \in\R^{p\times n}.
\end{equation}
Fix either ensemble and suppress its superscript, writing $F$ for the corresponding member of \eqref{eq:global-sensing-matrices} and $\CA$ for its sensing operator.
By the isometry \eqref{eq:svec-isometry},
\begin{equation}\label{eq:sensing-map-vector-relation}
    \CA(S)=\sqrt{d+1}\,F^\top\svecop(S),
    \qquad
    \svecop\big(\CA^*(r)\big)=\sqrt{d+1}\,Fr.
\end{equation}
Also write $y=\CA(S^\star)+\epsilon$ and
\begin{equation*}
    {\bf s}^*\bydef\svecop(S^\star),
    \qquad
    \widehat{\bf s}^t\bydef\svecop(\widehat S^t),
    \qquad
    G_t(v)\bydef\svecop\big(\sigma_t(\mmat(v))\big).
\end{equation*}
The residual equation in \eqref{eq:AMP-recursion} becomes
\begin{equation}\label{eq:vectorized-AMP-recursion}
    r^t
    =\frac{\sqrt{d+1}\,F^\top({\bf s}^*-\widehat{\bf s}^t)+\epsilon+b_t r^{t-1}}{\gamma_t}.
\end{equation}
Define the normalized back-projected field
\begin{equation}\label{eq:h-def}
    h^t\bydef\frac1d\svecop\big(\CA^*(r^t)\big)
    =\frac{\sqrt{d+1}}d Fr^t,
    \qquad h^{-1}\bydef0,
\end{equation}
and the scaled coefficients
\begin{equation}\label{eq:scaled-coeffs}
    a_t\bydef dc_t,
    \qquad
    \kappa_t\bydef\frac{d+1}{d\gamma_t},
    \qquad
    \nu_t\bydef\frac{\sqrt{d+1}}{d\gamma_t},
    \qquad
    \rho_t\bydef\frac{b_t}{\gamma_t}.
\end{equation}
Assumption~\ref{ass:basic} implies
\begin{equation}\label{eq:scaled-coeffs-bound}
    \sup_{d\,\geq\,1}\max_{0\,\leq\, t\,<\,T}
    \big(|a_t|+|\kappa_t|+|\rho_t|+\sqrt d\,|\nu_t|\big)<\infty.
\end{equation}
Using \eqref{eq:sensing-map-vector-relation}, the full recursion is equivalently
\begin{equation}\label{eq:reduced-AMP-recursion}
    h^t
    =\kappa_tFF^\top({\bf s}^*-\widehat{\bf s}^t)+\nu_tF\epsilon+\rho_t h^{t-1},
    \qquad
    \widehat{\bf s}^{t+1}=G_t(\widehat{\bf s}^t+a_t h^t).
\end{equation}

\subsection{Proof Roadmap}
We compare the conditional expectations given $\CB_0=\sigma(S^\star,\hS^0)$ by replacing one sensing column at a time. A leave-one-out construction makes the column being replaced independent of the background trajectory and reduces the comparison to two estimates. We then use a high-moment Efron--Stein bound to control fluctuations around the conditional expectations. Together, these estimates prove Theorem~\ref{thm:final-universality}.

\paragraph{Lindeberg replacement: from two ensembles to one column.}
The interpolating matrices $F^{(i)}$ defined in \eqref{eq:interpolating-matrix-def} satisfy $F^{(1)} = F$ and $F^{(n+1)} = \widetilde F$, while $F^{(i)}$ and $F^{(i+1)}$ differ only in their $i$th column. Thus, for any observable $\Phi$ of the AMP trajectory,
\begin{equation}
    \E[\Phi(F) \mid \CB_0] - \E[\Phi(\widetilde F) \mid \CB_0]
    = \sum_{i=1}^n \left\{\E[\Phi(F^{(i)}) \mid \CB_0] - \E[\Phi(F^{(i+1)}) \mid \CB_0]\right\}.
\end{equation}
Because $n\leq C_0p=\CO(d^2)$, an error of order $d^{-5/2}$ for each column replacement gives a total error of order at most $d^{-1/2}$. The main task is therefore to compare the effect of inserting the rank-one column $f_i$ with that of inserting the covariance-matched Gaussian column $\widetilde f_i$.

\paragraph{Leave-one-out decoupling.}
The AMP trajectory generated after inserting a candidate column $u \in \{f_i,\widetilde f_i\}$ is a nonlinear function of $u$, so moment matching cannot be applied directly. We remove the $i$th column, run the background dynamics with $F_\smi$, and condition on the leave-one-out sigma-algebra $\CG_i$ from \eqref{eq:loo-sigma-algebra-def}. The background iterates and all response operators are then $\CG_i$-measurable, whereas the candidate column $u$ is independent of $\CG_i$. The response to inserting $u$ is
\begin{equation}
    q^t(u) \bydef \widehat{\bf s}^{t}(u) - \widehat{\bf s}^{t}_\smi.
\end{equation}
Proposition \ref{prop:exact-loo} gives an exact coupled recursion for $q^t(u)$ and the corresponding field perturbation $w^t(u)$. Proposition \ref{prop:delta-moments} shows that $q^t(u)$, $w^t(u)$, and $m^t(u)$ are of order $d^{-1/2}$ in every fixed $L^m$, while the second-order Taylor remainder is of order $d^{-1}$.

\paragraph{The two one-column comparisons.}
Expanding the three observables around the leave-one-out trajectory produces two types of leading terms:
\begin{equation}
    \langle \lambda_i,q^t(u)\rangle,
    \qquad \lambda_i \text{ is $\CG_i$-measurable and } \|\lambda_i\| = \CO(d^{-1/2}),
\end{equation}
and
\begin{equation}
    \langle q^t(u),H_iq^s(u)\rangle,
    \qquad H_i \text{ is $\CG_i$-measurable and } \|H_i\|_\op = \CO(d^{-1}).
\end{equation}
Propositions \ref{prop:linear-comparison-complete} and \ref{prop:remaining-quadratic} show that replacing $u=f_i$ by $u=\widetilde f_i$ changes the expectation of either quantity by at most $\CO(d^{-5/2})$. The signal overlap produces a linear term; the cross-time overlap produces two linear terms and one quadratic term; and a smooth spectral statistic produces linear and quadratic terms together with a cubic remainder of the required size. Theorem \ref{thm:reduction} formalizes this reduction and the resulting Lindeberg telescope.

\paragraph{Linearization and moment matching.}
Unrolling the linearized leave-one-out recursion in Section~\ref{sec:linearized} gives an approximation that is at most quadratic in $u$:
\begin{equation}
    \bar q^t(u) = \sum_{r<t}\CB_{t,r}u\,\beta_{i,0}^r(u),
    \qquad
    \beta_{i,0}^r(u) = \kappa_r\langle u,a_i^r\rangle + \nu_r\epsilon_i,
    \qquad
    a_i^r = {\bf s}^* - \widehat{\bf s}^r_\smi,
\end{equation}
The coefficients $\CB_{t,r}, a^r_i, \nu_r, \epsilon_i, \kappa_r$ are all $\CG_i$-measurable, and $\|q^t(u)-\bar q^t(u)\|_{L^m}=\CO(d^{-1})$. The two candidate columns have the same first two moments. Hence, for every $\CG_i$-measurable vector $b$,
\begin{equation}
    \begin{aligned}
        \E[\langle b,u\rangle\beta_{i,0}^r(u)\mid\CG_i]
        &= \frac{\kappa_r}{p}\langle b,a_i^r\rangle,\\
        \E[\beta_{i,0}^r(u)\beta_{i,0}^s(u)\mid\CG_i]
        &= \frac{\kappa_r\kappa_s}{p}\langle a_i^r,a_i^s\rangle + \nu_r\nu_s\epsilon_i^2.
    \end{aligned}
\end{equation}
The right-hand sides are identical for $u=f_i$ and $u=\widetilde f_i$, so the leading linear and quadratic contributions of $\bar q$ cancel between the two ensembles.

The quadratic comparison in Section \ref{sec:proof-quadratic} then follows by centering the remaining quadratic forms in $u$. A more careful estimate in Lemma \ref{lem:refined-centered-quad} supplies the extra factor needed to reach the $d^{-5/2}$ one-column scale. The linear comparison in Section \ref{sec:proof-linear} is more delicate since the propagated nonlinear remainder is too large under a direct norm bound. A fourth-order Fr\'echet expansion of the spectral denoiser separates quadratic, cubic, and quartic contributions. The quartic term is small directly, the cubic term is controlled through the trilinear decomposition of Lemma \ref{lem:trilinear-decomp}, and the quadratic term requires a second moment-matching reduction.

\paragraph{Concentration.}
The Lindeberg argument compares conditional expectations, not the random observables themselves. Section \ref{sec:concentration} jointly resamples one independent block $(u_i,\epsilon_i)$. The one-block sensitivity estimate in Lemma \ref{lem:one-vector-sensitivity} is of order $d^{-3/2}$ for each normalized observable. The high-moment Efron--Stein inequality of Lemma \ref{lem:moment-efron-stein} therefore gives
\begin{equation}
    \left\|\Phi^{\sf E}-\E[\Phi^{\sf E}\mid\CB_0]\right\|_{L^m}
    \lesssim_m \sqrt{n}\,d^{-3/2}
    \leq C_m d^{-1/2},
    \qquad \sf E \in \{\ro,\mathrm{G}\}.
\end{equation}
Combining these two fluctuation bounds with the conditional-mean Lindeberg comparison gives the $L^m$ conclusion of Theorem \ref{thm:final-universality}. Markov's inequality at arbitrarily high fixed moment order then gives stochastic domination.

\paragraph{Organization of the detailed proof.}
The remaining sections prove the components just described:
\begin{enumerate}[label = (\roman*)]
    \item Section \ref{sec:aux} derives the exact leave-one-out dynamics and performs the Lindeberg reduction to the two one-column comparisons.
    \item Section \ref{sec:linearized} constructs the linearized response and the approximation $\bar q^t(u)$ used by both comparisons.
    \item Sections \ref{sec:proof-quadratic} and \ref{sec:proof-linear} prove the quadratic and linear one-column estimates, respectively.
    \item Section \ref{sec:concentration} proves high-moment concentration of the observables. Appendices \ref{app:analytic} and \ref{app:prob} collect the analytic spectral-calculus estimates and probabilistic input bounds.
\end{enumerate}

\subsection{Proof of the Main Theorem}
We combine the comparison and concentration estimates proved in Sections~\ref{sec:aux}--\ref{sec:concentration}.
\begin{proof}
    Let $\Phi^{\sf E}$ denote any one of the three observables in the theorem, evaluated in ensemble $\sf E \in \{\ro, \mathrm{G}\}$. We first compare their conditional expectations given $\CB_0$. Fix a realization of $(S^\star, \hS^0)$ satisfying Assumption~\ref{ass:signal-init}. Under the corresponding conditional law, the signal and initialization are deterministic, while the sensing columns and the observation noise retain the same laws and independence as before. Moreover, the proofs of Propositions~\ref{prop:remaining-quadratic} and \ref{prop:linear-comparison-complete} use $S^\star$ and $\hS^0$ only through the deterministic bounds in Assumption~\ref{ass:signal-init}. Their constants are therefore uniform over the realization of $\CB_0$.
    We may consequently apply the Lindeberg argument in the proof of Theorem~\ref{thm:reduction} with every expectation replaced by conditional expectation given $\CB_0$. Since $n\leq C_0p=\CO(d^2)$ and each one-column comparison is $\CO(d^{-5/2})$, this gives
    \begin{equation}\label{eq:main-conditional-comparison}
        \abs{\E[\Phi^\ro \mid \CB_0] - \E[\Phi^{\mathrm{G}} \mid \CB_0]} \le C_\mu d^{-1/2}
    \end{equation}
    almost surely. Here and below, the dependence on $\mu$ is omitted for the signal and cross-time observables.
    By Corollary~\ref{cor:concentration-goe-rank-one} and Theorem~\ref{thm:concentration-general-ensemble}, for every integer $m \ge 2$,
    \begin{equation}
        \norm{\Phi^{\sf E} - \E[\Phi^{\sf E} \mid \CB_0]}_{L^m} \,\le\, C_{m,\mu} d^{-1/2}, \qquad \sf E \,\in\, \{\ro, \mathrm{G}\}
    \end{equation}
    Combining these estimates with \eqref{eq:main-conditional-comparison} and the triangle inequality yields
    \begin{equation}
        \begin{aligned}
            \norm{\Phi^\ro - \Phi^{\mathrm{G}}}_{L^m} &\le \norm{\Phi^\ro - \E[\Phi^\ro \mid \CB_0]}_{L^m} + \norm{\Phi^{\mathrm{G}} - \E[\Phi^{\mathrm{G}} \mid \CB_0]}_{L^m}\\
            &+ \norm{\E[\Phi^\ro \mid \CB_0] - \E[\Phi^{\mathrm{G}} \mid \CB_0]}_{L^m}\\
            &\le C_{m,\mu} d^{-1/2}
        \end{aligned}
    \end{equation}
    Applying this to the three choices of $\Phi$ proves \eqref{eq:main-universality-moments}.
    Finally, fix $\varepsilon, D > 0$ and choose an integer $m \ge 2$ such that $m\varepsilon > D$. Markov's inequality gives
    \begin{equation}
        \P\left(\abs{\Phi^\ro - \Phi^{\mathrm{G}}} \,>\, d^{-1/2 + \varepsilon}\right) \,\le\, C_{m,\mu}^m d^{-m\varepsilon} \,\le\, d^{-D}
    \end{equation}
    for all sufficiently large $d$. This proves \eqref{eq:main-universality-stochastic}.
\end{proof}

\section{Lindeberg Reduction and Exact Leave-One-Out Dynamics}\label{sec:aux}
We replace one sensing column at a time and compute the change in each observable by comparing with a trajectory from which that column has been removed. This reduces the comparison of expectations to the linear and quadratic estimates stated below. Appendices~\ref{app:analytic} and \ref{app:prob} contain the analytic and probabilistic bounds used in the argument.

\subsection{Interpolation and Exact Leave-One-Out Recursion}
For $i \in [n + 1]$, define the interpolating model matrices by
\begin{equation}\label{eq:interpolating-matrix-def}
    F^{(i)} \bydef [\widetilde f_1, \ldots, \widetilde f_{i - 1}, f_i, f_{i + 1}, \ldots, f_n]
\end{equation}
Thus, with $F$ and $\widetilde F$ as in \eqref{eq:global-sensing-matrices}, $F^{(1)} = F$ and $F^{(n + 1)} = \widetilde F$. For $i \in [n]$ set the leave-one-out matrix
\begin{equation}\label{eq:loo-matrix-def}
    F_\smi \bydef [\widetilde f_1, \ldots, \widetilde f_{i - 1}, 0, f_{i + 1}, \ldots, f_n]
\end{equation}
For $u \in \R^p$, define the one-column replacement matrix
\begin{equation}\label{eq:single-column-matrix-def}
    F(u) \bydef F_\smi + u e_i^\top,
\end{equation}
We suppress the index $i$ in $F(u)$. Let $(\widehat{\bf s}^t(u), h^t(u))$ be the iterates generated from $F(u)$, and write
\begin{equation}\label{eq:loo-iterates-q-def}
    \widehat{\bf s}^t_\smi \,\bydef\, \widehat{\bf s}^t(F_\smi), \quad \hS^t_\smi \,\bydef\, \mmat(\widehat{\bf s}^t_\smi), \quad q^t(u) \,\bydef\, \widehat{\bf s}^t(u) - \widehat{\bf s}^t_\smi
\end{equation}
Define the leave-one-out $\sigma$-algebra by
\begin{equation}\label{eq:loo-sigma-algebra-def}
    \CG_i \bydef \sigma(F_\smi, S^\star, \epsilon, \hS^0).
\end{equation}
The vectors $q^t(f_i)$ and $q^t(\widetilde f_i)$ record the changes from the leave-one-out iterate when the missing column is reinserted. Write $h^t_\smi=h^t(F_\smi)$ for the corresponding leave-one-out field, and define
\begin{equation}\label{eq:w-difference-def}
    w^t(u) \bydef h^t(u) - h^t_\smi
\end{equation}
By construction, $q^0(u) = 0, w^{-1}(u) = 0$. We also set
\begin{equation}\label{eq:Ki-def}
    K_i \bydef F_\smi F_\smi^\top.
\end{equation}
We derive the following leave-one-out recursion formula for the differences $q^t$ and $w^t$ defined in \eqref{eq:loo-iterates-q-def} and \eqref{eq:w-difference-def}, in terms of the coefficients $\kappa_t,\nu_t,\rho_t$ from \eqref{eq:scaled-coeffs}.
\begin{proposition}\label{prop:exact-loo}
    For $u \in \R^p$ and $t \leq T - 1$,
    \begin{equation}\label{eq:w-exact}
        w^t(u) = -\kappa_t K_i q^t(u) + \rho_t w^{t - 1}(u) + u \beta^t_i(u)
    \end{equation}
    where
    \begin{equation}\label{eq:beta-exact}
        \beta^t_i(u) \bydef \kappa_t \langle u, {\bf s}^* - \widehat{\bf s}^t(u) \rangle + \nu_t \epsilon_i
    \end{equation}
    Furthermore, we have
    \begin{equation}
        q^{t + 1}(u) = G_t(y^t_i + m^t(u)) - G_t(y^t_i)
    \end{equation}
    where
    \begin{equation}\label{eq:loo-y-m-def}
        y^t_i \,\bydef\, \widehat{\bf s}^t_\smi + a_t h^t_\smi \,\in\, \R^p, \quad m^t(u) \,\bydef\, q^t(u) + a_t w^t(u) \,\in\, \R^p
    \end{equation}
    Consequently,
    \begin{equation}\label{eq:q-linearized}
        q^{t + 1}(u) = L^t_i m^t(u) + \varepsilon^t_i(u)
    \end{equation}
    where
    \begin{equation}\label{eq:L-epsilon-def}
        L^t_i \,\bydef\, D G_t(y^t_i):\R^p \,\to\, \R^p, \quad \|\varepsilon^t_i(u)\| \,\leq\, C \|m^t(u)\|^2
    \end{equation}
\end{proposition}
\begin{proof}[Proof of Proposition \ref{prop:exact-loo}]
    Since $F(u) = F_\smi + u e_i^\top$ and $K_i=F_\smi F_\smi^\top$, we have $F(u)F(u)^\top=K_i+uu^\top$. Substituting into \eqref{eq:reduced-AMP-recursion} gives
    \begin{equation}
        h^t(u) = \kappa_t K_i({\bf s}^* - \widehat{\bf s}^t(u)) + \kappa_t uu^\top ({\bf s}^* - \widehat{\bf s}^t(u)) + \nu_t F_\smi \epsilon + \nu_t\epsilon_i u  + \rho_t h^{t-1}(u)
    \end{equation}
    where $h^t(u)$ is defined in \eqref{eq:h-def}. The leave-one-out field satisfies
    \begin{equation}
        h^t_\smi = \kappa_t K_i({\bf s}^* - \widehat{\bf s}^t_\smi) + \nu_t F_\smi \epsilon + \rho_t h^{t - 1}_\smi
    \end{equation}
    Subtracting gives
    \begin{equation}
        w^t(u) = -\kappa_t K_i q^t(u) + \rho_t w^{t - 1}(u) + u\left( \kappa_t \langle u, {\bf s}^* - \widehat{\bf s}^t(u) \rangle + \nu_t \epsilon_i \right)
    \end{equation}
    which proves the first two equations of the proposition. Next, we have
    \begin{equation}
        \widehat{\bf s}^{t +1}(u) \,=\, G_t(\widehat{\bf s}^t(u) + a_th^t(u)), \quad \widehat{\bf s}^{t + 1}_\smi \,=\, G_t(\widehat{\bf s}^t_\smi + a_t h^t_\smi)
    \end{equation}
    Subtracting gives the third equation, with $y_i^t$ and $m^t(u)$ as in \eqref{eq:loo-y-m-def}. By Lemma \ref{lem:spectral-taylor} (spectral calculus) and the definition of $L_i^t$ in \eqref{eq:L-epsilon-def}, we have
    \begin{equation}
        \begin{aligned}
            \|\varepsilon^t_i(u)\| &= \|G_t(y^t_i + m^t(u)) - G_t(y^t_i) - D G_t(y^t_i)[m^t(u)]\| \\
            &\leq \frac{1}{2}\|\sigma''_t\|_\infty \|m^t(u)\|^2
        \end{aligned}
    \end{equation}
\end{proof}
\subsection{Reduction to Linear and Quadratic One-Column Comparisons}
\begin{theorem}\label{thm:reduction}
    Under Assumptions \ref{ass:basic} and \ref{ass:signal-init}, suppose that for every $s, t \leq T$ and every leave-one-out index $i$, the following hold for some $C$ uniformly over $d$.
    \begin{enumerate}[label = (\roman*)]
        \item For every $\CG_i$-measurable vector $\lambda_i \in \R^p$ with $\|\lambda_i\| \leq C_1 d^{-1/2}$ a.s.,
        \begin{equation}\label{eq:abstract-linear-prop}
            \abs{\E\big[\inner{\lambda_i}{q^t(f_i)}-\inner{\lambda_i}{q^t(\widetilde f_i)}\big]}
            \le C d^{-5/2}.
        \end{equation}
        for some constant $C$ not dependent on $d$.
        \item For every $\CG_i$-measurable symmetric matrix $H_i \in  \R^{p \times p}$ with $\|H_i\|_\op \leq C_2 d^{-1}$ a.s.,
        \begin{equation}\label{eq:abstract-quadratic-prop}
            \left|\E\langle q^t(f_i), H_i q^s(f_i) \rangle -  \E\langle q^t(\widetilde f_i), H_i q^s(\widetilde f_i) \rangle \right| \leq C d^{-5/2}
        \end{equation}
        for some constant $C$ not dependent on $d$.
    \end{enumerate}
    Then Theorem \ref{thm:final-universality} holds.
\end{theorem}
\begin{proposition}\label{prop:remaining-quadratic}
        Fix $s, t \leq T$. Let $H_i$ be a $\CG_i$-measurable symmetric matrix with $\|H_i\|_\op \leq C_2 d^{-1}$. Then
        \begin{equation}
            \left|\E\langle q^t(f_i), H_i q^s(f_i) \rangle -  \E\langle q^t(\widetilde f_i), H_i q^s(\widetilde f_i) \rangle \right| \leq C d^{-5/2}
        \end{equation}
\end{proposition}
\begin{proposition}\label{prop:linear-comparison-complete}
    Fix $t \leq T$. Let $\lambda_i$ be any $\CG_i$-measurable vector satisfying $\|\lambda_i\| \leq C_1d^{-1/2}$ a.s. Then
    \begin{equation}\label{eq:linear-comparison-complete}
        \abs{\E\big[\inner{\lambda_i}{q^t(f_i)}-\inner{\lambda_i}{q^t(\widetilde f_i)}\big]}
        \le C d^{-5/2}.
    \end{equation}
\end{proposition}

\subsection{Proof of the Lindeberg Reduction}

To prove Theorem \ref{thm:reduction}, we state deterministic bounds on the iterates, the moment estimates for one-column perturbations, and the observable expansions.
\begin{lemma}\label{lem:base-det-bounds}
    Under Assumptions \ref{ass:basic} and \ref{ass:signal-init}, for every $t \leq T$,
    \begin{equation}
        \|\widehat{\bf s}^t(F)\| \leq C \sqrt{d}, \quad \left\|\hS^t \right\|_\op \leq C
    \end{equation}
    For some $C = C(C_0, C_1)$. The same bound holds for the leave-one-out iterates $\widehat{\bf s}^t_\smi$ and $\hS^t_\smi$.
\end{lemma}
\begin{proof}[Proof of Lemma \ref{lem:base-det-bounds}]
    At $t=0$, the bounds follow from Assumption~\ref{ass:signal-init} and the isometry of $\svecop$. For $t\geq1$, the matrices $\hS^t$ and $\hS^t_\smi$ are outputs of the spectral map $\sigma_{t-1}$, so their spectra are contained in $[-\|\sigma_{t-1}\|_\infty,\|\sigma_{t-1}\|_\infty]$. The lemma follows.
\end{proof}

\begin{proposition}\label{prop:delta-moments}
    For every integer $m\ge 1$ there is $C_m<\infty$ such that uniformly in $i$ and in $u\in\{f_i,\widetilde f_i\}$,
    \begin{equation}\label{eq:q-w-moment-bound}
        \max_{0\le t\le T}\norm{q^t(u)}_{L^m}
        +
        \max_{0\le t\le T-1}\Big(\norm{w^t(u)}_{L^m}+\norm{m^t(u)}_{L^m}\Big)
        \le C_m d^{-1/2}.
    \end{equation}
    \begin{equation}
        \max_{0 \,\le\, t \,\le\, T - 1} \norm{\beta^t_i(u)}_{L^m} \,\le\, C_m d^{-1/2}
    \end{equation}
    Moreover,
    \begin{equation}\label{eq:eps-moment-bound}
        \max_{0\,\le\, t\,\le\, T-1}\norm{\varepsilon_i^t(u)}_{L^m}\,\le\, C_m d^{-1}.
    \end{equation}
\end{proposition}

The proof of Proposition \ref{prop:delta-moments} is deferred to Appendix \ref{app:prob}.

\begin{proof}[Proof of Theorem \ref{thm:reduction}]
    For any observable $\Phi$ of the AMP iterates, treat it as a function of the collection of sensing matrix, $S^\star$, the initialization $\hS^0$, and the observation noise $\epsilon$. Since the latter $3$ quantities are fixed for the purpose of universality proof, we write $\Phi(F), \Phi(\widetilde F), \Phi(F^{(i)})$ for the observables generated by Gaussian sensing, rank-one sensing, and interpolated matrices $F^{(i)}$ defined in \eqref{eq:interpolating-matrix-def}. Write telescoping sum
    \begin{equation}\label{eq:lindeberg-telescope}
        \E\Phi(F) - \E \Phi(\widetilde F) = \sum_{i = 1}^n \E\Phi(F^{(i)}) - \E\Phi(F^{(i + 1)})
    \end{equation}
    We treat each observable separately.
    \begin{enumerate}
        \item Correlation with signal. See that the change caused by column replacement is
        \begin{equation}
            \frac{1}{d} \inner{{\bf s}^*}{q^t(f_i)}
        \end{equation}
        Because ${\bf s}^*$ is independent of $F,\widetilde F$, and $\|d^{-1} {\bf s}^*\| \leq Cd^{-1/2}$ based on Assumption \ref{ass:signal-init}, \eqref{eq:abstract-linear-prop} applies and we have
        \begin{equation}
            \left| \frac{1}{d} \E  \Tr(\hS^t(F^{(i)}) S^\star) - \frac{1}{d} \E\Tr(\hS^t(F^{(i + 1)}) S^\star) \right| \,=\, \left| \frac{1}{d} \E \inner{{\bf s}^*}{q^t(f_i)} - \frac 1 d \E \inner{{\bf s}^*}{q^t(\widetilde f_i)} \right| \,\leq\, Cd^{-5/2}
        \end{equation}
        Summing over $i$ in \eqref{eq:lindeberg-telescope} and using $n\leq C_0p=\CO(d^2)$ gives the claimed $\CO(d^{-1/2})$ bound.
        \item Cross-time correlation. Fix $0 \leq s, t \leq T$ and define
        \begin{equation}
            \Phi^\cross(A, B) = \frac{1}{d}\Tr(AB)
        \end{equation}
        Then for perturbations $\Delta_A,\Delta_B\in\Sym$,
        \begin{equation}\label{eq:cross-exact-expansion}
            \Phi^{\mathrm{cross}}(A+\Delta_A,B+\Delta_B)-\Phi^{\mathrm{cross}}(A,B)
            =\frac1d\Tr(\Delta_A B)+\frac1d\Tr(A\Delta_B)+\frac1d\Tr(\Delta_A\Delta_B).
        \end{equation}
        So in the vectorized form with $A = \hS^t_\smi, B = \hS^s_\smi$, the increment is a sum of two linear forms of norms $\CO(d^{-1/2})$ and a bilinear form with norm $\frac{1}{d}$. Apply \eqref{eq:cross-exact-expansion} with
        \begin{equation}
            A \,=\, \hS^t_\smi, \quad B \,=\, \hS^s_\smi, \quad \Delta_A \,=\, \mmat(q^t(u)), \quad \Delta_B \,=\, \mat(q^s(u))
        \end{equation}
        So we have
        \begin{equation}
            \begin{aligned}
                \frac 1 d \Tr(\hS^t(F(u)) \hS^s(u)) - \frac 1 d \Tr(\hS^t_\smi \hS^s_\smi) = \frac 1 d \inner{q^t(u)}{\widehat{\bf s}^s_\smi} + \frac{1}{d} \inner{q^s(u)}{\widehat{\bf s}^t_\smi} + \frac 1 d \inner{q^t(u)}{q^s(u)}
            \end{aligned}
        \end{equation}
        By Lemma \ref{lem:base-det-bounds}, $\|\widehat{\bf s}^t_\smi\|, \|\widehat{\bf s}^s_\smi\| \leq C \sqrt{d}$. So from hypothesis \eqref{eq:abstract-linear-prop}, we have
        \begin{equation}
            \abs{\frac 1 d \E \inner{\widehat{\bf s}^t_\smi}{q^s(f_i)} - \frac 1 d \E \inner{\widehat{\bf s}^t_\smi}{q^s(\widetilde f_i)}} \leq C d^{-5/2}
        \end{equation}
        for any $t, s$. Further, with $H_i \bydef \frac{1}{d}I_p$, so that $\|H_i\|_\op = \frac 1 d$. By hypothesis \eqref{eq:abstract-quadratic-prop}, we have
        \begin{equation}
            \abs{\frac 1 d \E \inner{q^t(f_i)}{q^s(f_i)} - \frac 1 d \E \inner{q^t(\widetilde f_i)}{q^s(\widetilde f_i)}} \leq C d^{-5/2}
        \end{equation}
        Therefore, we have
        \begin{equation}
            \abs{\E\Phi^\cross(\hS^t(f_i), \hS^s( f_i)) - \E\Phi^\cross(\hS^t(\widetilde f_i), \hS^s( \widetilde f_i))} \leq C d^{-5/2}
        \end{equation}
        Summing over $n\leq C_0p=\CO(d^2)$ terms gives, for any $s,t\in[T]$,
        \begin{equation}
            \abs{\E[\Phi^\cross_{s,t}(F) - \Phi^\cross_{s,t}(\widetilde F)]} = \CO(d^{-1/2})
        \end{equation}
        as desired.
        \item Spectral Statistics. Fix $\mu\in\SC^3(\R)$. Define
        \begin{equation}
            \Phi_\mu(S) \bydef \frac{1}{d} \Tr(\mu(S))
        \end{equation}
        We expand the spectral statistics as follows. Let
        \begin{equation}
            A_i \bydef \hS^t_\smi, \quad \Delta_i(u) \bydef \mmat(q^t(u))
        \end{equation}
        Lemma \ref{lem:base-det-bounds} bounds the operator norms of both $A_i$ and $A_i+\Delta_i(u)$ uniformly. Applying Lemma \ref{lem:smooth-spectral-expansion} on this deterministic spectral interval gives
        \begin{equation}
            \Phi_\mu(A_i +\Delta_i(u)) - \Phi_\mu(A_i) = \frac 1 d \Tr(\mu'(A_i) \Delta_i(u)) + Q_{\mu, A_i}(\Delta_i(u)) + R_{\mu, A_i}(\Delta_i(u))
        \end{equation}
        where $d^{-1}\|\mu'(A_i)\|_\Frob\leq\|\mu'\|_{L^\infty(I_\ast)}d^{-1/2}$, so it satisfies \eqref{eq:abstract-linear-prop}. Furthermore, by Lemma \ref{lem:smooth-spectral-expansion}, the quadratic form $Q_{\mu,A_i}$ is representable by a symmetric matrix $K_{\mu, A_i}$ such that $\|K_{\mu, A_i}\|_\op \leq C_\mu d^{-1}$, so it satisfies \eqref{eq:abstract-quadratic-prop}. Finally, we have from Proposition \ref{prop:delta-moments} and the remainder bound from Lemma \ref{lem:smooth-spectral-expansion},
        \begin{equation}
            \E\abs{R_{\mu, A_i}(\Delta_i(u))} \leq C_\mu d^{-1} \E\norm{q^t(u)}^3 \leq C_\mu d^{-5/2}
        \end{equation}
        Therefore, the term in the telescoping sum has expectation at most $Cd^{-5/2}$. Summing over $n\leq C_0p=\CO(d^2)$ terms gives the desired $\CO(d^{-1/2})$ bound.
    \end{enumerate}
\end{proof}

\section{Linearized Leave-One-Out Dynamics}\label{sec:linearized}

We linearize the leave-one-out recursion from Section~\ref{sec:aux} and approximate its response by $\bar q^t(u)$, whose coefficients depend only on the background trajectory. This approximation will be used in both one-column comparisons. Define the $2p$-dimensional state difference
\begin{equation}\label{eq:delta-state-def}
    \delta^t(u) \bydef \begin{pmatrix}
        q^t(u) \\
        w^{t - 1}(u)
    \end{pmatrix}  = \begin{pmatrix}
        \widehat{\bf s}^t(u) - \widehat{\bf s}^t_\smi \\
        h^{t - 1}(u) - h^{t - 1}_\smi
    \end{pmatrix}\in \R^{2p}, \quad t \geq 0
\end{equation}
For each $0 \le t \le T - 1$, define the block operators, with $L^t_i$ from \eqref{eq:L-epsilon-def} of Proposition \ref{prop:exact-loo} and $K_i$ from \eqref{eq:Ki-def},
\begin{equation}\label{eq:Mt-def}
    M^t_i \bydef \begin{pmatrix}
        L^t_i(I - a_t \kappa_t K_i) & a_t \rho_t L^t_i \\
        -\kappa_t K_i & \rho_t I_p
    \end{pmatrix} \in \R^{2p \times 2p}, \quad B^t_i \bydef \begin{pmatrix}
        a_t L^t_i \\
        I_p
    \end{pmatrix} \in \R^{2p \times p }
\end{equation}
Then we have from \eqref{eq:w-exact} of Proposition \ref{prop:exact-loo} and \eqref{eq:q-linearized} of Proposition \ref{prop:exact-loo} the one-step recursion
\begin{equation}\label{eq:block-recursion}
    \delta^{t + 1}(u) = M^t_i \delta^t(u) + B^t_i u \beta^t_i(u) + e^t_i(u)
\end{equation}
where
\begin{equation}\label{eq:block-error-def}
    e^t_i(u) \bydef \begin{pmatrix}
        \varepsilon^t_i(u) \\
        0
    \end{pmatrix}
\end{equation}
For $k \le t - 1$, define the response operator
\begin{equation}\label{eq:response-operator-def}
    \CP_{k,t} \bydef \Pi_1 M_i^{t -1 } \cdots M_i^{k + 1} B^k_i :\R^p \to \R^p
\end{equation}
where $\Pi_1$ denotes the projection onto the first $\R^p$ block. Then the recursion in \eqref{eq:block-recursion} gives for $1 \le t \le T$
\begin{equation}\label{eq:q-unrolled}
    q^t(u) = \sum_{k = 0}^{t - 1} \CP_{k,t} u \beta^k_i(u) + z^t_i(u)
\end{equation}
with $e_i^\ell$ as in \eqref{eq:block-error-def},
\begin{equation}\label{eq:z-unrolled}
    z^t_i(u) \bydef \Pi_1 \sum_{\ell = 0}^{t - 1} M^{t-1}_i \cdots M^{\ell + 1}_i e^\ell_i(u)
\end{equation}
Consequently, for $0 \le t \le T - 1$,
\begin{equation}\label{eq:beta-kernel-recursion}
    \beta^t_i(u) = \beta^t_{i,0}(u) - \sum_{k = 0}^{t - 1}\Omega_{k,t}(u) \beta^k_i(u) - \eta^t_i(u)
\end{equation}
where
\begin{align}
    \beta^t_{i,0}(u) &\bydef \kappa_t \inner{u}{{\bf s}^* - \widehat{\bf s}^t_\smi} + \nu_t \epsilon_i \label{eq:beta0-def} \\
    \Omega_{k,t}(u) &\bydef \kappa_t \inner{u}{\CP_{k,t}u} \label{eq:omega-def} \\
    \eta^t_i(u) &\bydef \kappa_t \inner{u}{z^t_i(u)} \label{eq:eta-def}
\end{align}
We now state the bounds on these operators.
\begin{lemma}\label{lem:response-op-bounds}
    For every $m \ge 1$, there exists $C_m < \infty$ such that uniformly in $i$,
    \begin{equation}
        \max_{0 \,\le\, k \,<\, t \,\le\, T} \norm{\CP_{k,t}}_{L^m(\op)} \,\le\, C_m, \quad \max_{t \,\le\, T} \norm{z^t_i(u)}_{L^m} \,\le\, C_m d^{-1}
    \end{equation}
    uniformly in $u \in \{f_i, \widetilde f_i\}$
\end{lemma}
\begin{proof}[Proof of Lemma \ref{lem:response-op-bounds}]
    The operator $\CP_{k,t}$ is a product of the block matrices $M_i^r$ in \eqref{eq:Mt-def}, followed by a projection. Each block has operator norm bounded by $C(1+\|K_i\|_\op)$, since $\|L_i^t\|_\op=\|DG_t(y_i^t)\|_\op\leq\|\sigma_t'\|_\infty\leq C$. Assumption~\ref{ass:basic}, Lemma~\ref{lem:sensing-envelope}, and \eqref{eq:G-lipschitz} therefore give
    \begin{equation}
        \norm{M^t_i}_\op \,\le\, C(1 + \norm{K_i}_\op) \,\implies\, \norm{M^t_i}_{L^m(\op)} \,\le\, C_m
    \end{equation}
    The products contain at most $T$ factors. Since $T$ is fixed, their operator norms have moments of every order bounded uniformly in $d$, proving the first bound.

    For the second equation, we simply apply \eqref{eq:z-unrolled} and Minkowski to get
    \begin{equation}
        \begin{aligned}
            \norm{z^t_i(u)}_{L^m} &\le\, \sum_{\ell \,=\, 1}^{t-1} \norm{\Pi_1 M^{t-1}_i \cdots M^{\ell + 1}_i \Pi_1^\top}_{L^{2m}(\op)} \norm{\varepsilon^\ell_i(u)}_{L^{2m}} \\
            &\le\, C_m \sum_{k\,=\,0}^{t-1} \norm{\varepsilon^k_i(u)}_{L^{2m}} \,\le\, C'_m d^{-1}
        \end{aligned}
    \end{equation}
    by Proposition \ref{prop:delta-moments}.
\end{proof}
We also need bounds on the scalar amplitudes and kernels.
\begin{corollary}\label{cor:scalar-bounds}
    For every $m \ge 1$ there exists $C_m < \infty$ such that uniformly in $i$ and $u \in \{f_i, \widetilde f_i\}$,
    \begin{equation}
        \max_{0 \,\le\, t \,\le\, T - 1} \norm{\beta^t_i(u)}_{L^m} \,\le\, C_m d^{-1/2}
    \end{equation}
    \begin{equation}
        \max_{0 \,\le\, t \,\le\, T - 1} \norm{\beta^t_{i,0}(u)}_{L^m} \,\le\, C_m d^{-1/2}
    \end{equation}
    \begin{equation}
        \max_{0 \,\le\, k \,<\, t \,\le\, T - 1} \norm{\Omega_{k,t}(u)}_{L^m} \,\le\, C_m, \quad \max_{0 \,\le\, t \,\le\, T - 1} \norm{\eta^t_i(u)}_{L^m}\,\le\, C_m d^{-1}
    \end{equation}
\end{corollary}
\begin{proof}[Proof of Corollary \ref{cor:scalar-bounds}]
    Proposition~\ref{prop:delta-moments} already bounds $\beta_i^t(u)$. For $\beta^t_{i,0}(u)$, \eqref{eq:beta0-def}, Lemma~\ref{lem:base-det-bounds}, Assumption~\ref{ass:signal-init}, and \eqref{eq:linear-moment-uncond} give
    \begin{equation}
        \norm{\inner{u}{{\bf s}^* - \widehat{\bf s}^t_\smi}}_{L^m}
        \le C_m d^{-1}\norm{{\bf s}^* - \widehat{\bf s}^t_\smi}_{L^m}
        \le C_m d^{-1/2}.
    \end{equation}
    Since $\nu_t\epsilon_i=\CO_{L^m}(d^{-1/2})$, this gives
    \begin{equation}
        \norm{\beta^t_{i,0}(u)}_{L^m}\le C_m d^{-1/2}.
    \end{equation}
    For the $\Omega_{k,t}$, apply Cauchy--Schwarz
    \begin{equation}
        |\Omega_{k,t}| \,\le\, |\kappa_t| \norm{\CP_{k,t}}_\op \norm{u}^2, \quad \norm{\Omega_{k,t}(u)}_{L^m} \,\le\, \abs{\kappa_t} \norm{\CP_{k,t}}_{L^{2m}(\op)} \norm{u}_{L^{4m}}^2
    \end{equation}
    which is bounded in $L^m$ uniformly in $d$ by Lemma \ref{lem:sensing-envelope} and Lemma \ref{lem:response-op-bounds}. By Cauchy--Schwarz again,
    \begin{equation}
        |\eta^t_i(u)| \,\le\, |\kappa_t| \norm{u} \norm{z^t_i(u)}, \quad \norm{\eta^t_i(u)}_{L^m} \,\le\, \abs{\kappa_t} \norm{u}_{L^{2m}} \norm{z^t_i(u)}_{L^{2m}}
    \end{equation}
    So the last bound is obtained by applying Lemma \ref{lem:response-op-bounds}.
\end{proof}

\subsection{Linearized Column Amplitudes}
Fix $i$ and suppress it in the notation below. We approximate the column amplitudes $\beta^t(u)$ by the sequence $(\alpha^t(u))_{t=0}^{T-1}$ defined recursively by
\begin{equation}\label{eq:alpha-lin-system}
    \alpha^t(u) \,=\, \beta^t_{i,0}(u) - \sum_{k\,=\,0}^{t - 1} \Omega_{k,t}(u) \alpha^k(u), \quad 0 \,\le\, t \,\le\, T - 1
\end{equation}
Then define the linearized version of the leave-one-out difference $q^t$ by
\begin{equation}\label{eq:q-lin-def}
    q^t_\lin(u) \,\bydef\, \sum_{k\,=\,0}^{t-1} \CP_{k,t} u \alpha^k(u), \quad t \,\ge\, 1
\end{equation}
with $q^0_\lin(u) = 0$. We can control the error incurred by replacing $q^t$ with its linearized version
\begin{proposition}\label{prop:q-lin-error}
    For every $m \ge 1$ there exists $C_m < \infty$ such that uniformly in $i$ and $u \in \{f_i, \widetilde f_i\}$.
    \begin{equation}
        \max_{t \le T} \norm{q^t(u) - q^t_\lin(u)}_{L^m} \le C_m d^{-1}
    \end{equation}
\end{proposition}
\begin{proof}[Proof of Proposition \ref{prop:q-lin-error}]
    Subtracting \eqref{eq:alpha-lin-system} from \eqref{eq:beta-kernel-recursion}, with $r^t(u) \bydef \beta^t(u) - \alpha^t(u)$ distinct from the AMP residual $r^t$ in \eqref{eq:AMP-recursion}, we obtain
    \begin{equation}
        r^t(u) = -\sum_{k=0}^{t-1}\Omega_{k,t}(u) r^k(u) - \eta^t(u)
    \end{equation}
    Since $T$ is fixed and we have shown in Corollary \ref{cor:scalar-bounds} that $\Omega_{k,t}(u)$ have moments of all orders uniformly bounded in $d$, we have by Minkowski and induction,
    \begin{equation}\label{eq:r-scalar-bound}
        \max_{t \le T - 1} \norm{r^t(u)}_{L^m} \le C_m d^{-1}
    \end{equation}
    Now subtract \eqref{eq:q-lin-def} from \eqref{eq:q-unrolled}, we have
    \begin{equation}
        q^t(u) - q ^t_\lin(u) = \sum_{k=0}^{t-1} \CP_{k,t}u r^k(u) + z^t_i(u)
    \end{equation}
    Using Lemma \ref{lem:response-op-bounds}, Lemma \ref{lem:sensing-envelope}, and \eqref{eq:r-scalar-bound}, we get
    \begin{equation}
        \begin{aligned}
            \norm{q^t(u) - q^t_\lin(u)}_{L^m} &\le\, \sum_{k\,=\,0}^{t-1} \norm{\CP_{k,t}}_{L^{3m}(\op)} \norm{u}_{L^{3m}} \norm{r^k(u)}_{L^{3m}} + \norm{z^t_i(u)}_{L^m} \\
            &\le C_m d^{-1}.
        \end{aligned}
    \end{equation}
    as desired.
\end{proof}
\subsection{Kernel Bounds}
Recall that the scalar response kernel $\Omega_{k,t}(u)$ is the quadratic form given by
\begin{equation}
    \Omega_{k,t}(u) = \kappa_t u^\top \CP_{k,t}u
\end{equation}
We approximate this quadratic form by its conditional mean given $\CG_i$. Define
\begin{equation}\label{eq:omega-bar-def}
    \bom_{k,t} \,\bydef\, \frac{\kappa_t}{p} \Tr \CP_{k,t}, \quad 0 \,\le\, k \,<\, t \,\le\, T
\end{equation}
and define also the fluctuation of $\Omega_{k,t}$ around its mean,
\begin{equation}\label{eq:omega-fluctuation-def}
    \tom_{k,t}(u) \bydef \Omega_{k,t}(u) - \bom_{k,t}
\end{equation}
We prove concentration
\begin{lemma}\label{lem:omega-fluctuations}
    For every $m \ge 1$, there exists $C_m < \infty$ such that uniformly in $i$, in $0 \le k < t \le T$, and in $u \in \{f_i, \widetilde f_i\}$,
    \begin{equation}
        \norm{\tom_{k,t}(u)}_{L^m} \le C_m d^{-1/2}
    \end{equation}
    and
    \begin{equation}
        \norm{\bom_{k,t}}_{L^m} \le C_m
    \end{equation}
\end{lemma}

\begin{proof}[Proof of Lemma \ref{lem:omega-fluctuations}]
    By the definitions \eqref{eq:omega-def}, \eqref{eq:response-operator-def}, \eqref{eq:omega-bar-def}, and \eqref{eq:omega-fluctuation-def}, we have
    \begin{equation}
        \tom_{k,t}(u) = \kappa_t \left( u^\top \CP_{k,t}u - \frac 1 p  \Tr \CP_{k,t} \right).
    \end{equation}
    Applying Corollary \ref{cor:one-column-moments} (ii) with $A = \CP_{k,t}$ and using Lemma \ref{lem:response-op-bounds}, we get
    \begin{equation}
        \norm{\tom_{k,t}(u)}_{L^{m}} \le C_m \norm{\CP_{k,t}}_{L^{2m}(\op)} d^{-1/2} \le C_m d^{-1/2}.
    \end{equation}
    Now we have
    \begin{equation}
        |\bom_{k,t}| \le |\kappa_t| \norm{\CP_{k,t}}_{\op},
    \end{equation}
    so
    \begin{equation}
        \norm{\bom_{k,t}}_{L^m} \le |\kappa_t| \norm{\CP_{k,t}}_{L^m(\op)} \le C_m
    \end{equation}
    by Lemma \ref{lem:response-op-bounds}.
\end{proof}

We replace $\Omega_{k,t}(u)$ by $\bom_{k,t}$ to obtain the recursion
\begin{equation}\label{eq:bar-alpha-recursion}
    \bar \alpha^t(u) \bydef \beta^t_{i,0}(u) - \sum_{k=0}^{t-1} \bom_{k,t} \bar\alpha^k(u)
\end{equation}
Its coefficients are $\CG_i$-measurable, and its dependence on $u$ comes through the affine terms $\beta^t_{i,0}(u)$. The next proposition bounds the approximation error.
\begin{proposition}\label{prop:bar-alpha}
    For every fixed $t \le T - 1$ and every integer $m \ge 1$, the following holds uniformly in $i$ and in $u \in \{f_i, \widetilde f_i\}$.
    \begin{enumerate}[label = (\roman*)]
        \item There exist $\CG_i$-measurable coefficients $\{c_{t,r}\}_{r = 0}^t$ such that
        \begin{equation}\label{eq:bar-alpha-linear-combo}
            \bar \alpha^t(u) \,=\, \sum_{r\,=\,0}^t c_{t,r} \beta^r_{i,0}(u), \quad c_{t,t} \,=\, 1
        \end{equation}
        and
        \begin{equation}\label{eq:c-coeff-bound}
            \max_{r \le t} \norm{c_{t,r}}_{L^m} \le C_m
        \end{equation}
        \item The error
        \begin{equation}\label{eq:d-def}
            d^t(u) \bydef \beta^t(u) - \bar\alpha^t(u)
        \end{equation}
        is bounded by
        \begin{equation}\label{eq:d-bound}
            \norm{d^t(u)}_{L^m} \le C_m d^{-1}
        \end{equation}
    \end{enumerate}
\end{proposition}
\begin{proof}[Proof of Proposition \ref{prop:bar-alpha}]
    We prove (i) by induction on $t$. Since $\bar\alpha^0(u)=\beta^0_{i,0}(u)$, we take $c_{0,0}=1$. Assume \eqref{eq:bar-alpha-linear-combo} holds for all $r<t$. Expanding the recursion gives
    \begin{equation}
        \begin{aligned}
            \bar\alpha^t(u) &=\, \beta^t_{i,0}(u) - \sum_{k\,=\,0}^{t-1} \bom_{k,t} \sum_{r\,=\,0}^k c_{k,r} \beta^r_{i,0}(u) \\
            &=\, \beta^t_{i,0}(u) + \sum_{r\,=\,0}^{t-1} \left(- \sum_{k\,=\,r}^{t-1} \bom_{k,t} c_{k,r} \right)\beta^r_{i,0}(u)
        \end{aligned}
    \end{equation}
    But observe that $\left(- \sum_{k=r}^{t-1} \bom_{k,t} c_{k,r} \right)$ is $\CG_i$-measurable. Therefore, \eqref{eq:bar-alpha-linear-combo} holds with
    \begin{equation}
        c_{t,t} \,=\, 1, \quad c_{t,r} \,=\, - \sum_{k \,=\, r}^{t-1} \bom_{k,t} c_{k,r} \quad (0 \,\le\, r \,<\, t)
    \end{equation}
    Since $T$ is fixed, the moment bounds on $\bom_{k,t}$ in Lemma~\ref{lem:omega-fluctuations} and repeated applications of H\"older's and Minkowski's inequalities give \eqref{eq:c-coeff-bound}.

    To prove (ii), we subtract the $\CG_i$-measurable approximation \eqref{eq:bar-alpha-linear-combo} from the exact recursion \eqref{eq:beta-kernel-recursion} to get
    \begin{equation}\label{eq:d-recursion}
        d^t(u) \,=\, - \sum_{k\,=\,0}^{t-1} \bom_{k,t} d^k(u) - \sum_{k \,=\, 0}^{t - 1} \tom_{k,t}(u)\beta^k_i(u) - \eta^t_i(u)
    \end{equation}
    At $t=0$, the right-hand side is zero. Assume \eqref{eq:d-bound} holds for all $k<t$. Lemma~\ref{lem:omega-fluctuations} and Corollary~\ref{cor:scalar-bounds} give
    \begin{equation}
        \norm{ \sum_{k\,=\,0}^{t-1} \tom_{k,t}(u)\beta^k_i(u) }_{L^m} \,\le\, \sum_{k\,=\,0}^{t-1} \norm{\tom_{k,t}(u)}_{L^{2m}} \norm{\beta^k_i(u)}_{L^{2m}} \,\le\, C_m d^{-1}
    \end{equation}
    and
    \begin{equation}
        \norm{\eta^t_i(u)}_{L^m} \le C_m d^{-1}
    \end{equation}
    Applying Minkowski's and H\"older's inequalities to \eqref{eq:d-recursion}, using the fixed horizon and the uniform moment bounds for $\bom_{k,t}$, gives
    \begin{equation}
        \norm{d^t(u)}_{L^m} \,\le\, C_m d^{-1} + C_m \sum_{k\,=\,0}^{t-1} \norm{d^k(u)}_{L^{2m}} \,\le\, C_m d^{-1}
    \end{equation}
    so the claim follows from induction.
\end{proof}
Using these amplitude approximations, define
\begin{equation}\label{eq:bar-q-def}
    \bar q^t(u) \,=\, \sum_{k\,=\,0}^{t-1} \CP_{k,t} u \bar\alpha^k(u), \quad t \,\ge\, 1
\end{equation}
with $\bar q^0(u)=0$. The following corollary gives its error bound and a representation in terms of $\beta_{i,0}^r(u)$.
\begin{corollary}\label{cor:bar-q}
    For every $t \le T$,
    \begin{equation}\label{eq:q-minus-barq-decomp}
        q^t(u) - \bar q^t(u) = \sum_{k=0}^{t-1} \CP_{k,t} u\  d^k(u) + z^t_i(u)
    \end{equation}
    So for every integer $m \ge 1$, we have uniformly in $i$ and $u \in \{f_i, \widetilde f_i\}$,
    \begin{equation}\label{eq:q-minus-barq}
        \norm{q^t(u) - \bar q^t(u)}_{L^m} \le C_m d^{-1}
    \end{equation}
    Moreover, there exist $\CG_i$-measurable matrices $(\CB_{t,r})_{r = 0}^{t-1} : \R^p \to \R^p$ such that
    \begin{equation}\label{eq:bar-q-linear-combo}
        \bar q^t(u) = \sum_{r = 0}^{t-1} \CB_{t,r} u \beta^r_{i,0}(u)
    \end{equation}
    and satisfies
    \begin{equation}\label{eq:B-operator-bound}
        \max_{r \le t - 1} \norm{\CB_{t,r}}_{L^m(\op)} \le C_m
    \end{equation}
\end{corollary}
\begin{proof}[Proof of Corollary \ref{cor:bar-q}]
    Subtract \eqref{eq:bar-q-def} from the exact formula \eqref{eq:q-unrolled}, we obtain
    \begin{equation}
        q^t(u) - \bar q^t(u) = \sum_{k=0}^{t-1} \CP_{k,t} u\  d^k(u) + z^t_i(u)
    \end{equation}
    By Proposition \ref{prop:bar-alpha}, Lemma \ref{lem:response-op-bounds}, Lemma \ref{lem:sensing-envelope}, and the bound on $z^t_i(u)$ from Lemma \ref{lem:response-op-bounds}, we have from H\"older,
    \begin{equation}
        \begin{aligned}
            \norm{q^t(u) - \bar q^t(u)}_{L^m} &\le\, \sum_{k\,=\,0}^{t-1}\norm{\CP_{k,t}}_{L^{3m}(\op)} \norm{u}_{L^{3m}} \norm{d^k(u)}_{L^{3m}} + \norm{z^t_i(u)}_{L^m} \\
            &\le C_m d^{-1}
        \end{aligned}
    \end{equation}
    so \eqref{eq:q-minus-barq} is shown. Now plug the formulation of $\bar \alpha^t$ in \eqref{eq:bar-alpha-linear-combo} of Proposition \ref{prop:bar-alpha} into the definition of $\bar q^t$ in \eqref{eq:bar-q-def}, we have
    \begin{equation}
        \begin{aligned}
            \bar q^t(u) \,=\, \sum_{k\,=\,0}^{t-1}\sum_{r\,=\,0}^k \CP_{k,t} u \ c_{k,r} \beta^r_{i,0}(u) \,=\, \sum_{r\,=\,0}^{t-1} \left( \sum_{k\,=\,r}^{t-1} c_{k,r} \CP_{k,t} \right) u \beta^r_{i,0}(u)
        \end{aligned}
    \end{equation}
    Therefore, \eqref{eq:bar-q-linear-combo} holds with
    \begin{equation}\label{eq:Btr-def}
        \CB_{t,r} \bydef \sum_{k=r}^{t-1} c_{k,r}\CP_{k,t}.
    \end{equation}
    The $L^m$ bound on the operator norm follows by triangle inequality, Proposition \ref{prop:bar-alpha}, and Lemma \ref{lem:response-op-bounds}.
\end{proof}

We use these bounds to prove Propositions~\ref{prop:remaining-quadratic} and \ref{prop:linear-comparison-complete} in the next two sections.

\section{Quadratic One-Column Comparison}\label{sec:proof-quadratic}
To prove Proposition~\ref{prop:remaining-quadratic}, we replace $q^t(u)$ by $\bar q^t(u)$ from Section~\ref{sec:linearized}. The trace terms then agree by moment matching; the centered quadratic-form terms are bounded using the one-column moment estimates.

\begin{proof}[Proof of Proposition \ref{prop:remaining-quadratic}]
    We first show that the quadratic form is well-approximated by the $\CG_i$-linearized replacement. In particular, by Corollary \ref{cor:bar-q}, Proposition \ref{prop:delta-moments}, and the hypothesis that $\norm{H_i}_\op \le Cd^{-1}$, we have from H\"older
    \begin{equation}\label{eq:quad-error-r1}
        \begin{aligned}
            \E \abs{\inner{q^t(u)}{H_i q^s(u)} - \inner{\bar q^t(u)}{H_i \bar q^s(u)} } &\le \norm{H_i}_\op \bigg( \norm{q^t(u) - \bar q^t(u)}_{L^2} \norm{q^s(u)}_{L^2} \\
            & \ + \norm{\bar q^t(u)}_{L^2} \norm{q^s(u) - \bar q^s(u)}_{L^2} \bigg) \\
            &\le Cd^{-5/2}
        \end{aligned}
    \end{equation}
    uniformly in $u \in \{f_i, \widetilde f_i\}$ because both $q^r(u)$ and $\bar q^r(u)$ are $\CO_{L^2}(d^{-1/2})$ while $q^r(u) - \bar q^r(u) = \CO_{L^2}(d^{-1})$. Therefore, it suffices to compare $\inner{\bar q^t(u)}{H_i \bar q^s(u)}$ over $u \in \{f_i, \widetilde f_i\}$. By expansion \eqref{eq:bar-q-linear-combo}, we have
    \begin{equation}
        \bar q^t(u)=\sum_{r=0}^{t-1}\mathcal B_{t,r}u\ \beta_{i,0}^r(u),
        \qquad
        \bar q^s(u)=\sum_{\ell=0}^{s-1}\mathcal B_{s,\ell}u\ \beta_{i,0}^{\ell}(u).
    \end{equation}
    So we can write
    \begin{equation}\label{eq:barq-quad-expansion}
        \inner{\bar q^t(u)}{H_i \bar q^s(u)} \,=\, \sum_{r\,=\,0}^{t-1} \sum_{\ell\,=\,0}^{s-1} \beta^r_{i,0}(u) \beta^\ell_{i,0}(u) \inner{u}{M_{t,s,r,\ell} u}
    \end{equation}
    where $M_{t,s,r,\ell} = \CB_{t,r}^\top H_i \CB_{s,\ell}$, with $\CB_{t,r}$ defined in \eqref{eq:Btr-def}. By \eqref{eq:B-operator-bound}, the matrices $\CB_{t,r}$ have moments bounded uniformly in $d$. Hence, for any $m \ge 1$,
    \begin{equation}\label{eq:M-op-bound}
        \max_{t,s,r,\ell} \norm{M_{t,s,r,\ell}}_{L^m(\op)} \le C_m d^{-1}
    \end{equation}
    Since $M_{t,s,r,\ell}$ is $\CG_i$-measurable and $u$ is isotropic, we again approximate the quadratic form by the trace. Write
    \begin{equation}
        u^\top M_{t,s,r,\ell} u \,=\, \frac 1 p \Tr M_{t,s,r,\ell} + \tilde m_{t,s,r,\ell}(u), \quad \tilde m_{t,s,r,\ell}(u) \,\bydef\, u^\top M_{t,s,r,\ell} u - \frac 1 p \Tr M_{t,s,r,\ell}
    \end{equation}
    Plugging this into \eqref{eq:barq-quad-expansion} gives
    \begin{equation}\label{eq:barq-leading-split}
        \inner{\bar q^t(u)}{H_i \bar q^s(u)} = T(u) + R(u)
    \end{equation}
    where
    \begin{equation}
        T(u) \,\bydef\, \frac 1 p \sum_{r \,=\, 0}^{t - 1} \sum_{\ell \,=\, 0}^{s - 1} \Tr(M_{t,s,r,\ell}) \beta^r_{i,0}(u) \beta^\ell_{i,0}(u), \quad R(u) \,\bydef\, \sum_{r \,=\, 0}^{t - 1} \sum_{\ell \,=\, 0}^{s - 1} \tilde m_{t,s,r,\ell}(u)\beta^r_{i,0}(u) \beta^\ell_{i,0}(u)
    \end{equation}
    To compare the $T$ term, we condition on $\CG_i$. Recall, $a^r_i \bydef {\bf s}^* - \widehat{\bf s}^r_\smi \in \R^p, \ \ \beta^{r}_{i,0} = \kappa_r \inner{u}{a^r_i} +\nu_r\epsilon_i$. Therefore, by Lemma \ref{lem:one-column-moments}, we have
    \begin{equation}\label{eq:betar-betal-condition-on-Gi}
        \E[\beta^r_{i,0}(u) \beta^\ell_{i,0}(u) \ | \ \CG_i] = \frac{\kappa_r \kappa_\ell}{p} \inner{a^r_i}{a^\ell_i} + \nu_r \nu_\ell \epsilon_i^2
    \end{equation}
    This conditional expectation is the same for $f_i$ and $\widetilde f_i$. Taking expectations gives
    \begin{equation}\label{eq:trace-part-equal}
        \E[T(f_i) | \CG_i] \,=\, \E[T(\widetilde f_i) | \CG_i]  \,\implies\, \E[T(f_i) - T(\widetilde f_i)] \,=\, 0
    \end{equation}
    Now we bound the remainder for the quadratic form. By Corollary \ref{cor:scalar-bounds}, we have $\norm{\beta^r_{i,0}(u)}_{L^4} \le Cd^{-1/2}$. So by \eqref{eq:M-op-bound} and Corollary \ref{cor:one-column-moments}, we have $\norm{\tilde m_{t,s,r,\ell}}_{L^2} \le Cd^{-3/2}$. Therefore, by H\"older,
    \begin{equation}
        \E|R(u)| \,\le\, \sum_{r \,=\, 0}^{t-1} \sum_{\ell\,=\,0}^{s-1} \norm{\beta^r_{i,0}(u)}_{L^4} \norm{\beta^\ell_{i,0}(u)}_{L^4} \norm{\tilde m_{t,s,r,\ell}}_{L^2} \,\le\, Cd^{-5/2}
    \end{equation}
    uniformly in $u \in \{f_i, \widetilde f_i\}$. Together with \eqref{eq:barq-leading-split} and\eqref{eq:trace-part-equal}, we have
    \begin{equation}
        \abs{\E\inner{\bar q^t(f_i)}{H_i \bar q^s(f_i)}  - \E\inner{\bar q^t(\widetilde f_i)}{H_i \bar q^s(\widetilde f_i)}  } \le Cd^{-5/2}
    \end{equation}
    Together with \eqref{eq:quad-error-r1}, this proves the proposition.
\end{proof}

\section{Linear One-Column Comparison}\label{sec:proof-linear}

We now prove the linear comparison in Proposition~\ref{prop:linear-comparison-complete}. The approximation $\bar q^t(u)$ from \eqref{eq:bar-q-def} has the same expected linear contribution under rank-one and Gaussian sensing. We therefore need to control the propagated nonlinear remainder $z^t_i(u)$ from \eqref{eq:z-unrolled}. Throughout, $u\in\{f_i,\widetilde f_i\}$ and $\CG_i$ is the sigma-algebra in \eqref{eq:loo-sigma-algebra-def}.

\begin{proposition}\label{prop:leading-linear-barq}
    Fix $t \le T$. Let $\lambda_i$ be any $\CG_i$-measurable vector. Then
    \begin{equation}
        \E[\inner{\lambda_i}{\bar q^t(f_i)} | \CG_i] = \E[ \inner{\lambda_i}{\bar q^t(\widetilde f_i)} | \CG_i], \quad a.s.
    \end{equation}
\end{proposition}

The proof is deferred to Section \ref{subsec:linear-comparison-deferred-proofs}.

By Proposition \ref{prop:leading-linear-barq}, we can ignore the term $\bar q^t(u)$ when comparing the linear term. We will use Corollary \ref{cor:bar-q} to show that, for verifying the hypothesis in \eqref{eq:abstract-linear-prop}, it suffices to work with the remainder $z^t_i(u)$.
\begin{proposition}\label{prop:remaining-linear}
    Fix $t \le T$. Let $\lambda_i$ be any $\CG_i$-measurable vector satisfying $\norm{\lambda_i} \le C d^{-1/2}$ almost surely. Then
    \begin{equation}\label{eq:linear-remainder-reduction}
        \abs{\E\big[\inner{\lambda_i}{q^t(f_i)} - \inner{\lambda_i}{q^t(\widetilde f_i)}\big]} \le \abs{\E\big[\inner{\lambda_i}{z_i^t(f_i)} - \inner{\lambda_i}{z_i^t(\widetilde f_i)}\big]} + C d^{-5/2}.
    \end{equation}
    In particular, if the projected nonlinear remainder satisfies
    \begin{equation}\label{eq:projected-z-needed}
        \abs{\E\big[\inner{\lambda_i}{z_i^t(f_i)} - \inner{\lambda_i}{z_i^t(\widetilde f_i)}\big]} \le C d^{-5/2},
    \end{equation}
    then
    \begin{equation}\label{eq:remaining-linear-conclusion}
        \abs{\E\big[\inner{\lambda_i}{q^t(f_i)} - \inner{\lambda_i}{q^t(\widetilde f_i)}\big]} \le C d^{-5/2}.
    \end{equation}
\end{proposition}

The proof is deferred to Section \ref{subsec:linear-comparison-deferred-proofs}.

\subsection{Linearizing \texorpdfstring{$m^t(u)$}{m(t,u)}}
To control the nonlinear remainder, we first express $m^t(u)=q^t(u)+a_tw^t(u)$ in terms of the state difference $\delta^t(u)$ from \eqref{eq:delta-state-def}. Define $\Gamma_t:\R^{2p}\to\R^p$ by
\begin{equation}\label{eq:Gamma-def}
    \Gamma_{t}(x, y) \bydef (I_p - a_t \kappa_t K_i) x + a_t \rho_t y
\end{equation}
where $K_i$ is defined in \eqref{eq:Ki-def}, and the coefficients $a_t,\kappa_t,\rho_t$ are from \eqref{eq:scaled-coeffs}. The recursion in \eqref{eq:w-exact} of Proposition \ref{prop:exact-loo}, with $\beta_i^t(u)$ defined in \eqref{eq:beta-exact} of Proposition \ref{prop:exact-loo}, gives
\begin{equation}\label{eq:m-from-delta}
    m^t(u) = \Gamma_t \delta^t(u) + a_t u \beta^t_i(u)
\end{equation}
For $0 \le k \le t \le T - 1$, define operators
\begin{equation}\label{eq:m-response-def}
    \CM_{k,t} \bydef \begin{cases}
        \Gamma_t M_i^{t-1} \cdots M_i^{k+1} B^k_i, & 0 \le k \le t - 1 \\
        a_t I_p, & k = t
    \end{cases}
\end{equation}
where $M^t_i, B^t_i$ are defined by \eqref{eq:Mt-def}. With $e_i^\ell(u)$ defined in \eqref{eq:block-error-def}, we define the propagated error by
\begin{equation}\label{eq:zeta-def}
    \zeta^t(u) \bydef \Gamma_t \sum_{\ell=0}^{t-1} M_i^{t-1} \cdots M_i^{\ell+1} e^\ell_i(u)
\end{equation}
with $\zeta^0(u) = 0$. This allows us to rewrite $m^t(u)$ in terms of $\beta^k_i(u)$ using the operators $\CM_{k,t}$ from \eqref{eq:m-response-def} and the propagated error $\zeta^t$ from \eqref{eq:zeta-def}.
\begin{lemma}\label{lem:m-representation}
    For every fixed $t \leq T - 1$,
    \begin{equation}\label{eq:m-exact-unrolled}
        m^t(u) = \zeta^t(u) + \sum_{k=0}^t \CM_{k,t} u \beta^k_i(u)
    \end{equation}
    Moreover, for every $m \ge 1$, we have
    \begin{equation}\label{eq:m-response-op-bound}
        \max_{k\le t} \norm{\CM_{k,t}}_{L^m(\op)} \le C_m
    \end{equation}
    \begin{equation}\label{eq:zeta-exact-bound}
        \norm{\zeta^t(u)}_{L^m} \le C_m d^{-1}
    \end{equation}
    uniformly in $u \in \{f_i, \widetilde f_i\}$. Define also the $\CG_i$ approximation for $m^t(u)$, with $\bar\alpha^k(u)$ from \eqref{eq:bar-alpha-recursion},
    \begin{equation}\label{eq:mbar-def}
        \bar m^t(u) \bydef \sum_{k=0}^t \CM_{k,t}u \bar \alpha^k(u)
    \end{equation}
    Then there exist $\CG_i$-measurable operators $(\CR_{t,r})_{r = 0}^t:\R^p \to \R^p$ such that
    \begin{equation}\label{eq:mbar-linear-combo}
        \bar m^t(u) \,=\, \sum_{r\,=\,0}^t \CR_{t,r}u \beta^r_{i,0}(u), \quad \max_{r \,\le\, t} \norm{\CR_{t,r}}_{L^m(\op)} \,\le\, C_m
    \end{equation}
    With $d^k(u)$ from \eqref{eq:d-def} of Proposition \ref{prop:bar-alpha}, define the unbolded remainder
    \begin{equation}\label{eq:s-remainder-def}
        s^t(u) \bydef \sum_{k=0}^t \CM_{k,t} u\ d^k(u).
    \end{equation}
    Then we have
    \begin{align}
        m^t(u) - \bar m^t(u) &= s^t(u) + \zeta^t(u) \label{eq:m-minus-mbar-decomp} \\
        \norm{\bar m^t(u)}_{L^m} &\le C_m d^{-1/2} \label{eq:mbar-size} \\
        \norm{s^t(u)}_{L^m} &\le C_m d^{-1}\label{eq:s-remainder-bound}
    \end{align}
\end{lemma}

The proof is deferred to Section \ref{subsec:linear-comparison-deferred-proofs}.

\subsection{Propagated Remainder}

We first state the third-order spectral expansion of the one-column error $\varepsilon^k_i(u)$, obtained by applying Lemma \ref{lem:spectral-taylor-4} to the leave-one-out iterates.

\begin{corollary}\label{cor:epsilon-third-order-vector}
    Fix $i\in[n]$ and $0\le k\le T-1$. Let $\varepsilon_i^k(u)$ be as
    defined in \eqref{eq:q-linearized} of Proposition \ref{prop:exact-loo}. Let
    \begin{equation}
        y^k_i \bydef \widehat{\bf s}^k_\smi+a_kh^k_\smi
    \end{equation}
    as in Proposition \ref{prop:exact-loo}. Define
    \begin{equation}\label{eq:vector-local-def}
        \bf B_i^k(g_1,g_2) \bydef \bf B_{k,y^k_i}(g_1,g_2), \qquad
        \bf T_i^k(g_1,g_2,g_3) \bydef \mathbf T_{k,y^k_i}(g_1,g_2,g_3),
    \end{equation}
    and
    \begin{equation}\label{eq:vector-local-remainder-def}
        \rho_{i,k}^{(4)}(g) \bydef \rho_{k,y^k_i}^{(4)}(g).
    \end{equation}
    Then for every $u\in\R^p$,
    \begin{equation}\label{eq:epsilon-third-order-vector}
        \varepsilon_i^k(u) = \bf B_i^k(m^k(u),m^k(u)) + \mathbf T_i^k(m^k(u),m^k(u),m^k(u))
        +
        \rho_{i,k}^{(4)}(m^k(u)).
    \end{equation}
    Moreover, there is a constant $C<\infty$, depending only on
    $M_{\sigma,4}$ from \eqref{eq:M-sigma-4-def} in Lemma \ref{lem:spectral-taylor-4}, such that for all $g,g_1,g_2,g_3,w\in\R^p$,
    \begin{equation}\label{eq:vector-local-bound}
        \norm{\mathbf B_i^k(g_1,g_2)} \le C\norm{g_1}\norm{g_2},
        \qquad
        \norm{\mathbf T_i^k(g_1,g_2,g_3)} \le C\norm{g_1}\norm{g_2}\norm{g_3},
        \qquad
        \norm{\rho_{i,k}^{(4)}(g)} \le C\norm{g}^4,
    \end{equation}
    and
    \begin{equation}\label{eq:scalarized-local-B-HS}
        \left\|\inner{w}{\mathbf B_i^k(\cdot,\cdot)}\right\|_\HS
        \le C\sqrt d\norm{w},
        \qquad
        \left\|\inner{w}{\mathbf T_i^k(\cdot,\cdot,\cdot)}\right\|_\HS
        \le Cd\norm{w}.
    \end{equation}
\end{corollary}

\begin{proof}[Proof of Corollary \ref{cor:epsilon-third-order-vector}]
    By the definition of $\varepsilon_i^k(u)$ in \eqref{eq:q-linearized} of Proposition \ref{prop:exact-loo}, with $G_k$ defined in \eqref{eq:vectorized-AMP-recursion},
    \begin{equation}
        \varepsilon_i^k(u) = G_k(y_i^k+m^k(u))-G_k(y_i^k)-DG_k(y_i^k)m^k(u).
    \end{equation}
    Applying Lemma \ref{lem:spectral-taylor-4} with $t=k$, $y=y_i^k$, and
    $g=m^k(u)$ gives \eqref{eq:epsilon-third-order-vector}. The bounds
    \eqref{eq:vector-local-bound} and \eqref{eq:scalarized-local-B-HS} are \eqref{eq:local-spectral-B-bound}--\eqref{eq:local-spectral-T-HS} of Lemma \ref{lem:spectral-taylor-4}
    under the definitions in \eqref{eq:vector-local-def}.
\end{proof}

In what follows, we suppress the index $i$ and write $\bf B^k=\bf B_i^k$ and $\bf T^k=\bf T_i^k$.

Fix $t \le T$ and a $\CG_i$-measurable vector $\lambda_i$ satisfying $\norm{\lambda_i} \le C d^{-1/2}$ a.s. For $0 \le k \le t - 1$, define
\begin{equation}\label{eq:vartheta-def}
    \vartheta_{k,t} \bydef \Pi_1 \left( M_i^{t-1} \cdots M_i^{k + 1} \right)^\top \Pi_1^\top \lambda_i \in \R^p
\end{equation}
where $\Pi_1$ projects onto the first copy of $\R^p$ in $\R^{2p}$. Thus,
\begin{equation}\label{eq:vartheta-use}
    \inner{\lambda_i}{\Pi_1 M_i^{t-1}\cdots M_i^{k+1} \Pi_1^\top \varepsilon^k_i(u)} = \inner{\vartheta_{k,t}}{\varepsilon^k_i(u)}
\end{equation}
The following lemma bounds $\vartheta$ and gives the projected remainder formula.
\begin{lemma}\label{lem:projected-z-representation}
    For every $0 \le k \le t - 1$ and every integer $m \ge 1$,
    \begin{equation}\label{eq:projected-z-unrolled}
        \inner{\lambda_i}{z^t_i(u)} = \sum_{k=0}^{t-1} \inner{\vartheta_{k,t}}{\varepsilon^k_i(u)}
    \end{equation}
    and we also have
    \begin{equation}\label{eq:vartheta-bound}
        \norm{\vartheta_{k,t}}_{L^m} \le C_m d^{-1/2}
    \end{equation}
\end{lemma}
\begin{proof}[Proof of Lemma \ref{lem:projected-z-representation}]
    The product $\left( M_i^{t-1} \cdots M_i^{k + 1} \right)^\top$, with $M_i^r$ as in \eqref{eq:Mt-def}, has operator norm moments bounded uniformly in $d$ by Lemma \ref{lem:response-op-bounds}. Since $\norm{\lambda_i} \le Cd^{-1/2}$, this gives \eqref{eq:vartheta-bound}. Equation \eqref{eq:projected-z-unrolled} is a direct consequence of \eqref{eq:vartheta-use} and the definition \eqref{eq:z-unrolled}.
\end{proof}
For $0 \le k \le t- 1$, define
\begin{equation}\label{eq:projected-BT-def}
    \FB_{k,t}(g_1, g_2) \,\bydef\, \inner{\vartheta_{k,t}}{\bf B^k(g_1, g_2)}, \quad \FT_{k,t}(g_1, g_2, g_3) \,\bydef\, \inner{\vartheta_{k,t}}{\bf T^k(g_1, g_2, g_3)}
\end{equation}
where $\vartheta_{k,t}$ is defined in \eqref{eq:vartheta-def}, and $\bf B, \bf T$ are defined as in Lemma \ref{lem:spectral-taylor-4} and \eqref{eq:vector-local-def}. Then we have from \eqref{eq:vector-local-bound},
\begin{equation}
    |\FB_{k,t}(g_1, g_2)| \,\le\, C\norm{\vartheta_{k,t}} \norm{g_1}\norm{g_2}, \quad |\FT_{k,t}(g_1, g_2, g_3)| \,\le\, C\norm{\vartheta_{k,t}} \norm{g_1}\norm{g_2}\norm{g_3}
\end{equation}
and also $\norm{\FB_{k,t}}_{\HS} \le C \sqrt{d} \norm{\vartheta_{k,t}}$. Therefore, by \eqref{eq:vartheta-bound} of Lemma \ref{lem:projected-z-representation} and H\"older, we have for any $m$,
\begin{equation}\label{eq:projected-local-B-moment}
    \norm{\FB_{k,t}(g_1, g_2)}_{L^m} \le C_m d^{-1/2}\norm{g_1}_{L^{3m}} \norm{g_2}_{L^{3m}}
\end{equation}
\begin{equation}\label{eq:projected-local-T-moment}
    \norm{\FT_{k,t}(g_1, g_2, g_3)}_{L^m} \le C_m d^{-1/2}\norm{g_1}_{L^{4m}} \norm{g_2}_{L^{4m}} \norm{g_3}_{L^{4m}}
\end{equation}
and also using \eqref{eq:scalarized-local-B-HS},
\begin{equation}\label{eq:projected-local-HS-bound}
    \norm{\FB_{k,t}}_{L^m(\HS)} \le C_m, \quad \norm{\FT_{k,t}}_{L^m(\HS)} \le C_m d^{1/2}
\end{equation}
Using the projected bilinear and trilinear forms in \eqref{eq:projected-BT-def}, we define the following decomposition of $\inner{\lambda_i}{z^t_i(u)}$ into quadratic, cubic, and remainder terms.
\begin{equation}\label{eq:I234-def}
    \begin{aligned}
        \I^{(2)}_t(u) &\bydef\, \sum_{k\,=\,0}^{t-1} \FB_{k,t}(m^k(u), m^k(u)), \quad \I^{(3)}_t(u) \,\bydef\, \sum_{k\,=\,0}^{t-1}\FT_{k,t}(m^k(u), m^k(u), m^k(u)) \\
        \I^{(4)}_t(u) & \bydef \sum_{k=0}^{t-1} \inner{\vartheta_{k,t}}{\rho^{(4)}_{i,k}(m^k(u))}
    \end{aligned}
\end{equation}
By \eqref{eq:epsilon-third-order-vector} and Lemma \ref{lem:projected-z-representation}, we have
\begin{equation}\label{eq:projected-z-I234}
    \inner{\lambda_i}{z^t_i(u)} = \I^{(2)}_t(u) + \I^{(3)}_t(u) + \I^{(4)}_t(u)
\end{equation}
We will now approximate the error $\varepsilon^k_i$ using the $\CG_i$ approximation $\bar m^t(u)$
\begin{lemma}\label{lem:epsilon-quad-lin}
    For $0 \le k \le T-1$, define
    \begin{equation}\label{eq:epsilon-quad-lin-def}
        \varepsilon^k_{i,\quadt, \lin}(u) \bydef \bf B^k(\bar m^k(u), \bar m^k(u))
    \end{equation}
    Then for every $m \ge 1$,
    \begin{equation}\label{eq:epsilon-quad-lin-size}
        \norm{\varepsilon^k_{i,\quadt,\lin}(u)}_{L^m} \le C_m d^{-1}
    \end{equation}
    \begin{equation}\label{eq:epsilon-quad-lin-error}
        \norm{\varepsilon^k_i(u) - \varepsilon^k_{i,\quadt,\lin}(u) }_{L^m} \le C_m d^{-3/2}
    \end{equation}
    uniformly for $u \in \{f_i, \widetilde f_i\}$
\end{lemma}

The proof is deferred to Section \ref{subsec:linear-comparison-deferred-proofs}.

For $0 \le \ell < k \le T$, define the following matrix that captures how the error propagates through the iterations
\begin{equation}\label{eq:U-propagator-def}
    \CU_{\ell,k} \bydef \Gamma_k M_i^{k-1} \cdots M_i^{\ell + 1} \Pi_1^\top : \R^p \to \R^p
\end{equation}
with the empty product taken to be the identity. Then we define
\begin{equation}\label{eq:zeta-quad-lin-def}
    \zeta^k_{\quadt, \lin}(u) \bydef \sum_{\ell=0}^{k-1} \CU_{\ell,k} \varepsilon^\ell_{i,\quadt,
    \lin}(u), \quad \tilde \zeta^k(u) \bydef \zeta^k(u) - \zeta^k_{\quadt, \lin}(u)
\end{equation}
with $\zeta^0_{\quadt,\lin}(u) = \tilde \zeta^0(u) = 0$ and recall $\zeta^k(u)$ from \eqref{eq:zeta-def}. We can then bound these terms as follows.

\begin{corollary}\label{cor:zeta-quad-lin}
    For every fixed $k\le T$ and every integer $m\ge 1$,
    \begin{equation}\label{eq:U-response-bound}
        \max_{0\,\le\, \ell\,<\,k}\norm{\CU_{\ell,k}}_{L^m(\op)}\,\le\, C_m,
    \end{equation}
    \begin{equation}\label{eq:zeta-quad-lin-size}
        \norm{\zeta_{\quadt,\lin}^k(u)}_{L^m}\le C_m d^{-1},
    \end{equation}
    \begin{equation}\label{eq:zeta-quad-lin-error}
        \norm{\tilde \zeta^k(u)}_{L^m}\le C_m d^{-3/2},
    \end{equation}
    uniformly in $u\in\{f_i,\tilde  f_i\}$.
\end{corollary}

\begin{proof}[Proof of Corollary \ref{cor:zeta-quad-lin}]
    The operator bound \eqref{eq:U-response-bound} follows from the definition of $\CU_{\ell,k}$ in \eqref{eq:U-propagator-def}, the definition of $\Gamma_k$ in \eqref{eq:Gamma-def}, the block matrices $M_i^r$ in \eqref{eq:Mt-def}, and the bounds in Lemma \ref{lem:sensing-envelope}, \eqref{eq:m-response-op-bound} of Lemma \ref{lem:m-representation}, and Lemma \ref{lem:response-op-bounds}. Then \eqref{eq:zeta-quad-lin-size} is immediate from \eqref{eq:epsilon-quad-lin-size} of Lemma \ref{lem:epsilon-quad-lin}. Finally,
    \begin{equation}
        \tilde\zeta^k(u)
        =\sum_{\ell=0}^{k-1} \CU_{\ell,k}\big(\varepsilon_i^\ell(u)-\varepsilon_{i,\quadt,\lin}^\ell(u)\big),
    \end{equation}
    so \eqref{eq:zeta-quad-lin-error} follows from \eqref{eq:epsilon-quad-lin-error} of Lemma \ref{lem:epsilon-quad-lin}.
\end{proof}

We can show that $\I^{(3)}_t(u), \I^{(4)}_t(u)$ are small.

\begin{proposition}\label{prop:projected-z-cubic-quartic}
    For every $t\le T$, the cubic and quartic contributions are small in both
    the GOE and rank-one cases. More precisely,
    \begin{equation}\label{eq:I34-individual-bounds}
        \max_{u\in\{f_i,\widetilde f_i\}}
        \left\{
            \abs{\E\!\left[\I_t^{(3)}(u)\right]} + \abs{\E\!\left[\I_t^{(4)}(u)\right]}
        \right\} \le C d^{-5/2}.
    \end{equation}
    Consequently, by the triangle inequality,
    \begin{equation}\label{eq:I34-comparison}
        \abs{\E\!\left[\I_t^{(3)}(f_i)-\I_t^{(3)}(\widetilde f_i)\right]} + \abs{\E\!\left[\I_t^{(4)}(f_i)-\I_t^{(4)}(\widetilde f_i)\right]} \le C d^{-5/2},
    \end{equation}
\end{proposition}

The proof is deferred to Section \ref{subsec:linear-comparison-deferred-proofs}.

It remains to compare the quadratic term $\I^{(2)}_t(u)$. We claim that the following two terms are the main contributors
\begin{equation}\label{eq:quadratic-main-terms}
    \begin{aligned}
        \I^{(2, \lin)}_t(u) &\bydef \sum_{k=0}^{t-1} \FB_{k,t}(\bar m^k(u), \bar m^k(u)) \\
        \I^{(\fb, \main)}_t(u) &\bydef 2 \sum_{k=0}^{t-1} \FB_{k,t}(\bar m^k(u), \zeta^k_{\quadt,\lin}(u))
    \end{aligned}
\end{equation}
where recall $\bar m^t$ from Lemma \ref{lem:m-representation}\eqref{eq:mbar-def}.
\begin{proposition}\label{prop:projected-z-quad-reduction}
    For every fixed $t \le T$, with $\I_t^{(2,\lin)}$ and $\I_t^{(\fb,\main)}$ as defined in \eqref{eq:quadratic-main-terms}, we have
    \begin{equation}
        \abs{
        \E[\I^{(2)}_t(f_i) - \I^{(2)}_t(\widetilde f_i)] - \E[\I^{(2,\lin)}_t(f_i) - \I^{(2,\lin)}_t(\widetilde f_i)] - \E[\I^{(\fb, \main)}_t(f_i) - \I^{(\fb, \main)}_t(\widetilde f_i)]
        } \le C d^{-5/2}
    \end{equation}
\end{proposition}

The proof is deferred to Section \ref{subsec:linear-comparison-deferred-proofs}.

We compare $\I^{(2,\lin)}_t$ and $\I^{(\fb, \main)}_t$ individually.

\begin{proposition}\label{prop:projected-z-quad-lin-comp}
    For every $t \le T$, we have
    \begin{equation}\label{eq:I2lin-comp}
        \abs{\E[\I^{(2,\lin)}_t(f_i) - \I^{(2,\lin)}_t(\widetilde f_i)]} \le C d^{-5/2}
    \end{equation}
\end{proposition}

The proof is deferred to Section \ref{subsec:linear-comparison-deferred-proofs}.

\begin{proposition}\label{prop:projected-z-feedback-main-comp}
    For every $t \le T$, we have
    \begin{equation}\label{eq:Ifbmain-comp}
        \abs{\E\left[\I^{(\fb,\main)}_t(f_i) - \I^{(\fb,\main)}_t(\widetilde f_i) \right]} \le Cd^{-5/2}
    \end{equation}
\end{proposition}

The proof is deferred to Section \ref{subsec:linear-comparison-deferred-proofs}. Now we can bound the last term
\begin{proposition}\label{prop:projected-z-complete}
    For any $t \le T$ and any $\lambda_i \in \R^p$ that is $\CG_i$-measurable and satisfies $\norm{\lambda_i} \le C d^{-1/2}$ a.s., we have
    \begin{equation}\label{eq:projected-z-proved}
        \abs{\E \inner{ \lambda_i}{z^t_i(\widetilde f_i)} - \E\inner{\lambda_i}{z^t_i(f_i)}} \le C d^{-5/2}
    \end{equation}
    That is, \eqref{eq:projected-z-needed} of Proposition \ref{prop:remaining-linear} holds
\end{proposition}
\begin{proof}[Proof of Proposition \ref{prop:projected-z-complete}]
    By \eqref{eq:projected-z-I234}, with $\I^{(2)}_t$, $\I^{(3)}_t$, and $\I^{(4)}_t$ defined in \eqref{eq:I234-def}, we have
    \begin{equation}
        \inner{\lambda_i}{z^t_i(u)} = \I^{(2)}_t(u) + \I^{(3)}_t(u) + \I^{(4)}_t(u)
    \end{equation}
    Proposition \ref{prop:projected-z-cubic-quartic} shows that the cubic and quartic pieces contribute $\CO(d^{-5/2})$. Proposition \ref{prop:projected-z-quad-reduction} reduces the quadratic term to $\I^{(2,\lin)}_t$ and $\I^{(\fb,\main)}_t$, and Proposition \ref{prop:projected-z-quad-lin-comp} and Proposition \ref{prop:projected-z-feedback-main-comp} bound the differences of these between GOE and rank-one models in $\CO(d^{-5/2})$. Therefore, we get \eqref{eq:projected-z-proved}
\end{proof}
\begin{proof}[Proof of Proposition \ref{prop:linear-comparison-complete}]
    This follows from Propositions~\ref{prop:remaining-linear} and \ref{prop:projected-z-complete}.
\end{proof}

\subsection{Deferred Proofs}\label{subsec:linear-comparison-deferred-proofs}

We now prove the statements whose proofs were deferred above, in the order in which they were stated. Here again, $u\in\{f_i,\widetilde f_i\}$ and $\CG_i$ are as in \eqref{eq:global-sensing-matrices} and \eqref{eq:loo-sigma-algebra-def}.

\begin{proof}[Proof of Proposition \ref{prop:leading-linear-barq}]
    Using \eqref{eq:bar-q-linear-combo}, with $\CB_{t,r}$ defined in \eqref{eq:Btr-def}, we have
    \begin{equation}
        \bar q^t(u) = \sum_{r=0}^{t-1} \CB_{t,r}u \ \beta^r_{i,0}(u)
    \end{equation}
    where each $\CB_{t,r}$ is $\CG_i$-measurable. We thus have
    \begin{equation}
        \inner{\lambda_i}{\bar q^t(u)} = \sum_{r=0}^{t-1} \inner{\CB_{t,r}^\top \lambda_i}{u} \beta^r_{i,0}(u)
    \end{equation}
    With $\beta^r_{i,0}$ as defined in \eqref{eq:beta0-def}, write
    \begin{equation}
        a^r_i \,\bydef\, {\bf s}^* - \widehat{\bf s}^r_\smi \,\in\, \R^p, \quad \beta^{r}_{i,0} \,=\, \kappa_r \inner{u}{a^r_i} +\nu_r\epsilon_i
    \end{equation}
    Conditioning on $\CG_i$, using Lemma \ref{lem:one-column-moments}, we have
    \begin{equation}
        \E\left[\inner{\CB_{t,r}^\top\lambda_i}{u}\beta^r_{i,0}(u) \ \big| \ \CG_i \right] = \frac{\kappa_r}{p} \inner{\CB_{t,r}^\top\lambda_i}{a^r_i}
    \end{equation}
    because the covariance matches as $I_p / p$ by \eqref{eq:column-covariance} and the term involving $\nu_r \epsilon_i$ vanishes as $u$ is centered. Notice that the RHS does not depend on $u$, so the conditional expectation is the same for $f_i, \widetilde f_i$. Summing over $r$ proves the claim.
\end{proof}

\begin{proof}[Proof of Proposition \ref{prop:remaining-linear}]
    By \eqref{eq:q-minus-barq-decomp}, with $\CP_{k,t}$, $d^k(u)$, and $z_i^t(u)$ as in \eqref{eq:response-operator-def}, \eqref{eq:d-def} of Proposition \ref{prop:bar-alpha}, and \eqref{eq:z-unrolled}, we have
    \begin{equation}
        q^t(u) - \bar q^t(u) = \sum_{k=0}^{t-1} \CP_{k,t}u \ d^k(u) + z^t_i(u)
    \end{equation}
    \begin{equation}
        \implies\, \inner{ \lambda_i}{q^t(u) - \bar q^t(u)} \,=\, \sum_{k\,=\,0}^{t - 1} \inner{\CP_{k,t}^\top \lambda_i}{u} d^k(u) + \inner{\lambda_i}{z^t_i(u)}
    \end{equation}
    Since we have $\norm{\lambda_i} \le C d^{-1/2}$, use Lemma \ref{lem:response-op-bounds} and Cauchy--Schwarz, we have
    \begin{equation}
        \norm{\CP_{k,t}^\top \lambda_i}_{L^{2m}} \le C_m d^{-1/2}
    \end{equation}
    From Corollary \ref{cor:one-column-moments}, we have
    \begin{equation}
       \norm{\inner{\CP^\top_{k,t}\lambda_i}{u}}_{L^{2m}} \le C_m d^{-3/2}
    \end{equation}
    So together with Proposition \ref{prop:bar-alpha}, we have
    \begin{equation}
        \E\abs{\inner{\CP_{k,t}^\top \lambda_i}{u}d^k(u)} \,\le\, \norm{\inner{\CP^\top_{k,t} \lambda_i}{u}}_{L^2} \norm{d^k(u)}_{L^2} \,\le\, C d^{-5/2}
    \end{equation}
    Since $t \le T$, summing over $k$ and using Proposition \ref{prop:leading-linear-barq}, we obtain \eqref{eq:linear-remainder-reduction}. If \eqref{eq:projected-z-needed} holds, then \eqref{eq:remaining-linear-conclusion} follows immediately from \eqref{eq:linear-remainder-reduction}.
\end{proof}

\begin{proof}[Proof of Lemma \ref{lem:m-representation}]
    Iterating the block recursion \eqref{eq:block-recursion}, with $\beta_i^k(u)$ and $e_i^k(u)$ defined in \eqref{eq:beta-exact} of Proposition \ref{prop:exact-loo} and \eqref{eq:block-error-def}, gives
    \begin{equation}
        \delta^t(u) \,=\, \sum_{k\,=\,0}^{t-1} M_i^{t-1} \cdots M_i^{k + 1} B^k_i u \beta^k_i(u) + \sum_{\ell\,=\,0}^{t-1} M_i^{t-1} \cdots M_i^{\ell + 1}e^\ell_i(u).
    \end{equation}
    Applying $\Gamma_t$ and adding $a_tu\beta_i^t(u)$ gives \eqref{eq:m-exact-unrolled}. The moment bounds for $K_i$ and $M_i^r$ in Lemmas~\ref{lem:sensing-envelope} and \ref{lem:response-op-bounds} give \eqref{eq:m-response-op-bound}.

    Since $e^\ell_i(u) = \Pi_1^\top \varepsilon^\ell_i(u)$ and each $\varepsilon^\ell_i(u) = \CO_{L^m}(d^{-1})$ by Proposition \ref{prop:delta-moments}, the same response-operator bound gives \eqref{eq:zeta-exact-bound}. Now we plug the expansion of $\bar \alpha^t(u)$ from \eqref{eq:bar-alpha-linear-combo} of Proposition \ref{prop:bar-alpha} into \eqref{eq:mbar-def}, with $\CM_{k,t}$ defined in \eqref{eq:m-response-def}, to get
    \begin{equation}
        \begin{aligned}
            \bar m^t(u) \,=\, \sum_{k\,=\,0}^t \sum_{r\,=\,0}^k \CM_{k,t} u c_{k,r} \beta^r_{i,0}(u) \,=\, \sum_{r \,=\, 0}^t \left( \sum_{k\,=\,r}^t c_{k,r} \CM_{k,t} \right) u \beta^r_{i,0}(u).
        \end{aligned}
    \end{equation}
    Therefore, \eqref{eq:mbar-linear-combo} holds with
    \begin{equation}\label{eq:Rtr-def}
        \CR_{t,r} \bydef \sum_{k=r}^t c_{k,r} \CM_{k,t},
    \end{equation}
    which is $\CG_i$-measurable. By subtracting from \eqref{eq:m-exact-unrolled}, we have
    \begin{equation}
        m^t(u) - \bar m^t(u) \,=\, \sum_{k\,=\,0}^t \CM_{k,t} u \left(\beta^k_i(u) - \bar \alpha^k(u) \right) + \zeta^t(u) \,=\, \zeta^t(u) + s^t(u)
    \end{equation}
    which gives \eqref{eq:m-minus-mbar-decomp}. From \eqref{eq:mbar-linear-combo}, Corollary \ref{cor:scalar-bounds}, Lemma \ref{lem:sensing-envelope}, and \eqref{eq:m-response-op-bound}, we get \eqref{eq:mbar-size}. Lastly, \eqref{eq:s-remainder-bound} follows from the definition of $s^t$ in \eqref{eq:s-remainder-def}, Proposition \ref{prop:bar-alpha}, Lemma \ref{lem:sensing-envelope}, and \eqref{eq:m-response-op-bound}.
\end{proof}

\begin{proof}[Proof of Lemma \ref{lem:epsilon-quad-lin}]
    From \eqref{eq:vector-local-bound} and \eqref{eq:mbar-size} of Lemma \ref{lem:m-representation}, we immediately have \eqref{eq:epsilon-quad-lin-size}. For \eqref{eq:epsilon-quad-lin-error}, we write
    \begin{equation}
        m^k(u) = \bar m^k(u) + r^k(u), \quad r^k(u) \bydef s^k(u) + \zeta^k(u)
    \end{equation}
    Then \eqref{eq:m-minus-mbar-decomp} of Lemma \ref{lem:m-representation}, \eqref{eq:s-remainder-bound} of Lemma \ref{lem:m-representation}, and \eqref{eq:zeta-exact-bound} of Lemma \ref{lem:m-representation} give
    \begin{equation}\label{eq:r-m-size}
        \norm{r^k(u)}_{L^m} \le C_m d^{-1}
    \end{equation}
    By \eqref{eq:epsilon-third-order-vector}, with $\bf B^k,\bf T^k$ and $\rho_{i,k}^{(4)}$ defined in \eqref{eq:vector-local-def} and \eqref{eq:vector-local-remainder-def}, we have
    \begin{equation}
        \begin{aligned}
            \varepsilon^k_i(u) - \varepsilon^k_{i,\quadt,\lin}(u) &= 2 \bf B^k(\bar m^k(u), r^k(u)) + \bf B^k(r^k(u), r^k(u)) \\
            \ &+ \bf T^k(m^k(u), m^k(u), m^k(u)) + \rho^{(4)}_{i,k}(m^k(u))
        \end{aligned}
    \end{equation}
    Now apply \eqref{eq:vector-local-bound}, \eqref{eq:mbar-size} of Lemma \ref{lem:m-representation}, \eqref{eq:r-m-size}, and Proposition \ref{prop:delta-moments}, we have
    \begin{equation}
        2\norm{\bf B^k(\bar m^k(u), r^k(u))}_{L^m} \,\le\, C_m d^{-3/2}, \quad \norm{\bf B^k(r^k(u), r^k(u))}_{L^m} \,\le\, C_m d^{-2}
    \end{equation}
    \begin{equation}
        \norm{\bf T^k(m^k(u), m^k(u), m^k(u))}_{L^m} \,\le\, C_m d^{-3/2},\quad \norm{\rho^{(4)}_{i,k}(m^k)}_{L^m} \,\le\, C_m d^{-2}
    \end{equation}
    So the conclusion follows.
\end{proof}

\begin{proof}[Proof of Proposition \ref{prop:projected-z-cubic-quartic}]
    With $\vartheta_{k,t}$ defined in \eqref{eq:vartheta-def}, for the quartic term we have from Proposition \ref{prop:delta-moments}, \eqref{eq:vector-local-bound}, and \eqref{eq:vartheta-bound} of Lemma \ref{lem:projected-z-representation} that
    \begin{equation}
        \E\abs{\I^{(4)}_t(u)} \,\le\, \sum_{k\,=\,0}^{t-1} \norm{\vartheta_{k,t}}_{L^2} \norm{\rho^{(4)}_{i,k}(m^k(u))}_{L^{2}} \,\le\, C d^{-1/2}d^{-2} \,\le\, Cd^{-5/2}
    \end{equation}
    uniformly in $u \in \{f_i, \widetilde f_i\}$. For the cubic term, we write
    \begin{equation}
        m^k(u) = \bar m^k(u) + r^k(u), \quad r^k(u) = s^k(u) + \zeta^k(u)
    \end{equation}
    so by \eqref{eq:zeta-exact-bound} of Lemma \ref{lem:m-representation}, \eqref{eq:mbar-size} of Lemma \ref{lem:m-representation}, and \eqref{eq:s-remainder-bound} of Lemma \ref{lem:m-representation}, we have
    \begin{equation}
        \norm{\bar m^k(u)}_{L^m} \le C_m d^{-1/2}, \quad \norm{r^k(u)}_{L^m} \le C_m d^{-1}
    \end{equation}
    Now, with $\FT_{k,t}$ defined in \eqref{eq:projected-BT-def}, we expand
    \begin{equation}
        \begin{aligned}
            &\FT_{k,t}(m^k(u), m^k(u), m^k(u)) - \FT_{k,t}(\bar m^k(u), \bar m^k(u), \bar m^k(u)) \\
            =& \FT_{k,t}(r^k, \bar m^k, \bar m^k) + \FT_{k,t}(\bar m^k, r^k, \bar m^k) + \FT_{k,t}(\bar m^k, \bar m^k, r^k) + \FT_{k,t}(r^k,r^k,\bar m^k) \\
            +& \FT_{k,t}(r^k, \bar m^k, r^k) + \FT_{k,t}(\bar m^k, r^k, r^k) + \FT_{k,t}(r^k,r^k, r^k)
        \end{aligned}
    \end{equation}
    where we have suppressed the argument $u$. Each term contains at least one factor of $r^k(u)$. Applying \eqref{eq:projected-local-T-moment} and \eqref{eq:vartheta-bound} gives
    \begin{equation}
        \begin{aligned}
            &\norm{\FT_{k,t}(m^k(u), m^k(u), m^k(u)) - \FT_{k,t}(\bar m^k(u), \bar m^k(u), \bar m^k(u))}_{L^1} \\
            &\le\, C d^{-1/2}\left( \norm{\bar m^k(u)}_{L^4}^2 \norm{r^k(u)}_{L^4} + \norm{\bar m^k(u)}_{L^4} \norm{r^k(u)}_{L^4}^2 + \norm{r^k(u)}_{L^4}^3 \right)  \,\le\, Cd^{-5/2}
        \end{aligned}
    \end{equation}
    Therefore, it remains to bound $\FT_{k,t}(\bar m^k, \bar m^k, \bar m^k)$. By \eqref{eq:mbar-linear-combo} of Lemma \ref{lem:m-representation}, with $\CR_{k,r}$ defined in \eqref{eq:Rtr-def}, we have $\bar m^k(u) = \sum_{r=0}^k \CR_{k,r} u \beta^r_{i,0}(u)$. So we can write
    \begin{equation}\label{eq:cubic-main-expand}
        \FT_{k,t} (\bar m^k, \bar m^k, \bar m^k) \,=\, \sum_{r,s,q \,=\, 0}^k \beta^r_{i,0}(u) \beta^s_{i,0}(u) \beta^q_{i,0}(u) \FC_{k,t,r,s,q}(u,u,u)
    \end{equation}
    where the trilinear form $\FC$ is defined by
    \begin{equation}
        \FC_{k,t,r,s,q}(u_1, u_2, u_3) \bydef \FT_{k,t}(\CR_{k,r} u_1, \CR_{k,s} u_2, \CR_{k,q} u_3)
    \end{equation}
    Since $k\le T$ is fixed, the sum contains only a fixed number of terms, so it suffices to bound each summand separately. We will decompose $\FC_{k,t,r,s,q}(u,u,u)$ into a linear term and a centered quadratic-form term, and then combine their $L^2$ bounds with $\beta^r_{i,0}(u)\beta^s_{i,0}(u)\beta^q_{i,0}(u)=\CO_{L^2}(d^{-3/2})$ using Cauchy--Schwarz. Using \eqref{eq:mbar-linear-combo} of Lemma \ref{lem:m-representation}, \eqref{eq:projected-local-HS-bound}, and Lemma \ref{lem:hs-composition}, we have
    \begin{equation}\label{eq:cubic-main-HS}
        \norm{\FC_{k,t;r,s,q}}_{L^m(\mathrm{HS})}\le C_m d^{1/2}.
    \end{equation}
    We apply Lemma \ref{lem:trilinear-decomp} to get
    \begin{equation}
        \FC_{k,t,r,s,q}(u,u,u) \,=\, \inner{b}{u} + \sum_{\nu\,=\,1}^p \inner{e_\nu}{u} \left( u^\top M_\nu u - \frac 1 p \Tr(M_\nu) \right)
    \end{equation}
    From \eqref{eq:cubic-main-HS}, Lemma \ref{lem:trilinear-decomp} also gets
    \begin{equation}\label{eq:b-M-sum-bound}
        \norm{b}_{L^m} \,\le\, C_m d^{-1/2}, \quad \norm{\sum_{\nu\,=\,1}^p \norm{M_\nu}_\Frob^2}_{L^m} \,\le\, C_m d
    \end{equation}
    For the $\inner{b}{u}$ term, we have from Corollary \ref{cor:one-column-moments}(i),
    \begin{equation}
        \norm{\inner{b}{u}}_{L^2} \le Cd^{-3/2}
    \end{equation}
    while Corollary \ref{cor:scalar-bounds} implies
    \begin{equation}
        \norm{\beta^r_{i,0}(u) \beta^s_{i,0}(u) \beta^q_{i,0}(u)}_{L^2} \le C d^{-3/2}
    \end{equation}
    Therefore, Cauchy--Schwarz implies
    \begin{equation}
        \E \abs{\beta^r_{i,0}(u)\beta^s_{i,0}(u) \beta^q_{i,0}(u) \inner{b}{u} } \le Cd^{-3}
    \end{equation}
    which is $o(d^{-5/2})$ and therefore $\CO(d^{-5/2})$. For the $\nu$-sum, define
    \begin{equation}
        Y(u) \bydef \sum_{\nu=1}^p \inner{e_\nu}{u} \left(u^\top M_\nu u - \frac 1 p \Tr(M_\nu) \right)
    \end{equation}
    By Cauchy--Schwarz, %
    \begin{equation}
        \begin{aligned}
            \abs{Y(u)} &= \abs{\sum_{\nu = 1}^p \inner{e_\nu}{u} \left( u^\top M_\nu u - \frac 1 p \Tr(M_\nu) \right)} \\
            &\le\, \left( \sum_{\nu \,=\, 1}^p \inner{e_\nu}{u}^2 \right)^{1/2} \left( \sum_{\nu \,=\, 1}^p \left| u^\top M_\nu u - \frac 1 p \Tr(M_\nu) \right|^2 \right)^{1/2} \\
            &= \norm{u} \left( \sum_{\nu = 1}^p \abs{u^\top M_\nu u - \frac 1 p \Tr(M_\nu)}^2 \right)^{1/2}
        \end{aligned}
    \end{equation}
    Write the conditional moment norm as
    \begin{equation}
        \norm{X}_{L^r(u | \CG_i)} \bydef \left( \E[|X|^r \ | \ \CG_i] \right)^{1/r}
    \end{equation}
    By conditional H\"older, we have
    \begin{equation}
        \norm{Y(u)}_{L^2(u | \CG_i)} \,\le\, \norm{u}_{L^4(u | \CG_i)} \norm{\left(  \sum_{\nu \,=\, 1}^p \abs{u^\top M_\nu u - \frac 1 p \Tr(M_\nu}^2 \right)^{1/2}}_{L^4(u | \CG_i)}
    \end{equation}
    Since $u$ is independent from $\CG_i$, we have from \ref{lem:sensing-envelope} a.s., $\norm{u}_{L^4(u | \CG_i)} = \norm{u}_{L^4} \le C$. For the second factor, we have by Minkowski,
    \begin{equation}
        \begin{aligned}
            &\norm{\left( \sum_{\nu \,=\, 1}^p \abs{u^\top M_\nu u - \frac 1 p \Tr(M_\nu)}^2 \right)^{1/2}}_{L^4(u | \CG_i)}^2 \,=\, \norm{\sum_{\nu \,=\, 1}^p \abs{u^\top M_\nu u - \frac 1 p \Tr(M_\nu)}^2}_{L^2(u | \CG_i)} \\
            \le\, &\sum_{\nu \,=\, 1}^p \norm{\left(u^\top M_\nu u - \frac 1 p \Tr(M_\nu)\right)^2}_{L^2(u | \CG_i)} \,=\, \sum_{i \,=\, 1}^p \norm{u^\top M_\nu u - \frac 1 p \Tr(M_\nu)}_{L^4(u | \CG_i)}^2
        \end{aligned}
    \end{equation}
    Lemma \ref{lem:refined-centered-quad} implies
    \begin{equation}
        \norm{u^\top M_\nu u - \frac 1 p \Tr(M_\nu)}_{L^4(u | \CG_i)} \le C d^{-3/2} \norm{M_\nu}_\Frob
    \end{equation}
    Therefore, almost surely, we have
    \begin{equation}
        \begin{aligned}
            \norm{Y(u)}_{L^2(u | \CG_i)} &\le\, C\left( \sum_{\nu \,=\, 1}^p \norm{u^\top M_\nu u - \frac 1 p \Tr(M_\nu)}_{L^4(u | \CG_i)}^2 \right)^{1/2} \,\le\, C d^{-3/2} \left( \sum_{\nu \,=\, 1}^p \norm{M_\nu}_\Frob^2 \right)^{1/2}
        \end{aligned}
    \end{equation}
    Taking expectations over $\CG_i$ gives
    \begin{equation}
        \begin{aligned}
            \norm{Y(u)}_{L^2} &=\, \norm{\norm{Y(u)}_{L^2(u | \CG_i)}}_{L^2} \,\le\, C d^{-3/2} \norm{\left( \sum_{\nu \,=\, 1}^p \norm{M_\nu}_\Frob^2 \right)^{1/2}}_{L^2} \,=\, Cd^{-3/2} \norm{\sum_{\nu \,=\, 1}^p \norm{M_\nu}_\Frob^2}_{L^1}^{1/2}
        \end{aligned}
    \end{equation}
    Using \eqref{eq:b-M-sum-bound} with $m = 1$ gives
    \begin{equation}\label{eq:Yu-L2-bound}
        \norm{Y(u)}_{L^2} \le C \cdot d^{-3/2} \cdot d^{1/2} \le Cd^{-1}
    \end{equation}
    Now since $\beta^r_{i,0}\beta^s_{i,0}\beta^q_{i,0} = \CO_{L^2}(d^{-3/2})$, we get
    \begin{equation}
        \E\abs{\beta^r_{i,0}(u) \beta^s_{i,0}(u) \beta^q_{i,0}(u) Y(u)} \le C d^{-5/2}
    \end{equation}
    Since the sum in \eqref{eq:cubic-main-expand} is over a fixed number of terms, we have $\CO(d^{-5/2})$ contribution, which proves \eqref{eq:I34-comparison}.
\end{proof}

\begin{proof}[Proof of Proposition \ref{prop:projected-z-quad-reduction}]
    The sums in \eqref{eq:I234-def} and \eqref{eq:quadratic-main-terms} contain a fixed number of terms, so we bound each term separately. Fix $k<t$ and abbreviate
    $$\bar m \bydef \bar m^k(u),\quad s = s^k(u),\quad z = \zeta^k_{\quadt,\lin}(u), \quad \tilde z = \tilde \zeta^k(u)
    $$
    Then we have $m^k(u) = \bar m + s + z + \tilde z$ by \eqref{eq:m-minus-mbar-decomp} of Lemma \ref{lem:m-representation} and \eqref{eq:zeta-quad-lin-def}. Expanding the bilinear form gives us
    \begin{equation}\label{eq:quad-reduction-expand}
            \begin{aligned}
            \FB_{k,t}(m^k(u),m^k(u))
            &=\FB_{k,t}(\bar m,\bar m)
            +2\FB_{k,t}(\bar m,z)
            +2\FB_{k,t}(\bar m,s)
            +2\FB_{k,t}(\bar m,\tilde z)\\
            &+\FB_{k,t}(s,s)
            +2\FB_{k,t}(s,z)
            +2\FB_{k,t}(s,\tilde z)\\
            &+\FB_{k,t}(z,z)
            +2\FB_{k,t}(z,\tilde z)
            +\FB_{k,t}(\tilde z,\tilde z).
        \end{aligned}
    \end{equation}
    We will show that only the first two terms on the RHS matter. By \eqref{eq:mbar-size} of Lemma \ref{lem:m-representation}, \eqref{eq:s-remainder-bound} of Lemma \ref{lem:m-representation}, \eqref{eq:zeta-quad-lin-size}, and \eqref{eq:zeta-quad-lin-error}, we have
    \begin{equation}
        \norm{\bar m}_{L^m} \,\le\, C_m d^{-1/2}, \quad \norm{s}_{L^m} + \norm{z}_{L^m} \,\le\, C_m d^{-1}, \quad \norm{\tilde z}_{L^m} \,\le\, C_m d^{-3/2}
    \end{equation}
    Observe that every term in \eqref{eq:quad-reduction-expand} that contains any two factors among $\{s, z, \tilde z\}$, possibly repeated, has $L^1$-norm bounded by $\CO(d^{-5/2})$ by \eqref{eq:projected-local-B-moment}. Terms with $\bar m$ mixed with $\tilde z$ also satisfy $\CO(d^{-5/2})$ bound. In particular these terms are bounded by H\"older
    \begin{equation}
        C d^{-1/2}\times
        \begin{cases}
            \norm{\bar m}_{L^3}\norm{\tilde z}_{L^3},\\
            \norm{s}_{L^3}\norm{s}_{L^3},\\
            \norm{s}_{L^3}\norm{z}_{L^3},\\
            \norm{z}_{L^3}\norm{z}_{L^3},
        \end{cases} = \CO(d^{-5/2})
    \end{equation}
    Therefore, it remains only to control $\FB_{k,t}(\bar m,s)$. Recall the formula for $\bar m^t, s^t$ from \eqref{eq:mbar-linear-combo} of Lemma \ref{lem:m-representation} and \eqref{eq:s-remainder-def} of Lemma \ref{lem:m-representation}, with $\beta^r_{i,0}(u)$, $d^a(u)$, and $\CM_{a,k}$ defined in \eqref{eq:beta0-def}, \eqref{eq:d-def} of Proposition \ref{prop:bar-alpha}, and \eqref{eq:m-response-def}; we have
    \begin{equation}
        \FB_{k,t}(\bar m, s) = \sum_{r,a=0}^k \beta^r_{i,0}(u) d^a(u) \FQ_{k,t,r,a}(u,u)
    \end{equation}
    where
    \begin{equation}
        \FQ_{k,t,r,a}(u_1, u_2) \bydef \FB_{k,t}(\CR_{k,r}u_1, \CM_{a,k}u_2)
    \end{equation}
    By \eqref{eq:projected-local-HS-bound}, \eqref{eq:mbar-linear-combo} of Lemma \ref{lem:m-representation}, \eqref{eq:m-response-op-bound} of Lemma \ref{lem:m-representation}, and Lemma \ref{lem:hs-composition}, we get
    \begin{equation}
        \norm{\FQ_{k,t,r,a}}_{L^m(\HS)} \le C_m
    \end{equation}
    Let $M_{k,t,r,a}$ be the matrix representation of $\FQ$ as in Lemma \ref{lem:bilinear-matrix}, so that $\norm{M_{k,t,r,a}}_{L^m(\Frob)} \le C_m$. Now we fix the $(k,t,r,a)$ subscript and suppress it. Write
    \begin{equation}
        \FQ(u,u) = \frac 1 p \Tr(M) + \left(u^\top M u - \frac 1 p \Tr(M) \right)
    \end{equation}
    We have from Corollary \ref{cor:scalar-bounds} and Proposition \ref{prop:bar-alpha}(ii), and Cauchy--Schwarz
    \begin{equation}
        \norm{\beta^r_{i,0}(u) d^a(u)}_{L^2} \,\le\, Cd^{-3/2}, \quad \norm{\frac 1 p \Tr(M)}_{L^2} \,\le\, p^{-1/2}\norm{M}_{L^2(\Frob)} \,\le\, C d^{-1}
    \end{equation}
    We see that the trace part contributes $\CO(d^{-5/2})$ to $\FB_{k,t}(\bar m, s)$ and thus is negligible. But by Lemma \ref{lem:refined-centered-quad}, we have
    \begin{equation}
        \norm{u^\top M u - \frac 1 p \Tr(M)}_{L^2} \le Cd^{-3/2}
    \end{equation}
    So the centered quadratic form contributes $\CO(d^{-3})$. Therefore, we have
    \begin{equation}
        \E \abs{\FB_{k,t}(\bar m, s)} \le C d^{-5/2}
    \end{equation}
\end{proof}

\begin{proof}[Proof of Proposition \ref{prop:projected-z-quad-lin-comp}]
    By \eqref{eq:mbar-linear-combo} of Lemma \ref{lem:m-representation}, with $\CR_{k,r}$ defined in \eqref{eq:Rtr-def}, we have
    \begin{equation}
        \I^{(2,\lin)}_t(u) \,=\, \sum_{k\,=\,0}^{t-1}\sum_{r,s\,=\,0}^k \beta^r_{i,0}(u)\beta^s_{i,0}(u) \FQ_{k,t,r,s}(u,u)
    \end{equation}
    with bilinear form $\FQ_{k,t,r,s}(u_1, u_2) \bydef \FB_{k,t}(\CR_{k,r} u_1, \CR_{k,s} u_2) $. By \eqref{eq:projected-local-HS-bound}, \eqref{eq:mbar-linear-combo} of Lemma \ref{lem:m-representation}, and Lemma \ref{lem:hs-composition}, we have that $\norm{\FQ_{k,t,r,s}}_{L^m}\le C_m$. Now let $M_{k,t,r,s}$ be the symmetric matrix representation of $\FQ_{k,t,r,s}$ from Lemma \ref{lem:bilinear-matrix} Then we have
    \begin{equation}
        \norm{M_{k,t,r,s}}_{L^m(\Frob)} \le \norm{\FQ_{k,t,r,s}}_{L^m(\HS)} \le C_m
    \end{equation}
    Write
    \begin{equation}
        \FQ_{k,t,r,s}(u,u) \,=\, \frac 1 p \Tr(M_{k,t,r,s}) + \tilde m_{k,t,r,s}(u), \quad \tilde m_{k,t,r,s}(u) \,\bydef\, \left(u^\top M_{k,t,r,s} u - \frac 1 p \Tr(M_{k,t,r,s}) \right)
    \end{equation}
    By Lemma \ref{lem:refined-centered-quad}, we have $\norm{\tilde m_{k,t,r,s}(u)}_{L^2} \le C d^{-3/2}$. We thus write
    \begin{equation}
        \I^{(2,\lin)}_t(u) = \sf T^\lin(u) + \sf R^\lin(u)
    \end{equation}
    where
    \begin{equation}
        \begin{aligned}
            \sf T^\lin(u) &\bydef\, \sum_{k\,=\,0}^{t-1} \sum_{r,s\,=\,0}^k \frac 1p \Tr(M_{k,t,r,s}) \beta^r_{i,0}\beta^s_{i,0}(u), \quad\sf R^\lin(u) \,\bydef\, \sum_{k\,=\,0}^{t-1} \sum_{r,s\,=\,0}^k \beta^r_{i,0}(u) \beta^s_{i,0}(u) \tilde m_{k,t,r,s}(u)
        \end{aligned}
    \end{equation}
    Recall from \eqref{eq:betar-betal-condition-on-Gi} that the value of $\E[\beta^r_{i,0}(u) \beta^s_{i,0}(u) \ | \ \CG_i]$ does not depend on whether $u$ is $f_i$ or $\widetilde f_i$. Therefore, we have
    \begin{equation}
        \E[\sf T^\lin(f_i) - \sf T^\lin(\widetilde f_i) \ | \ \CG_i] = 0
    \end{equation}
    For $\sf R^\lin$, we have from Corollary \ref{cor:scalar-bounds}
    \begin{equation}
        \norm{\beta^r_{i,0}(u) \beta^s_{i,0}(u)}_{L^2} \le C d^{-1}
    \end{equation}
    so we have
    \begin{equation}
        \E \abs{\sf R^\lin(u)} \le C d^{-1} d^{-3/2} \le C d^{-5/2}
    \end{equation}
    uniformly in $u \in \{f_i, \widetilde f_i\}$. This proves \eqref{eq:I2lin-comp}.
\end{proof}

\begin{proof}[Proof of Proposition \ref{prop:projected-z-feedback-main-comp}]
    Fix $k < t$. By the definitions \eqref{eq:zeta-quad-lin-def}, \eqref{eq:epsilon-quad-lin-def} of Lemma \ref{lem:epsilon-quad-lin}, \eqref{eq:U-propagator-def}, and \eqref{eq:mbar-linear-combo} of Lemma \ref{lem:m-representation}, with $\bf B^\ell$ defined in \eqref{eq:vector-local-def}, we have
    \begin{equation}
        \zeta^k_{\quadt,\lin}(u) \,=\, \sum_{\ell\,=\,0}^{k-1} \CU_{\ell, k} \varepsilon^\ell_{i, \quadt, \lin}(u) \,=\, \sum_{\ell \,=\, 0}^{k-1} \sum_{a,b\,=\,0}^\ell \beta^a_{i,0}(u)\beta^b_{i,0}(u) \CU_{\ell, k}\bf B^\ell(\CR_{\ell, a} u, \CR_{\ell, b} u)
    \end{equation}
    where $\varepsilon^\ell_{i,\quadt, \lin}$ are defined in Lemma \ref{lem:epsilon-quad-lin}. Therefore,
    \begin{equation}\label{eq:Ifbmain-expand}
        \I^{(\fb,\main)}_t(u) \,=\, 2 \sum_{k\,=\,0}^{t-1} \sum_{\ell \,=\, 0}^{k - 1} \sum_{r \,=\, 0}^k \sum_{a,b\,=\,0}^\ell \beta^r_{i,0}(u) \beta^a_{i,0}(u) \beta^b_{i,0}(u) \FC_{t,k,\ell,r,a,b}(u,u,u)
    \end{equation}
    where
    \begin{equation}
        \FC_{t,k,\ell,r,a,b}(u_1,u_2,u_3) \bydef \FB_{k,t}(\CR_{k,r}u_1, \CU_{\ell, k}\bf B^\ell(\CR_{\ell, a}u_2, \CR_{\ell, b}u_3))
    \end{equation}
    where $\FB_{k,t}$ is the bilinear form in \eqref{eq:projected-BT-def}. We claim that $\norm{\FC_{t,k,\ell,r,a,b}}_{L^m(\HS)} \le C_m d^{1/2}$. To see this, define bilinear form on $\R^p$
    \begin{equation}
        \FA_{t,k,\ell,r}(x,y) \bydef \FB_{k,t}(\CR_{k,r} x, \CU_{\ell, k}y)
    \end{equation}
    By \eqref{eq:scalarized-local-B-HS}, the bilinear map $\bf B^\ell$ satisfies \eqref{eq:hs-compose-hyp} with $K=C\sqrt d$. The composition bound \eqref{eq:hs-compose-trilinear} therefore gives
    \begin{equation}
        \norm{\FC_{t,k,\ell,r,a,b}}_\HS \le C \sqrt{d} \norm{\FB_{k,t}}_\HS \norm{\CR_{k,r}}_\op \norm{\CU_{\ell,k}}_\op \norm{\CR_{\ell,a}}_\op \norm{\CR_{\ell,b}}_\op
    \end{equation}
    So by H\"older, \eqref{eq:projected-local-HS-bound}, \eqref{eq:mbar-linear-combo} of Lemma \ref{lem:m-representation}, and \eqref{eq:U-response-bound}, we have
    \begin{equation}
        \norm{\FC_{t,k,\ell,r,a,b}}_{L^m(\HS)} \le C_m d^{1/2}
    \end{equation}
    So by H\"older and \eqref{eq:mbar-linear-combo} of Lemma \ref{lem:m-representation}, we have
    \begin{equation}
        \norm{\FC}_{L^m(\HS)} \,\le\, C \sqrt{d} \norm{\FA}_{L^{3m}(\HS)} \norm{\CR_{\ell, a}}_{L^{3m}(\op)} \norm{\CR_{\ell, b}}_{L^{3m}(\op)} \,\le\, C d^{1/2}
    \end{equation}
    Applying Lemma~\ref{lem:trilinear-decomp} to each $\FC$ gives
    \begin{equation}
        \FC(u,u,u) = \inner{b}{u} + \sum_{\nu = 1}^p \inner{e_\nu}{u} \left(u^\top M_\nu u - \frac 1p \Tr M_\nu \right)
    \end{equation}
    Using the fact we just proved that $\norm{\FC_{t,k,\ell,r,a,b}}_{L^m(\HS)} \le C_m d^{1/2}$, we have from the same Lemma \ref{lem:trilinear-decomp}
    \begin{equation}
        \norm{b}_{L^m} \,\le\, C_m d^{-1/2}, \quad \norm{\sum_{\nu \,=\, 1}^p \norm{M_\nu}_\Frob^2}_{L^m} \,\le\, C_m d
    \end{equation}
    Therefore, we apply Corollary \ref{cor:one-column-moments} and get $\norm{\inner{b}{u}}_{L^2} \le C d^{-3/2}$ and also similar to the proof of Proposition \ref{prop:projected-z-cubic-quartic}, in line \eqref{eq:Yu-L2-bound},
    \begin{equation}
        \norm{\sum_{\nu = 1}^p \inner{e_\nu}{u}\left(u^\top M_\nu u - \frac 1 p \Tr(M_\nu) \right)}_{L^2} \le C d^{-1}
    \end{equation}
    Since the product of the scalar $\beta$'s satisfies
    \begin{equation}
        \norm{\beta^r_{i,0}(u) \beta^a_{i,0}(u) \beta^b_{i,0}(u)}_{L^2} \le C d^{-3/2}
    \end{equation}
    we see that every term in \eqref{eq:Ifbmain-expand} has expectation of order $\CO(d^{-5/2})$, we sum over finitely many of them and get \eqref{eq:Ifbmain-comp}.
\end{proof}

\section{Concentration of AMP Observables}\label{sec:concentration}
We bound the fluctuations of each observable around its conditional expectation given $\CB_0 = \sigma(S^\star, \hS^0)$. Consider a sensing matrix
\begin{equation}
    F^\sfe = [u_1, \ldots, u_n] \in \R^{p \times n}
\end{equation}
and run the vectorized recursion in Section~\ref{sec:proof-main}. Denote the iterates by $\hbs_\sfe^t$ and $\hS^t_\sfe = \mmat(\hbs^t_\sfe)$. Here $\hbs^t_\smi$ denotes the leave-one-out iterate within this ensemble, whereas in the Lindeberg comparison it denoted the leave-one-out iterate of the interpolating matrix.

\begin{assumption}\label{ass:concentration-columns}
    The vectors $u_1, \ldots, u_n \in \R^p$ are independent and independent of $(S^\star, \hS^0, \epsilon)$, and satisfy
    \begin{equation}\label{eq:concentration-isotropy}
        \E u_i \,=\, 0, \qquad \E u_i u_i^\top \,=\, \frac{I_p}{p}, \qquad i \,\in\, [n].
    \end{equation}
    Further, there exists $C < \infty$ such that
    \begin{equation}\label{eq:concentration-subexp}
        \sup_{i \,\le\, n} \sup_{a \,\in\, \BS^{p - 1}} \norm{\sqrt{p}\, \inner{u_i}{a}}_{\psi_1} \,\le\, C.
    \end{equation}
    Let
    \begin{equation}\label{eq:concentration-envelope-def}
        K_i^\sfe \bydef \sum_{j \ne i} u_j u_j^\top.
    \end{equation}
    Assume also that, for every integer $m \ge 1$, there exists $C_m < \infty$ independent of $d$ and $i$ such that
    \begin{equation}\label{eq:concentration-envelope-moment}
        \norm{K_i^\sfe}_{L^m(\op)} \le C_m,
    \end{equation}
    and, for every deterministic $A \in \R^{p \times p}$,
    \begin{equation}\label{eq:concentration-rotational}
        \norm{u_i^\top A u_i - \frac{1}{p} \Tr A}_{L^m} \,\le\, C_m d^{-1/2} \norm{A_s}_\op, \qquad A_s \,\bydef\, \frac{A + A^\top}{2}.
    \end{equation}
\end{assumption}

\begin{lemma}\label{lem:moment-efron-stein}
    Let $Y \in \R^r$ and $X_i \in \R^{k_i}$, $1 \le i \le N$, be mutually independent. Let $X_1', \ldots, X_N'$ be independent copies, independent of $(Y, X_1, \ldots, X_N)$. Let
    \begin{equation}
        Z \,=\, \Phi(Y, X_1, \ldots, X_N) \,\in\, L^q, \qquad q \,\ge\, 2,
    \end{equation}
    and set
    \begin{equation}
        Z_i' = \Phi(Y, X_1, \ldots, X_{i - 1}, X_i', X_{i + 1}, \ldots, X_N).
    \end{equation}
    Define
    \begin{equation}\label{eq:conditional-replacement-variance}
        V^\pm \bydef \E\left[\left.\sum_{i = 1}^N (Z - Z_i')_\pm^2\,\right|\,Y, X_1, \ldots, X_N\right].
    \end{equation}
    For some fixed $\kappa < \infty$, we have
    \begin{equation}\label{eq:moment-efron-stein-one-sided}
        \norm{(Z - \E[Z \mid Y])_\pm}_{L^q} \le \sqrt{2\kappa q}\, \norm{V^\pm}_{L^{q/2}}^{1/2}.
    \end{equation}
    Consequently,
    \begin{equation}\label{eq:moment-efron-stein-increments}
        \norm{Z - \E[Z \mid Y]}_{L^q} \,\le\, 2^{1/q}\sqrt{2\kappa q}\left(\sum_{i \,=\, 1}^N \norm{Z - Z_i'}_{L^q}^2\right)^{1/2}.
    \end{equation}
\end{lemma}

\begin{proof}[Proof of Lemma \ref{lem:moment-efron-stein}]
    Because the joint law is a product law, Fubini's theorem implies that, for almost every $y$,
    \begin{equation}
        Z_y \bydef \Phi(y, X_1, \ldots, X_N)
    \end{equation}
    belongs to $L^q$ under the law of $(X_1, \ldots, X_N)$. For such $y$, the second inequality in each one-sided bound of \cite[Theorem 2]{BoucheronStephane2005MIfF} gives
    \begin{equation}
        \norm{(Z_y - \E_X Z_y)_\pm}_{L^q(X)} \le \sqrt{2\kappa q}\, \norm{V_y^\pm}_{L^{q/2}(X)}^{1/2},
    \end{equation}
    where $V_y^\pm$ is the section of \eqref{eq:conditional-replacement-variance} at $Y = y$ and $L^q(X)$ denotes the space where we take expectation over the randomness of $X = (X_1, \ldots, X_N)$. Raising this inequality to the $q$th power, integrating with respect to $y$, and using Fubini's theorem gives \eqref{eq:moment-efron-stein-one-sided}.

    Since $q/2 \ge 1$, conditional Jensen's inequality and Minkowski's inequality give
    \begin{equation}
        \begin{aligned}
            \norm{V^\pm}_{L^{q/2}} &\le \sum_{i = 1}^N \norm{\E[(Z - Z_i')_\pm^2 \mid Y, X_1, \ldots, X_N]}_{L^{q/2}} \\
            &\le \sum_{i = 1}^N \norm{Z - Z_i'}_{L^q}^2.
        \end{aligned}
    \end{equation}
    Finally,
    \begin{equation}
        \norm{Z - \E[Z \mid Y]}_{L^q}^q = \norm{(Z - \E[Z \mid Y])_+}_{L^q}^q + \norm{(Z - \E[Z \mid Y])_-}_{L^q}^q.
    \end{equation}
    Applying \eqref{eq:moment-efron-stein-one-sided} to both signs and using the preceding bound on $V^\pm$ proves \eqref{eq:moment-efron-stein-increments}.
\end{proof}

Fix $i \in [n]$. Let $F^\sfe_\smi$ be obtained from $F^\sfe$ by setting its $i$th column to $0$. For $u \in \R^p$ and $\eta \in \R$, define
\begin{equation}
    F_i^\sfe(u) \,\bydef\, F^\sfe_\smi + u e_i^\top, \qquad \epsilon^{(i,\eta)} \,\bydef\, \sum_{j \ne i} \epsilon_j e_j + \eta e_i.
\end{equation}
Let $(\hbs_i^t(u,\eta), h_i^t(u,\eta))$ be the iterates generated by $(F_i^\sfe(u), \epsilon^{(i,\eta)})$, and let $(\hbs_\smi^t, h_\smi^t)$ be the iterates generated by $(F_\smi^\sfe, \epsilon)$. These leave-one-out iterates are unchanged if $\epsilon$ is replaced by $\epsilon^{(i,0)}$, since $F_\smi^\sfe\epsilon = F_\smi^\sfe\epsilon^{(i,0)}$. Set
\begin{equation}
    q_i^t(u,\eta) \,\bydef\, \hbs_i^t(u,\eta) - \hbs_\smi^t, \qquad \CH_i^\sfe \,\bydef\, \sigma(\CB_0, (u_j, \epsilon_j)_{j \ne i}).
\end{equation}
The leave-one-out iterates and $K_i^\sfe$ are $\CH_i^\sfe$-measurable.

\begin{lemma}\label{lem:one-vector-sensitivity}
    Under Assumptions \ref{ass:basic}, \ref{ass:signal-init}, and \ref{ass:concentration-columns}, let $u \stackrel{d}{=} u_i$ and $\eta\sim\CN(0,\Delta)$ be independent of each other and jointly independent of $\CH_i^\sfe$. For every integer $m\geq1$, uniformly in $i$ and $0\leq s,t\leq T$, we have
    \begin{equation}\label{eq:concentration-q-moment}
        \norm{q_i^t(u,\eta)}_{L^m} \le C_m d^{-1/2}.
    \end{equation}
    If $\lambda_i \in \R^p$ is $\CH_i^\sfe$-measurable and $\norm{\lambda_i} \le L d^{-1/2}$ a.s., then
    \begin{equation}\label{eq:concentration-projected-linear}
        \norm{\inner{\lambda_i}{q_i^t(u,\eta)}}_{L^m} \le C_{m,L} d^{-3/2}.
    \end{equation}
    If $H_i \in \R^{p \times p}$ is $\CH_i^\sfe$-measurable and $\norm{H_i}_\op \le L d^{-1}$ a.s., then
    \begin{equation}\label{eq:concentration-projected-quadratic}
        \norm{\inner{q_i^t(u,\eta)}{H_i q_i^s(u,\eta)}}_{L^m} \le C_{m,L} d^{-2}.
    \end{equation}
\end{lemma}

\begin{remark}
    The proof is essentially the same as Proposition \ref{prop:exact-loo} and Section \ref{sec:linearized}, and is deferred to later in this section.
\end{remark}

\begin{theorem}\label{thm:concentration-general-ensemble}
    Under Assumptions \ref{ass:basic}, \ref{ass:signal-init}, and \ref{ass:concentration-columns}, fix $0\leq s,t\leq T$ and $\mu\in\SC^3(\R)$. Define
    \begin{equation}\label{eq:concentration-observables}
        \Phi_t^\sig \,\bydef\, \frac{1}{d} \Tr(\hS_\sfe^t S^\star), \qquad \Phi_{t,\mu}^\spec \,\bydef\, \frac{1}{d} \Tr(\mu(\hS_\sfe^t)), \qquad \Phi_{s,t}^\cross \,\bydef\, \frac{1}{d} \Tr(\hS_\sfe^s \hS_\sfe^t).
    \end{equation}
    For every integer $m \ge 2$,
    \begin{equation}\label{eq:concentration-observables-moment}
        \max_{\Phi \,\in\, \{\Phi_t^\sig,\Phi_{t,\mu}^\spec,\Phi_{s,t}^\cross\}} \norm{\Phi - \E[\Phi \mid \CB_0]}_{L^m} \,\le\, C_{m,\mu} d^{-1/2}.
    \end{equation}
    Consequently,
    \begin{equation}\label{eq:concentration-observables-stochastic}
        \Phi - \E[\Phi \mid \CB_0] = \CO_\prec(d^{-1/2})
    \end{equation}
    for every observable in \eqref{eq:concentration-observables}.
\end{theorem}

\begin{proof}[Proof of Theorem \ref{thm:concentration-general-ensemble}]
    Let
    \begin{equation}
        Y \,\bydef\, (\svecop(S^\star), \svecop(\hS^0)) \,\in\, \R^{2p}, \qquad X_i \,\bydef\, (u_i, \epsilon_i).
    \end{equation}
    The random vectors $Y, X_1, \ldots, X_n$ are mutually independent. Let $X_i' = (u_i', \epsilon_i')$ be an independent copy of $X_i$. For any one of the observables in \eqref{eq:concentration-observables}, write $\Phi_i(u,\eta)$ for the observable generated by $(F_i^\sfe(u), \epsilon^{(i,\eta)})$, and write $\Phi_\smi$ for the observable generated by $(F_\smi^\sfe, \epsilon^{(i,0)})$. Then
    \begin{equation}\label{eq:concentration-signal-increment}
        \Phi_{t,i}^\sig(u,\eta) - \Phi_{t,\smi}^\sig = \inner{d^{-1}{\bf s}^*}{q_i^t(u,\eta)}.
    \end{equation}
    The cross-time correlation satisfies
    \begin{equation}\label{eq:concentration-cross-increment}
        \Phi_{s,t,i}^\cross(u,\eta) - \Phi_{s,t,\smi}^\cross = \inner{d^{-1}\hbs_\smi^s}{q_i^t(u,\eta)} + \inner{d^{-1}\hbs_\smi^t}{q_i^s(u,\eta)} + \inner{q_i^s(u,\eta)}{d^{-1}q_i^t(u,\eta)}.
    \end{equation}
    To expand the spectral statistic, use Lemma \ref{lem:base-det-bounds} to localize Lemma \ref{lem:smooth-spectral-expansion} to the deterministic spectral interval and set
    \begin{equation}
        A_i \bydef \hS_\smi^t, \qquad E_i(u,\eta) \bydef \mmat(q_i^t(u,\eta)).
    \end{equation}
    Then
    \begin{equation}\label{eq:concentration-spectral-increment}
        \Phi_{t,\mu,i}^\spec(u,\eta) - \Phi_{t,\mu,\smi}^\spec = \inner{d^{-1}\svecop(\mu'(A_i))}{q_i^t(u,\eta)} + \inner{q_i^t(u,\eta)}{K_{\mu,A_i}q_i^t(u,\eta)} + R_{\mu,A_i}(E_i(u,\eta)),
    \end{equation}
    where
    \begin{equation}
        d^{-1}\norm{\mu'(A_i)}_\Frob \,\le\, \norm{\mu'}_{L^\infty(I_\ast)}d^{-1/2}, \qquad \norm{K_{\mu,A_i}}_\op \,\le\, C_\mu d^{-1},
    \end{equation}
    and, by \eqref{eq:smooth-spectral-remainder-bound} of Lemma \ref{lem:smooth-spectral-expansion} and \eqref{eq:concentration-q-moment} of Lemma \ref{lem:one-vector-sensitivity},
    \begin{equation}\label{eq:concentration-spectral-remainder}
        \norm{R_{\mu,A_i}(E_i(u,\eta))}_{L^m} \le C_\mu d^{-1}\norm{q_i^t(u,\eta)}_{L^{3m}}^3 \le C_{m,\mu} d^{-5/2}.
    \end{equation}

    By Assumption \ref{ass:signal-init} and Lemma \ref{lem:base-det-bounds},
    \begin{equation}
        \norm{d^{-1}{\bf s}^*} + \max_{r \le T}\norm{d^{-1}\hbs_\smi^r} \le C d^{-1/2}
    \end{equation}
    almost surely. Therefore, applying Lemma \ref{lem:one-vector-sensitivity} to \eqref{eq:concentration-signal-increment}--\eqref{eq:concentration-spectral-increment}, with $H_i = d^{-1}I_p$ for the last term in \eqref{eq:concentration-cross-increment}, gives
    \begin{equation}\label{eq:concentration-one-block-increment}
        \max_{i \,\le\, n}\max_{\Phi \,\in\, \{\Phi_t^\sig,\Phi_{t,\mu}^\spec,\Phi_{s,t}^\cross\}} \norm{\Phi_i(u,\eta) - \Phi_\smi}_{L^m} \,\le\, C_{m,\mu} d^{-3/2}.
    \end{equation}
    Here the quadratic terms are in fact $\CO_{L^m}(d^{-2})$, and the spectral remainder is $\CO_{L^m}(d^{-5/2})$.

    Write \eqref{eq:concentration-one-block-increment} once with $(u,\eta) = (u_i,\epsilon_i)$ and once with $(u,\eta) = (u_i',\epsilon_i')$. Since both pairs are independent of $\CH_i^\sfe$, the triangle inequality gives
    \begin{equation}\label{eq:concentration-resampling-increment}
        \max_{i \,\le\, n}\max_{\Phi \,\in\, \{\Phi_t^\sig,\Phi_{t,\mu}^\spec,\Phi_{s,t}^\cross\}} \norm{\Phi - \Phi^{(i)}}_{L^m} \,\le\, C_{m,\mu} d^{-3/2},
    \end{equation}
    where $\Phi^{(i)}$ is generated with $X_i'$ in place of $X_i$.

    The observables in \eqref{eq:concentration-observables} are bounded uniformly in $d$: this follows from Assumption \ref{ass:signal-init}, Lemma \ref{lem:base-det-bounds}, and the bound of $\mu$ on the deterministic spectral interval. Hence they belong to $L^m$ for every $m$. Applying \eqref{eq:moment-efron-stein-increments} of Lemma \ref{lem:moment-efron-stein} to $Y, X_1, \ldots, X_n$, and using $n\leq C_0p$ and $p\asymp d^2$, gives
    \begin{equation}
        \begin{aligned}
            \norm{\Phi - \E[\Phi \mid \CB_0]}_{L^m} &\le\, 2^{1/m}\sqrt{2\kappa m}\left(\sum_{i \,=\, 1}^n \norm{\Phi - \Phi^{(i)}}_{L^m}^2\right)^{1/2} \\
            &\le C_{m,\mu}\sqrt{n}\, d^{-3/2} \le C_{m,\mu} d^{-1/2}.
        \end{aligned}
    \end{equation}
    This proves \eqref{eq:concentration-observables-moment}. For $\veps, D > 0$, choose an integer $m \ge 2$ such that $m\veps > D$. Markov's inequality gives
    \begin{equation}
        \P\left(\abs{\Phi - \E[\Phi \mid \CB_0]} \,>\, d^{-1/2 + \veps}\right) \,\le\, C_{m,\mu}^m d^{-m\veps} \,\le\, d^{-D}
    \end{equation}
    for all sufficiently large $d$. This proves \eqref{eq:concentration-observables-stochastic}.
\end{proof}

\begin{corollary}\label{cor:concentration-goe-rank-one}
    Under Assumptions \ref{ass:basic} and \ref{ass:signal-init}, Theorem \ref{thm:concentration-general-ensemble} holds for both $F^\sfe = F$ and $F^\sfe = \widetilde F$.
\end{corollary}

\begin{remark}
    The proof is almost identical to that of Lemma \ref{lem:sensing-envelope} and is deferred to later in this section.
\end{remark}

\begin{corollary}\label{cor:concentration-consequences}
    Under the hypotheses of Corollary \ref{cor:concentration-goe-rank-one}, for every fixed integer $k \ge 1$, every $0 \le t \le T$, and $\sf E \in \{\ro, \mathrm{G}\}$,
    \begin{equation}
        \frac{1}{d}\Tr((\hS_{\sf E}^t)^k) - \E\left[\left.\frac{1}{d}\Tr((\hS_{\sf E}^t)^k)\,\right|\,\CB_0\right] = \CO_\prec(d^{-1/2}),
    \end{equation}
    and
    \begin{equation}
        \frac{1}{d}\norm{\hS_{\sf E}^t - S^\star}_\Frob^2 - \E\left[\left.\frac{1}{d}\norm{\hS_{\sf E}^t - S^\star}_\Frob^2\,\right|\,\CB_0\right] = \CO_\prec(d^{-1/2}).
    \end{equation}
\end{corollary}

\begin{proof}[Proof of Corollary \ref{cor:concentration-consequences}]
    By Lemma \ref{lem:base-det-bounds}, the cutoff $\mu_{k,\ast}$ from the proof of Corollary \ref{cor:power-from-spectral} satisfies $\mu_{k,\ast}(\hS_{\sf E}^t) = (\hS_{\sf E}^t)^k$. The first claim therefore follows from Theorem \ref{thm:concentration-general-ensemble}. For the MSE,
    \begin{equation}
        \frac{1}{d}\norm{\hS_{\sf E}^t - S^\star}_\Frob^2 = \Phi_{t,t}^\cross - 2\Phi_t^\sig + \frac{1}{d}\norm{S^\star}_\Frob^2.
    \end{equation}
    Subtracting conditional expectations, the last term cancels and gives
    \begin{equation}
        \frac{1}{d}\norm{\hS_{\sf E}^t - S^\star}_\Frob^2 - \E\left[\left.\frac{1}{d}\norm{\hS_{\sf E}^t - S^\star}_\Frob^2\,\right|\,\CB_0\right] = \Phi_{t,t}^\cross - \E[\Phi_{t,t}^\cross \mid \CB_0] - 2(\Phi_t^\sig - \E[\Phi_t^\sig \mid \CB_0]).
    \end{equation}
    The second claim now follows from Theorem \ref{thm:concentration-general-ensemble}.
\end{proof}

\subsection{Deferred Proofs}
\begin{proof}[Proof of Lemma \ref{lem:one-vector-sensitivity}]
    Conditioning on $\CH_i^\sfe$, \eqref{eq:concentration-subexp} gives, for every $\CH_i^\sfe$-measurable $a \in \R^p$,
    \begin{equation}\label{eq:concentration-linear-moment}
        \norm{\inner{u}{a}}_{L^m} \le C_m p^{-1/2}\norm{a}_{L^m}.
    \end{equation}
    Indeed, $\norm{X}_{L^m} \le (2m!)^{1/m}\norm{X}_{\psi_1}$. For $m \ge 2$, Minkowski's inequality and \eqref{eq:concentration-subexp} also give
    \begin{equation}
        \norm{u}_{L^m}^2 \,=\, \norm{\norm{u}^2}_{L^{m/2}} \,\le\, \sum_{j \,=\, 1}^p \norm{\inner{u}{e_j}}_{L^m}^2 \,\le\, C_m,
    \end{equation}
    and the same bound for $1 \le m < 2$ follows by monotonicity of the moment norms. The same conditioning extends \eqref{eq:concentration-rotational} to every $\CH_i^\sfe$-measurable matrix $A$:
    \begin{equation}\label{eq:concentration-random-rotational}
        \norm{u^\top A u - \frac{1}{p}\Tr A}_{L^m} \le C_m d^{-1/2}\norm{A_s}_{L^m(\op)}.
    \end{equation}

    Define, for $0 \le t \le T - 1$,
    \begin{equation}
        w_i^t(u,\eta) \bydef h_i^t(u,\eta) - h_\smi^t, \qquad m_i^t(u,\eta) \bydef q_i^t(u,\eta) + a_t w_i^t(u,\eta).
    \end{equation}
    The same calculation as in Proposition \ref{prop:exact-loo} gives
    \begin{equation}\label{eq:concentration-exact-loo}
        w_i^t(u,\eta) \,=\, -\kappa_t K_i^\sfe q_i^t(u,\eta) + \rho_t w_i^{t - 1}(u,\eta) + u\beta_i^t(u,\eta), \qquad \beta_i^t(u,\eta) \,=\, \kappa_t\inner{u}{{\bf s}^* - \hbs_i^t(u,\eta)} + \nu_t\eta.
    \end{equation}
    Let
    \begin{equation}
        y_i^t \bydef \hbs_\smi^t + a_t h_\smi^t, \qquad L_i^t \bydef DG_t(y_i^t).
    \end{equation}
    Then
    \begin{equation}\label{eq:concentration-q-taylor}
        q_i^{t + 1}(u,\eta) \,=\, L_i^t m_i^t(u,\eta) + \veps_i^t(u,\eta), \qquad \norm{L_i^t}_\op \,\le\, C, \qquad \norm{\veps_i^t(u,\eta)} \,\le\, C\norm{m_i^t(u,\eta)}^2,
    \end{equation}
    where
    \begin{equation}
        \veps_i^t(u,\eta) \bydef G_t(y_i^t + m_i^t(u,\eta)) - G_t(y_i^t) - DG_t(y_i^t)[m_i^t(u,\eta)].
    \end{equation}
    We suppress the dependence on $(u,\eta)$ below. By Assumption \ref{ass:signal-init}, Lemma \ref{lem:base-det-bounds}, \eqref{eq:concentration-linear-moment}, and $p \asymp d^2$,
    \begin{equation}\label{eq:concentration-beta-zero}
        \norm{\kappa_t\inner{u}{{\bf s}^* - \hbs_\smi^t} + \nu_t\eta}_{L^m} \le C_m d^{-1/2}.
    \end{equation}
    Since
    \begin{equation}
        \beta_i^t = \kappa_t\inner{u}{{\bf s}^* - \hbs_\smi^t} - \kappa_t\inner{u}{q_i^t} + \nu_t\eta,
    \end{equation}
    H\"older's inequality gives
    \begin{equation}
        \norm{\beta_i^t}_{L^m} \le C_m d^{-1/2} + C_m\norm{u}_{L^{2m}}\norm{q_i^t}_{L^{2m}},
    \end{equation}
    and
    \begin{equation}
        \norm{w_i^t}_{L^m} \le C_m\norm{K_i^\sfe}_{L^{2m}(\op)}\norm{q_i^t}_{L^{2m}} + C_m\norm{w_i^{t - 1}}_{L^m} + C_m\norm{u}_{L^{2m}}\norm{\beta_i^t}_{L^{2m}}.
    \end{equation}
    Furthermore,
    \begin{equation}
        \norm{q_i^{t + 1}}_{L^m} \,\le\, C\norm{m_i^t}_{L^m} + C\norm{m_i^t}_{L^{2m}}^2, \qquad \norm{m_i^t}_{L^m} \,\le\, \norm{q_i^t}_{L^m} + C\norm{w_i^t}_{L^m}.
    \end{equation}
    Starting from $q_i^0 = 0$ and $w_i^{-1} = 0$, we apply these inequalities simultaneously at every moment order. Induction on $t$, using \eqref{eq:concentration-envelope-moment}, gives
    \begin{equation}\label{eq:concentration-basic-sensitivity}
        \max_{0 \,\le\, t \,\le\, T}\norm{q_i^t}_{L^m} + \max_{0 \,\le\, t \,\le\, T - 1}\left(\norm{w_i^t}_{L^m} + \norm{m_i^t}_{L^m} + \norm{\beta_i^t}_{L^m}\right) \,\le\, C_m d^{-1/2},
    \end{equation}
    and hence
    \begin{equation}\label{eq:concentration-epsilon-size}
        \max_{0 \,\le\, t \,\le\, T - 1}\norm{\veps_i^t}_{L^m} \,\le\, C_m d^{-1}.
    \end{equation}
    This proves \eqref{eq:concentration-q-moment}.

    To prove \eqref{eq:concentration-projected-linear}, define, for $0 \le t \le T - 1$,
    \begin{equation}
        \delta_i^t \,\bydef\, \begin{pmatrix} q_i^t \\ w_i^{t - 1} \end{pmatrix}, \qquad M_i^t \,\bydef\, \begin{pmatrix} L_i^t(I_p - a_t \kappa_t K_i^\sfe) & a_t \rho_t L_i^t \\ -\kappa_t K_i^\sfe & \rho_t I_p \end{pmatrix}, \qquad B_i^t \,\bydef\, \begin{pmatrix} a_t L_i^t \\ I_p \end{pmatrix}.
    \end{equation}
    Equations \eqref{eq:concentration-exact-loo} and \eqref{eq:concentration-q-taylor} are equivalently
    \begin{equation}\label{eq:concentration-block-recursion}
        \delta_i^{t + 1} = M_i^t \delta_i^t + B_i^t u \beta_i^t + \begin{pmatrix} \veps_i^t \\ 0 \end{pmatrix}.
    \end{equation}
    Let $\Pi_1 : \R^{2p} \to \R^p$ be the projection onto the first block. For $0 \le k < t \le T$, define
    \begin{equation}
        P_{k,t}^{(i)} \bydef \Pi_1 M_i^{t - 1}\cdots M_i^{k + 1}B_i^k,
    \end{equation}
    and
    \begin{equation}
        z_i^t(u,\eta) \,\bydef\, \Pi_1\sum_{\ell \,=\, 0}^{t - 1} M_i^{t - 1}\cdots M_i^{\ell + 1}\begin{pmatrix} \veps_i^\ell(u,\eta) \\ 0 \end{pmatrix}, \qquad z_i^0 \,\bydef\, 0.
    \end{equation}
    Empty products are understood as the identity. Expanding \eqref{eq:concentration-exact-loo} and \eqref{eq:concentration-q-taylor} gives
    \begin{equation}\label{eq:concentration-unrolled}
        q_i^t(u,\eta) = \sum_{k = 0}^{t - 1} P_{k,t}^{(i)}u\beta_i^k(u,\eta) + z_i^t(u,\eta).
    \end{equation}
    The operators $M_i^t$, $B_i^t$, and $P_{k,t}^{(i)}$ are $\CH_i^\sfe$-measurable. By \eqref{eq:concentration-envelope-moment} and \eqref{eq:concentration-epsilon-size}, H\"older's inequality gives
    \begin{equation}\label{eq:concentration-response-bounds}
        \max_{0 \,\le\, k \,<\, t \,\le\, T}\norm{P_{k,t}^{(i)}}_{L^m(\op)} \,\le\, C_m, \qquad \max_{0 \,\le\, t \,\le\, T}\norm{z_i^t(u,\eta)}_{L^m} \,\le\, C_m d^{-1}.
    \end{equation}
    Set
    \begin{equation}
        \beta_{i,0}^t(u,\eta) \,\bydef\, \kappa_t\inner{u}{{\bf s}^* - \hbs_\smi^t} + \nu_t\eta, \qquad \Omega_{k,t} \,\bydef\, \kappa_t u^\top P_{k,t}^{(i)}u, \qquad \eta_i^t \,\bydef\, \kappa_t\inner{u}{z_i^t}.
    \end{equation}
    Then
    \begin{equation}\label{eq:concentration-beta-recursion}
        \beta_i^t = \beta_{i,0}^t - \sum_{k < t}\Omega_{k,t}\beta_i^k - \eta_i^t.
    \end{equation}
    Define
    \begin{equation}
        \overline\Omega_{k,t} \bydef \frac{\kappa_t}{p}\Tr P_{k,t}^{(i)}.
    \end{equation}
    By \eqref{eq:concentration-random-rotational}, \eqref{eq:concentration-response-bounds}, and H\"older's inequality,
    \begin{equation}\label{eq:concentration-kernel-bounds}
        \norm{\Omega_{k,t} - \overline\Omega_{k,t}}_{L^m} \,\le\, C_m d^{-1/2}, \qquad \norm{\overline\Omega_{k,t}}_{L^m} \,\le\, C_m, \qquad \norm{\eta_i^t}_{L^m} \,\le\, C_m d^{-1}.
    \end{equation}
    Define
    \begin{equation}
        \overline\alpha_i^t \bydef \beta_{i,0}^t - \sum_{k < t}\overline\Omega_{k,t}\overline\alpha_i^k.
    \end{equation}
    Since $T$ is fixed, there exist $\CH_i^\sfe$-measurable coefficients $c_{t,r}$ such that
    \begin{equation}\label{eq:concentration-alpha-expansion}
        \overline\alpha_i^t \,=\, \sum_{r \,=\, 0}^t c_{t,r}\beta_{i,0}^r, \qquad \max_{r \,\le\, t}\norm{c_{t,r}}_{L^m} \,\le\, C_m.
    \end{equation}
    With $d_i^t \bydef \beta_i^t - \overline\alpha_i^t$, subtracting the last recursion from \eqref{eq:concentration-beta-recursion} gives
    \begin{equation}
        d_i^t \,=\, -\sum_{k \,<\, t}\overline\Omega_{k,t}d_i^k - \sum_{k \,<\, t}(\Omega_{k,t} - \overline\Omega_{k,t})\beta_i^k - \eta_i^t.
    \end{equation}
    Applying this recursion simultaneously at every moment order, \eqref{eq:concentration-basic-sensitivity} and \eqref{eq:concentration-kernel-bounds} give by induction
    \begin{equation}\label{eq:concentration-amplitude-error}
        \max_{0 \,\le\, t \,\le\, T - 1}\norm{d_i^t}_{L^m} \,\le\, C_m d^{-1}.
    \end{equation}
    Define
    \begin{equation}
        B_{t,r}^{(i)} \,\bydef\, \sum_{k \,=\, r}^{t - 1} c_{k,r}P_{k,t}^{(i)}, \qquad e_i^t \,\bydef\, z_i^t + \sum_{k \,=\, 0}^{t - 1} P_{k,t}^{(i)}u d_i^k.
    \end{equation}
    From \eqref{eq:concentration-unrolled} and \eqref{eq:concentration-alpha-expansion},
    \begin{equation}\label{eq:concentration-q-expansion}
        q_i^t(u,\eta) = \sum_{r = 0}^{t - 1} B_{t,r}^{(i)}u\beta_{i,0}^r(u,\eta) + e_i^t(u,\eta),
    \end{equation}
    and \eqref{eq:concentration-response-bounds}, \eqref{eq:concentration-amplitude-error}, and H\"older's inequality give
    \begin{equation}\label{eq:concentration-q-expansion-bounds}
        \max_{0 \,\le\, r \,<\, t \,\le\, T}\norm{B_{t,r}^{(i)}}_{L^m(\op)} \,\le\, C_m, \qquad \max_{0 \,\le\, t \,\le\, T}\norm{e_i^t(u,\eta)}_{L^m} \,\le\, C_m d^{-1}.
    \end{equation}
    For $t = 0$, \eqref{eq:concentration-projected-linear} is immediate. For $1 \le t \le T$, \eqref{eq:concentration-beta-zero}, \eqref{eq:concentration-linear-moment}, H\"older's inequality, and \eqref{eq:concentration-q-expansion} give
    \begin{equation}
        \begin{aligned}
            \norm{\inner{\lambda_i}{q_i^t(u,\eta)}}_{L^m} &\le\, L d^{-1/2}\norm{e_i^t(u,\eta)}_{L^m} + \sum_{r \,<\, t}\norm{\inner{u}{(B_{t,r}^{(i)})^\top\lambda_i}}_{L^{2m}}\norm{\beta_{i,0}^r(u,\eta)}_{L^{2m}} \\
            &\le C_{m,L} d^{-3/2} + C_{m,L} d^{-2} \le C_{m,L} d^{-3/2}.
        \end{aligned}
    \end{equation}
    This proves \eqref{eq:concentration-projected-linear}. Finally, \eqref{eq:concentration-q-moment} and H\"older's inequality give
    \begin{equation}
        \norm{\inner{q_i^t(u,\eta)}{H_i q_i^s(u,\eta)}}_{L^m} \,\le\, L d^{-1}\norm{q_i^t(u,\eta)}_{L^{2m}}\norm{q_i^s(u,\eta)}_{L^{2m}} \,\le\, C_{m,L} d^{-2},
    \end{equation}
    which proves \eqref{eq:concentration-projected-quadratic}.
\end{proof}

\begin{proof}[Proof of Corollary \ref{cor:concentration-goe-rank-one}]
    Independence and isotropy are given by \eqref{eq:column-covariance}. Fix $a \in \BS^{p - 1}$ and let $M = \mmat(a)$, so $\norm{M}_\Frob = 1$. For the GOE column,
    \begin{equation}
        \sqrt{p}\,\inner{\widetilde f_i}{a} \sim \CN(0,1).
    \end{equation}
    For the rank-one column,
    \begin{equation}
        \sqrt{p}\,\inner{f_i}{a} = \frac{x_i^\top Mx_i - \Tr M}{\sqrt{2}}.
    \end{equation}
    If $(\lambda_j)_{j = 1}^d$ are the eigenvalues of $M$, then $\max_j\abs{\lambda_j} \le 1$. For $\abs{t} \le 1/4$,
    \begin{equation}
        \log \E\exp\left(t(x_i^\top Mx_i - \Tr M)\right) \,=\, \sum_{j \,=\, 1}^d\left(-t\lambda_j - \frac{1}{2}\log(1 - 2t\lambda_j)\right) \,\le\, 2t^2\sum_{j \,=\, 1}^d\lambda_j^2 \,=\, 2t^2.
    \end{equation}
    Here we used $-z - \log(1 - z) \le z^2$ for $\abs{z} \le 1/2$. Chernoff's bound, with $t = s/4$ for $0 \le s \le 1$ and $t = 1/4$ for $s \ge 1$, gives
    \begin{equation}\label{eq:rank-one-linear-tail}
        \P\left(\abs{x_i^\top Mx_i - \Tr M} > s\right) \le 2\exp\left(-\frac{1}{8}\min\{s^2,s\}\right).
    \end{equation}
    Integrating \eqref{eq:rank-one-linear-tail} gives $\norm{x_i^\top Mx_i - \Tr M}_{\psi_1} \le C$. Therefore, \eqref{eq:concentration-subexp} holds for both ensembles.

    We verify \eqref{eq:concentration-envelope-moment}. The first calculation in the proof of Lemma \ref{lem:sensing-envelope} gives
    \begin{equation}
        p\norm{\widetilde f_i}^2 \sim \chi_p^2, \qquad \norm{f_i}^2 = \frac{\norm{x_i}^4 - 2\norm{x_i}^2 + d}{d(d + 1)},
    \end{equation}
    and hence $\norm{u_i}_{L^m} \le C_m$ in both cases. Let
    \begin{equation}
        V_i \,\bydef\, (u_j)_{j \ne i} \,\in\, \R^{p \times n-1}, \qquad X_j \,\bydef\, \sqrt{p}\,u_j, \qquad N \,\bydef\, n - 1.
    \end{equation}
    The vectors $X_j$ are centered and isotropic. For $L \ge 2$, Hanson--Wright and the preceding norm formulas give
    \begin{equation}
        \P(\norm{\widetilde f_j} \,>\, L) \,\le\, e^{-c_Lp}, \qquad \P(\norm{f_j} \,>\, L) \,\le\, \P(\chi_d^2 \,>\, Ld/2) \,\le\, e^{-c_Ld},
    \end{equation}
    where $c_L \to \infty$ as $L \to \infty$. Since $N/p$ stays bounded above, a union bound and a sufficiently large fixed $L$ give
    \begin{equation}
        \P\left(\max_{j \ne i}\frac{\norm{X_j}}{\sqrt{p}} \,>\, L\max\left\{1,\left(\frac{N}{p}\right)^{1/4}\right\}\right) \,\le\, e^{-\sqrt{p}}
    \end{equation}
    for all sufficiently large $d$, uniformly in $i$. Thus, when $N \ge p$, Theorem 1 of \cite{1012.0294} and isotropy give
    \begin{equation}
        N^{-1}\norm{(X_j)_{j \ne i}}_\op^2 \le 1 + C\sqrt{p/N},
    \end{equation}
    while, when $N < p$, Remark 2 following that theorem gives $\norm{(X_j)_{j \ne i}}_\op \le C\sqrt{p}$. Since $V_i = p^{-1/2}[X_j]_{j \ne i}$ and $K_i^\sfe = V_iV_i^\top$, both cases give
    \begin{equation}\label{eq:pure-ensemble-operator-tail}
        \P\left(\norm{K_i^\sfe}_\op > C\right) \le 2e^{-c\sqrt{p}}.
    \end{equation}
    Moreover,
    \begin{equation}
        \norm{K_i^\sfe}_\op \,\le\, \Tr K_i^\sfe \,=\, \sum_{j \ne i}\norm{u_j}^2, \qquad \norm{\Tr K_i^\sfe}_{L^{2m}} \,\le\, \sum_{j \ne i}\norm{u_j}_{L^{4m}}^2 \,\le\, C_m n \,\le\, C_m d^2.
    \end{equation}
    If $\CE_i = \{\norm{K_i^\sfe}_\op \le C\}$, then Cauchy--Schwarz and \eqref{eq:pure-ensemble-operator-tail} give
    \begin{equation}
        \E\norm{K_i^\sfe}_\op^m \,\le\, C^m + \E[(\Tr K_i^\sfe)^m\mathbf{1}_{\CE_i^c}] \,\le\, C^m + \norm{\Tr K_i^\sfe}_{L^{2m}}^m\P(\CE_i^c)^{1/2} \,\le\, C_m.
    \end{equation}
    This proves \eqref{eq:concentration-envelope-moment}. Finally, \eqref{eq:quadratic-moment-lm} of Lemma \ref{lem:one-column-moments}, applied to deterministic $A$, gives \eqref{eq:concentration-rotational}. Therefore, Assumption \ref{ass:concentration-columns} holds for both sensing ensembles.
\end{proof}

\section{Discussion and Future Directions}\label{sec:discussion}

For a fixed number of AMP iterations, the signal overlaps, cross-time overlaps, and smooth spectral statistics agree under rank-one and covariance-matched Gaussian sensing up to order $d^{-1/2}$ in every fixed moment. Any Gaussian state-evolution limit for these observables therefore transfers to rank-one sensing. The proof follows the effect of replacing a sensing column through the AMP recursion, using leave-one-out estimates and smooth spectral expansions.

Extending the comparison to nonlinear observations requires care when the signal has a spectral spike. Assumption~\ref{ass:signal-init} permits such a spike: under the normalization \eqref{eq:rank-one-features}, an order-one quadratic single-index component
\begin{equation*}
    a\big((u^\top x_i)^2-1\big),\qquad\|u\|=1,
\end{equation*}
corresponds to adding $a\sqrt d\,uu^\top$ to $S^\star$, which still satisfies the Frobenius bound \eqref{eq:signal-init-bound}. Theorem~\ref{thm:final-universality} therefore covers such spikes in the linear observation model. Along the spike direction, however, rank-one sensing retains the centered $\chi_1^2$ variable $a(G^2-1)$, with $G=u^\top x_i\sim\CN(0,1)$, whereas the covariance-matched Gaussian projection is Gaussian with the same variance. A nonlinear output channel can distinguish these two laws before the AMP iteration begins, as illustrated in Figure~\ref{fig:numerical-nonlinear-trajectory}. This suggests that a comparison beyond the linear model may need to preserve the low-dimensional non-Gaussian component and Gaussianize only its orthogonal complement, as in the conditional Gaussian comparison for random-feature models in \cite{WenHuLuFanMisiakiewicz2025}. Proving such a result would require a replacement argument conditioned on the spike coordinates, together with control of their repeated interaction with both the output function and the spectral denoiser.

Another extension concerns unbounded or data-dependent spectral denoisers. The proof uses bounded fourth-order Fr\'echet derivatives of the spectral map and uniform bounds on the matrix iterates. It therefore does not cover unbounded denoisers such as the identity, spectral maps that depend on the empirical eigenvalue distribution, or Onsager coefficients estimated from the evolving sensing data. In particular, the rotationally invariant estimator and the full Bayes-optimal AMP recursion of \cite{MaillardAntoine2024Bloa} lie outside the theorem, although the experiments in Section~\ref{subsec:numerics-spectra} suggest that the rank-one and Gaussian dynamics remain close for the rotationally invariant estimator. Extending the comparison to this setting would require stability estimates that replace the present global derivative bounds and control the dependence of both the denoiser and the Onsager correction on the evolving empirical spectrum.

Our bounds depend on the fixed iteration horizon $T$. To let $T$ grow with $d$, one would need to control how the one-column errors accumulate over time. Comparing the full trajectories would also require estimates for test functions beyond the three observables studied here. Both questions remain open.

\par\addvspace{3.5ex plus 1ex minus .2ex}
\noindent\begin{minipage}{\textwidth}
\section*{Acknowledgement}
\addcontentsline{toc}{section}{Acknowledgement}

Duy Thuc Nguyen gratefully acknowledges support from the Harvard Generative AI Research Program Grant and the Herchel Smith Undergraduate Science Research Program. Yue M. Lu gratefully acknowledges support from a Harvard College Professorship, the Harvard FAS Dean's Fund for Promising Scholarship, and DARPA grant DIAL-FP-038. Subhabrata Sen gratefully acknowledges support from NSF grant DMS CAREER 2239234, ONR grant N00014-23-1-2489, and AFOSR grant FA9950-23-1-0429.

\medskip
GPT-5.6 and OpenAI Codex were used to assist in polishing parts of the exposition and in implementing the numerical simulations. All mathematical results and proofs were developed and verified by the authors. The authors reviewed all AI-assisted material and take full responsibility for the results and the content of this paper.
\end{minipage}

\addtocontents{toc}{\protect\setcounter{tocdepth}{1}}
\appendix

\section{Analytic Bounds}\label{app:analytic}

This appendix collects the deterministic, analytic bounds used above.

For an $m$-times continuously differentiable function $f$, we write $f^{[m]}$ for the $m$th divided difference. Thus $f^{[0]}=f$, and for $m\ge 1$,
\begin{equation}
    f^{[m]}(x_0,\dots,x_m)
    \bydef
    \begin{cases}
        \dfrac{f^{[m-1]}(x_0,\dots,x_{m-1})-f^{[m-1]}(x_1,\dots,x_m)}{x_0-x_m},
        & x_0\neq x_m,\\
        \text{ continuous extension},
        & x_0=x_m.
    \end{cases}
\end{equation}
If $f\in C^m(\R)$, then $f^{[m]}$ extends to a continuous symmetric function on $\R^{m+1}$, and the Hermite--Genocchi formula yields
\begin{equation}\label{eq:divided-difference-bound}
    \sup_{(x_0,\dots,x_m)\in\R^{m+1}}\abs{f^{[m]}(x_0,\dots,x_m)}
    \le \frac{1}{m!}\sup_{x\in\R}\abs{f^{(m)}(x)}.
\end{equation}

\begin{lemma}\label{lem:dk-bounds}
    Let $m\in\{1,2,3,4\}$ and let $\varphi\in C^m(\R)$. Define
    $\Sigma(M)\bydef \varphi(M)$ on $\Sym$ and
    \begin{equation}
        G(v)\bydef \svecop\big(\Sigma(\mat(v))\big).
    \end{equation}
    Then $\Sigma$ and $G$ are $m$-times Fr\'echet-differentiable. If
    $X=\sum_{a=1}^d \lambda_a P_a$ is an orthogonal spectral decomposition of
    $X\in\Sym$ into rank-one spectral projectors (with repeated eigenvalues
    listed with multiplicity), then
    \begin{equation}\label{eq:dk-formula}
        D^m\Sigma(X)[H_1,\dots,H_m]
        =
        \sum_{\pi\in S_m}
        \sum_{i_0,\dots,i_m=1}^d
        \varphi^{[m]}(\lambda_{i_0},\dots,\lambda_{i_m})
        P_{i_0}H_{\pi(1)}P_{i_1}\cdots H_{\pi(m)}P_{i_m}.
    \end{equation}
    If in addition $\norm{\varphi^{(m)}}_\infty<\infty$, then
    \begin{equation}\label{eq:dk-frobenius-bound}
        \norm{D^m\Sigma(X)[H_1,\dots,H_m]}_{\Frob}
        \le \norm{\varphi^{(m)}}_\infty \prod_{r=1}^m \norm{H_r}_{\Frob},
    \end{equation}
    and therefore
    \begin{equation}\label{eq:dk-vector-bound}
        \norm{D^m G(x)[h_1,\dots,h_m]}
        \le \norm{\varphi^{(m)}}_\infty \prod_{r=1}^m \norm{h_r}.
    \end{equation}
\end{lemma}

\begin{lemma}\label{lem:spectral-taylor}
    Let $\varphi\in C_b^2(\R)$, i.e.
    \begin{equation}
        \norm{\varphi'}_\infty+\norm{\varphi''}_\infty<\infty.
    \end{equation}
    Define $G(v)=\svecop(\varphi(\mat(v)))$. Then $G:\R^p\to\R^p$ is
    Fr\'echet differentiable. Moreover, for every $x,h\in\R^p$,
    \begin{equation}\label{eq:G-lipschitz}
        \norm{DG(x)}_{\op}
        \le \norm{\varphi'}_\infty,
    \end{equation}
    \begin{equation}\label{eq:G-second}
        \norm{G(x+h)-G(x)-DG(x)h}
        \le \frac12 \norm{\varphi''}_\infty\norm{h}^2.
    \end{equation}
\end{lemma}

\begin{proof}[Proof of Lemma \ref{lem:spectral-taylor}]
    By Lemma \ref{lem:dk-bounds} with $m=1,2$, the map $G$ is $C^2$ and
    satisfies
    \begin{equation}
        \sup_{z\in\R^p}\norm{DG(z)}_{\op}\le \norm{\varphi'}_\infty,
        \qquad
        \sup_{z\in\R^p}\norm{D^2G(z)}_{\op}\le \norm{\varphi''}_\infty.
    \end{equation}
    This proves \eqref{eq:G-lipschitz}. Taylor's formula with integral
    remainder in the finite-dimensional Banach space $\R^p$ gives
    \begin{equation}
        G(x+h)-G(x)-DG(x)h
        =\int_0^1(1-s)D^2G(x+sh)[h,h]ds.
    \end{equation}
    Taking norms gives \eqref{eq:G-second}.
\end{proof}

\begin{lemma}\label{lem:smooth-spectral-expansion}
    Fix $L<\infty$ and $\psi\in\SC^3(\R)$, and put
    \begin{equation}
        \Phi_\psi(A)\bydef \frac1d\Tr\big(\psi(A)\big),
        \qquad A\in\Sym.
    \end{equation}
    There exists $C_{\psi,L}<\infty$, depending only on $L$ and on the suprema of $\psi$ and its first three derivatives over $[-L-1,L+1]$, such that the following holds for all $A,E\in\Sym$ satisfying
    \begin{equation}\label{eq:smooth-spectral-local-domain}
        \max\{\|A\|_\op,\|A+E\|_\op\}\leq L.
    \end{equation}
    We have
    \begin{equation}\label{eq:smooth-spectral-taylor}
        \Phi_\psi(A + E) - \Phi_\psi(A) = \frac{1}{d} \Tr(\psi'(A) E) + Q_{\psi, A}(E) + R_{\psi, A}(E),
    \end{equation}
    where $Q_{\psi, A}$ is a quadratic form represented by a self-adjoint linear
    operator $\CK_{\psi, A}:\Sym \to \Sym$ with
    \begin{equation}
        Q_{\psi, A}(E) = \langle E, \CK_{\psi, A} E \rangle
    \end{equation}
    and
    \begin{equation}
        \|\CK_{\psi, A}\|_{F \to F} \leq C_{\psi,L} d^{-1}.
    \end{equation}
    The remainder satisfies
    \begin{equation}\label{eq:smooth-spectral-remainder-bound}
        |R_{\psi, A}(E)| \leq C_{\psi,L} d^{-1} \|E\|_\Frob^3.
    \end{equation}
    Equivalently, in vector form, there is a symmetric matrix
    $K_{\psi, A} \in \R^{p \times p}$ such that
    \begin{equation}
        Q_{\psi, A}(E) = {\svecop(E)}^\top K_{\psi, A} \svecop(E),
        \qquad
        \|K_{\psi, A}\|_\op \leq C_{\psi,L} d^{-1}.
    \end{equation}
    Moreover,
    \begin{equation}\label{eq:smooth-spectral-linear-coeff-bound}
        d^{-1}\norm{\psi'(A)}_{\Frob}\le d^{-1/2}\norm{\psi'}_\infty.
    \end{equation}
    In particular, if $\|E\|_\Frob = \CO(d^{-1/2})$, then
    \begin{equation}
        |R_{\psi,A}(E)| \leq C_{\psi,L} d^{-5/2}.
    \end{equation}
\end{lemma}

\begin{lemma}\label{lem:spectral-taylor-4}
    Fix $0\le t\le T-1$ and recall from \eqref{eq:vectorized-AMP-recursion} that $G_t(v)\bydef \svecop(\sigma_t(\mat(v)))$.
    For each $y\in\R^p$, there exist a symmetric bilinear map
    \begin{equation}
        \BFB_{t,y}:(\R^p)^2 \to \R^p,
    \end{equation}
    a symmetric trilinear map
    \begin{equation}
        \BFT_{t,y}:(\R^p)^3 \to \R^p,
    \end{equation}
    and a remainder map
    \begin{equation}
        \rho_{t,y}^{(4)}:\R^p \to \R^p
    \end{equation}
    such that for every $g\in\R^p$,
    \begin{equation}\label{eq:local-spectral-third-order}
        G_t(y+g)-G_t(y)-DG_t(y)g
        =
        \BFB_{t,y}(g,g)+\BFT_{t,y}(g,g,g)+\rho_{t,y}^{(4)}(g).
    \end{equation}
    Moreover, with
    \begin{equation}\label{eq:M-sigma-4-def}
        M_{\sigma,4}\bydef
        \max_{0\,\le\, s\,\le\, T-1}\max_{1\,\le\, r\,\le\, 4}\norm{\sigma_s^{(r)}}_\infty,
    \end{equation}
    there is a constant $C<\infty$, depending only on $M_{\sigma,4}$, such that
    for all $y\in\R^p$ and all $g,g_1,g_2,g_3,w\in\R^p$,
    \begin{equation}\label{eq:local-spectral-B-bound}
        \norm{\BFB_{t,y}(g_1,g_2)} \le C \norm{g_1}\norm{g_2},
    \end{equation}
    \begin{equation}\label{eq:local-spectral-T-bound}
        \norm{\BFT_{t,y}(g_1,g_2,g_3)} \le C \norm{g_1}\norm{g_2}\norm{g_3},
    \end{equation}
    \begin{equation}\label{eq:local-spectral-rem-bound}
        \norm{\rho_{t,y}^{(4)}(g)} \le C \norm{g}^4,
    \end{equation}
    \begin{equation}\label{eq:local-spectral-B-HS}
        \big\|\inner{w}{\BFB_{t,y}(\cdot,\cdot)}\big\|_{\mathrm{HS}}
        \le C \sqrt d \norm{w},
    \end{equation}
    and
    \begin{equation}\label{eq:local-spectral-T-HS}
        \big\|\inner{w}{\BFT_{t,y}(\cdot,\cdot,\cdot)}\big\|_{\mathrm{HS}}
        \le C d \norm{w}.
    \end{equation}
\end{lemma}

\begin{lemma}\label{lem:bilinear-matrix}
    Let $\mathfrak B:(\R^p)^2\to\R$ be bilinear. There exists a unique matrix $A\in\R^{p\times p}$ such that
    \begin{equation}
        \FB(x,y)=\inner{x}{Ay},
        \qquad x,y\in\R^p.
    \end{equation}
    If $M=(A+A^\top)/2$, then
    \begin{equation}
        \FB(u,u)=\inner{u}{Mu},
    \end{equation}
    and
    \begin{equation}
        \norm{M}_{\Frob}\le \norm{A}_{\Frob}=\norm{\FB}_{\mathrm{HS}}.
    \end{equation}
\end{lemma}

\begin{proof}[Proof of Lemma \ref{lem:bilinear-matrix}]
    The matrix $A$ is defined by $A_{ab}=\FB(e_a,e_b)$ in the standard basis. Since $u^\top Au=u^\top Mu$, the first identity is immediate. The Frobenius bound follows from orthogonal projection onto the symmetric subspace.
\end{proof}

\begin{lemma}\label{lem:trilinear-decomp}
    Let $\FT:(\R^p)^3 \to \R$ be trilinear and fix an orthonormal basis $(e_\nu)_{\nu = 1}^p$ of $\R^p$. Then there exist symmetric matrices $M_1, \ldots, M_p \in \R^{p \times p}$ and a vector $b \in \R^p$ such that for every $u \in \R^p$,
    \begin{equation}\label{eq:trilinear-decomp-main}
        \FT(u,u,u) = \inner{b}{u} + \sum_{\nu = 1}^p \inner{e_\nu}{u} \left(u^\top M_\nu u - \frac 1 p \Tr M_\nu \right)
    \end{equation}
    Further, we have bounds
    \begin{equation}
        \sum_{\nu\,=\,1}^p \norm{M_\nu}_\Frob^2 \,\le\, \norm{\FT}_\HS^2, \quad \norm{b} \,\le\, p^{-1/2} \norm{\FT}_\HS
    \end{equation}
\end{lemma}
\begin{proof}[Proof of Lemma \ref{lem:trilinear-decomp}]
    For each $\nu \in [p]$, define bilinear form $\FB_\nu(x, y) \bydef \FT(e_\nu, x, y)$. Let $A_\nu$ be the matrix representation of this bilinear form. From Lemma \ref{lem:bilinear-matrix} and let $M_\nu$ be its symmetric version. Then
    \begin{equation}
        \FT(u,u,u) \,=\, \sum_{\nu\,=\,1}^p \inner{e_\nu}{u} \FB_\nu(u,u) \,=\, \sum_{\nu \,=\, 1}^p \inner{e_\nu}{u} \inner{u}{M_\nu u}
    \end{equation}
    Set $m_\nu \bydef \frac 1 p \Tr M_\nu$ and $b \bydef \sum_{\nu = 1}^p m_\nu e_\nu$. This gives \eqref{eq:trilinear-decomp-main}. Furthermore, we have
    \begin{equation}
        \sum_{\nu\,=\,1}^p \norm{M_\nu}_\Frob^2 \,\le\, \sum_{\nu\,=\,1}^p \norm{A_\nu}_\Frob^2 \,=\, \sum_{\nu,a,b}^p \abs{\FT(e_\nu,e_a,e_b)}^2 \,=\, \norm{\FT}_\HS^2
    \end{equation}
    Similarly,
    \begin{equation}
        \norm{b}^2 \,=\, \sum_{\nu\,=\,1}^p m_\nu^2 \,\le\, \frac 1 p \sum_{\nu\,=\,1}^p \norm{M_\nu}_\Frob^2 \,\le\, \frac 1 p \norm{\FT}_\HS^2
    \end{equation}
\end{proof}

\begin{lemma}\label{lem:hs-composition}
Let $\FB:(\R^p)^2\to\R$ be bilinear and let $A_1,A_2\in \R^{p \times p}$. Define
\begin{equation}
    \FB_{A_1,A_2}(u_1,u_2) \bydef \FB(A_1u_1,A_2u_2).
\end{equation}
Then
\begin{equation}\label{eq:hs-compose-bilinear}
    \norm{\FB_{A_1,A_2}}_\HS
    \le \norm{\FB}_\HS \norm{A_1}_\op \norm{A_2}_\op.
\end{equation}
If $\FT:(\R^p)^3\to\R$ is trilinear and $A_1,A_2,A_3:\R^p\to\R^p$ are linear, then
\begin{equation}\label{eq:hs-compose-trilinear-direct}
    \norm{\FT(A_1\cdot,A_2\cdot,A_3\cdot)}_\HS
    \le \norm{\FT}_\HS \norm{A_1}_\op \norm{A_2}_\op \norm{A_3}_\op.
\end{equation}
Now let $\BFB:(\R^p)^2\to\R^p$ be bilinear and suppose that for some $K<\infty$,
\begin{equation}\label{eq:hs-compose-hyp}
    \norm{\inner{w}{\BFB(\cdot,\cdot)}}_\HS \le K \norm{w}
    \qquad \text{for every } w\in\R^p.
\end{equation}
Let $A,U,V,W:\R^p\to\R^p$ be linear and define
\begin{equation}
    \FC(u_1,u_2,u_3)
    \bydef
    \FB\big(Au_1, U\BFB(Vu_2,Wu_3)\big).
\end{equation}
Then
\begin{equation}\label{eq:hs-compose-trilinear}
    \norm{\FC}_\HS
    \le K \norm{\FB}_\HS \norm{A}_\op \norm{U}_\op \norm{V}_\op \norm{W}_\op.
\end{equation}
\end{lemma}

\subsection{Proof of Lemma \ref{lem:dk-bounds}}

\begin{proof}[Proof of Lemma \ref{lem:dk-bounds}]
    This is obtained from
    \cite{RubenssonEmanuelH2023Hodo}, Theorems 3.1 and 3.3 and the Hermitian
    formula in Section 4, applied to the path
    $X+\sum_{r=1}^m t_rH_r$. Since $\svecop$ and $\mat$ are linear isometries,
    the same differentiability holds for $G$.

    It remains only to prove the dimension-free norm bound. Fix
    $\pi\in S_m$ and let $Z_\pi$ be the corresponding ordered summand in
    \eqref{eq:dk-formula}. Diagonalize $X=U\Lambda U^\top$ and set
    $\widetilde H_r\bydef U^\top H_rU$. Then, for fixed endpoints $i_0,i_m$,
    \begin{equation}
        (U^\top Z_\pi U)_{i_0i_m}
        =
        \sum_{i_1,\dots,i_{m-1}=1}^d
        \varphi^{[m]}(\lambda_{i_0},\dots,\lambda_{i_m})
        (\widetilde H_{\pi(1)})_{i_0i_1}\cdots
        (\widetilde H_{\pi(m)})_{i_{m-1}i_m}.
    \end{equation}
    Therefore, with entrywise absolute values,
    \begin{equation}
        \abs{U^\top Z_\pi U}_{i_0i_m}
        \le
        \left(\norm{\varphi^{[m]}}_\infty
        \abs{\widetilde H_{\pi(1)}}\cdots \abs{\widetilde H_{\pi(m)}}\right)_{i_0, i_m}
    \end{equation}
    Frobenius submultiplicativity gives
    \begin{equation}
        \norm{Z_\pi}_{\Frob} \,\le\, \norm{\varphi^{[m]}}_\infty \prod_{r\,=\,1}^m \norm{\widetilde H_r}_{\Frob} \,=\, \norm{\varphi^{[m]}}_\infty \prod_{r\,=\,1}^m \norm{H_r}_{\Frob}.
    \end{equation}
    Summing the $m!$ ordered terms and using the Hermite--Genocchi bound
    \eqref{eq:divided-difference-bound}, namely
    $\norm{\varphi^{[m]}}_\infty\le \norm{\varphi^{(m)}}_\infty/m!$, yields
    \eqref{eq:dk-frobenius-bound}. The vector bound
    \eqref{eq:dk-vector-bound} follows by applying the isometries $\svecop$ and
    $\mat$.
\end{proof}

\subsection{Proof of Lemma \ref{lem:smooth-spectral-expansion}}

\begin{proof}[Proof of Lemma \ref{lem:smooth-spectral-expansion}]
    Choose a smooth cutoff $\chi_L$ that equals one on $[-L,L]$ and is supported on $[-L-1,L+1]$, and set $\bar\psi=\chi_L\psi$. Its first three derivatives are bounded by a constant depending only on $L$ and on the corresponding local bounds for $\psi$. Moreover,
    \begin{equation*}
        A+sE=(1-s)A+s(A+E),
        \qquad 0\leq s\leq1,
    \end{equation*}
    so \eqref{eq:smooth-spectral-local-domain} and convexity of the operator norm imply $\|A+sE\|_\op\leq L$. Hence $\psi$ and $\bar\psi$ induce the same spectral functional along the entire segment from $A$ to $A+E$. It is therefore enough to prove the claim for the bounded $C^3$ function $\bar\psi$; to simplify notation, we again denote it by $\psi$.

    Lemma \ref{lem:dk-bounds} implies that the matrix map $A\mapsto \psi(A)$ is $C^3$; hence the normalized trace functional $\Phi_\psi$ is $C^3$ on $\Sym$.

    We first identify the linear term. If
    $A=\sum_{a=1}^d \lambda_a P_a$ is a spectral decomposition of
    $A$, then the first-order Daleckii--Krein formula gives
    \begin{equation}
        D(\psi(\cdot))(A)[E]
        =
        \sum_{a,b=1}^d \psi^{[1]}(\lambda_a,\lambda_b)P_a E P_b.
    \end{equation}
    Taking traces and using cyclic invariance leaves only the terms $a=b$, and
    $\psi^{[1]}(\lambda_a,\lambda_a)=\psi'(\lambda_a)$. Therefore
    \begin{equation}\label{eq:Phi-linear-term}
        D\Phi_\psi(A)[E]
        =\frac1d\Tr\big(\psi'(A)E\big).
    \end{equation}

    Define the symmetric bilinear form
    \begin{equation}
        B_{\psi,A}(H_1,H_2)\bydef \frac12 D^2\Phi_\psi(A)[H_1,H_2].
    \end{equation}
    Differentiating \eqref{eq:Phi-linear-term} in the direction $H_1$ and
    testing against $H_2$ gives
    \begin{equation}\label{eq:Phi-Hessian-trace}
        D^2\Phi_\psi(A)[H_1,H_2]
        =\frac1d\Tr\big(D(\psi') (A)[H_1]H_2\big).
    \end{equation}
    Hence, by Cauchy--Schwarz and Lemma \ref{lem:dk-bounds} applied to
    $\psi'$ with $m=1$,
    \begin{equation}
        \abs{D^2\Phi_\psi(A)[H_1,H_2]}
        \le C_\psi d^{-1}\norm{H_1}_{\Frob}\norm{H_2}_{\Frob}.
    \end{equation}
    The Riesz representation theorem on the inner product space $(\Sym,\langle\cdot,
    \cdot\rangle_\Frob)$ therefore gives a unique self-adjoint linear operator
    $\CK_{\psi,A}:\Sym\to\Sym$ such that
    \begin{equation}
        B_{\psi,A}(H_1,H_2)=\inner{H_1}{\CK_{\psi,A}(H_2)},
        \qquad H_1,H_2\in\Sym,
    \end{equation}
    and
    \begin{equation}
        \norm{\CK_{\psi,A}}_{\Frob\to\Frob}\le C_\psi d^{-1}.
    \end{equation}
    Thus
    \begin{equation}
        Q_{\psi,A}(E)\bydef B_{\psi,A}(E,E)
        =\inner{E}{\CK_{\psi,A}(E)}
        =\frac12D^2\Phi_\psi(A)[E,E].
    \end{equation}

    Taylor's theorem with integral remainder gives
    \begin{equation}
        \Phi_\psi(A+E)-\Phi_\psi(A)-D\Phi_\psi(A)[E]-\frac12D^2\Phi_\psi(A)[E,E]
        =
        \int_0^1 \frac{(1-s)^2}{2}D^3\Phi_\psi(A+sE)[E,E,E]ds.
    \end{equation}
    We define the right-hand side to be $R_{\psi,A}(E)$. Differentiating
    \eqref{eq:Phi-Hessian-trace} once more, and using symmetry of higher
    derivatives, yields
    \begin{equation}
        D^3\Phi_\psi(B)[H_1,H_2,H_3]
        =
        \frac1d\Tr\big(D^2(\psi')(B)[H_1,H_2]H_3\big).
    \end{equation}
    Applying Lemma \ref{lem:dk-bounds} to $\psi'$ with $m=2$ and then
    Cauchy--Schwarz gives
    \begin{equation}
        \abs{D^3\Phi_\psi(B)[H_1,H_2,H_3]}
        \le C_\psi d^{-1}
        \norm{H_1}_{\Frob}\norm{H_2}_{\Frob}\norm{H_3}_{\Frob}.
    \end{equation}
    Consequently
    \begin{equation}
        \abs{R_{\psi,A}(E)}\le C_\psi d^{-1}\norm{E}_{\Frob}^3,
    \end{equation}
    which is \eqref{eq:smooth-spectral-remainder-bound}. This proves
    \eqref{eq:smooth-spectral-taylor}.

    For the vector form, let
    \begin{equation}
        K_{\psi,A}\bydef \svecop\circ \CK_{\psi,A}\circ \mat,
    \end{equation}
    viewed as its matrix in the standard Euclidean coordinates of $\R^p$.
    Since $\svecop$ is an isometry, $K_{\psi,A}$ is symmetric and has the same
    operator norm as $\CK_{\psi,A}$. Finally,
    \begin{equation}
        d^{-1}\norm{\psi'(A)}_{\Frob}
        \le d^{-1}\sqrt d\,\norm{\psi'}_\infty
        =d^{-1/2}\norm{\psi'}_\infty,
    \end{equation}
    and the final claim follows immediately from
    \eqref{eq:smooth-spectral-remainder-bound}.
\end{proof}

\subsection{Proof of Lemma \ref{lem:spectral-taylor-4}}

\begin{proof}[Proof of Lemma \ref{lem:spectral-taylor-4}]
    Fix $t$, write $\varphi=\sigma_t$ and $G=G_t$. Since
    $\varphi\in C_b^4(\R)$ and $\norm{\varphi^{(r)}}_\infty\le M_{\sigma,4}$ for
    $1\le r\le4$, Lemma \ref{lem:dk-bounds} implies that $G$ is $C^4$ and
    \begin{equation}\label{eq:G-derivative-bounds-up-to-4}
        \sup_{z\in\R^p}\norm{D^rG(z)}_{\op}\le M_{\sigma,4},
        \qquad 1\le r\le 4.
    \end{equation}

    Define
    \begin{equation}
        \BFB_{t,y}(g_1,g_2) \bydef \frac12 D^2G(y)[g_1,g_2],
    \end{equation}
    \begin{equation}
        \BFT_{t,y}(g_1,g_2,g_3) \bydef \frac16 D^3G(y)[g_1,g_2,g_3],
    \end{equation}
    and
    \begin{equation}
        \rho_{t,y}^{(4)}(g)
        \bydef
        \int_0^1 \frac{(1-s)^3}{6} D^4G(y+sg)[g,g,g,g] ds.
    \end{equation}
    Taylor's theorem with integral remainder gives
    \eqref{eq:local-spectral-third-order}.  The pointwise bounds \eqref{eq:local-spectral-B-bound}--\eqref{eq:local-spectral-rem-bound} follow immediately from \eqref{eq:G-derivative-bounds-up-to-4}.

    It remains to prove the Hilbert--Schmidt bounds. Fix $w\in\R^p$ and write $W=\mat(w)$. Let
    \begin{equation}
        Y \,\bydef\, \mat(y) \,=\, \sum_{a\,=\,1}^d \lambda_a P_a,
        \qquad P_a = u_a u_a^\top,
    \end{equation}
    with eigenvalues repeated according to multiplicity. For matrices
    $H_j=\mat(g_j)$, the second-order Daleckii--Krein formula in Lemma \ref{lem:dk-bounds} gives
    \begin{equation}
        \inner{w}{D^2G(y)[g_1,g_2]}
        =
        \sum_{\pi\in S_2}
        \sum_{a,b,c=1}^d
        \varphi^{[2]}(\lambda_a,\lambda_b,\lambda_c)
        \Tr(WP_a H_{\pi(1)} P_b H_{\pi(2)} P_c).
    \end{equation}
    Define the ordered bilinear form on the full matrix space $\R^{d\times d}$ by
    \begin{equation}
        \FB_{y,w}(H_1,H_2)
        \bydef
        \sum_{a,b,c=1}^d
        \varphi^{[2]}(\lambda_a,\lambda_b,\lambda_c)
        \Tr(WP_a H_1 P_b H_2 P_c).
    \end{equation}
    Since $\BFB_{t,y}=\frac12D^2G(y)$,
    \begin{equation}
        \inner{w}{\BFB_{t,y}(g_1,g_2)}
        =\frac12\Big(\FB_{y,w}(H_1,H_2)+\FB_{y,w}(H_2,H_1)\Big).
    \end{equation}
    Permuting arguments preserves Hilbert--Schmidt norm, and restricting a
    multilinear form from $\R^{d\times d}$ to the subspace $\Sym$ cannot
    increase its Hilbert--Schmidt norm. Thus
    \begin{equation}\label{eq:B-scalar-HS-reduce}
        \big\|\inner{w}{\BFB_{t,y}(\cdot,\cdot)}\big\|_{\mathrm{HS}}
        \le \norm{\FB_{y,w}}_{\mathrm{HS}(\R^{d\times d})}.
    \end{equation}

    Let
    \begin{equation}
        \sf{E}_{ab}\bydef u_a u_b^\top,
        \qquad 1\le a,b\le d.
    \end{equation}
    Then $(\sf{E}_{ab})_{a,b}$ is a Frobenius-orthonormal basis of
    $\R^{d\times d}$. A direct computation shows that
    \begin{equation}
        \FB_{y,w}(\sf{E}_{ab},\sf{E}_{bc})
        =
        \varphi^{[2]}(\lambda_a,\lambda_b,\lambda_c)(u_c^\top W u_a),
    \end{equation}
    while $\FB_{y,w}(\sf{E}_{rs},\sf{E}_{tu})=0$ unless $s=t$. Hence
    \begin{equation}
        \norm{\FB_{y,w}}_{\mathrm{HS}}^2
        =
        \sum_{a,b,c=1}^d
        \abs{\varphi^{[2]}(\lambda_a,\lambda_b,\lambda_c)}^2
        \abs{u_c^\top W u_a}^2
        \le C d\norm{W}_{\Frob}^2
        = C d\norm{w}^2.
    \end{equation}
    Combining this with \eqref{eq:B-scalar-HS-reduce} proves
    \eqref{eq:local-spectral-B-HS}.

    The trilinear estimate is analogous. Define the trilinear form on
    $\R^{d\times d}$ by
    \begin{equation}
        \FT_{y,w}(H_1,H_2,H_3)
        \bydef
        \sum_{a,b,c,e=1}^d
        \varphi^{[3]}(\lambda_a,\lambda_b,\lambda_c,\lambda_e)
        \Tr(WP_a H_1 P_b H_2 P_c H_3 P_e).
    \end{equation}
    Since $\BFT_{t,y}=\frac16D^3G(y)$,
    $\inner{w}{\BFT_{t,y}(\cdot,\cdot,\cdot)}$ is the average over the six
    permutations of the arguments of $\FT_{y,w}$, restricted to $\Sym$. Hence
    \begin{equation}
        \big\|\inner{w}{\BFT_{t,y}(\cdot,\cdot,\cdot)}\big\|_{\mathrm{HS}}
        \le \norm{\FT_{y,w}}_{\mathrm{HS}(\R^{d\times d})}.
    \end{equation}
    Moreover,
    \begin{equation}
        \FT_{y,w}(\sf{E}_{ab},\sf{E}_{bc},\sf{E}_{ce})
        =
        \varphi^{[3]}(\lambda_a,\lambda_b,\lambda_c,\lambda_e)(u_e^\top W u_a),
    \end{equation}
    and the form vanishes on all other triples of matrix units. Therefore
    \begin{equation}
        \norm{\FT_{y,w}}_{\mathrm{HS}}^2
        \le C d^2\norm{W}_{\Frob}^2
        = C d^2\norm{w}^2.
    \end{equation}
    This proves \eqref{eq:local-spectral-T-HS}.
\end{proof}

\subsection{Proof of Lemma \ref{lem:hs-composition}}

\begin{proof}[Proof of Lemma \ref{lem:hs-composition}]
    Let $L$ be the matrix representation of the scalar bilinear form $\FB$, so
    $\FB(x,y)=\inner{x}{Ly}$ and $\norm{\FB}_\HS=\norm{L}_\Frob$. Then
    \begin{equation}
        \FB_{A_1,A_2}(u_1,u_2)=\inner{u_1}{A_1^\top L A_2 u_2},
    \end{equation}
    and therefore
    \begin{equation}
        \norm{\FB_{A_1,A_2}}_\HS
        =\norm{A_1^\top L A_2}_\Frob
        \le \norm{L}_\Frob\norm{A_1}_\op\norm{A_2}_\op,
    \end{equation}
    proving \eqref{eq:hs-compose-bilinear}.

    For \eqref{eq:hs-compose-trilinear-direct}, identify $\FT$ with its
    coefficient tensor $T=(T_{abc})$ in the standard orthonormal basis. The
    coefficient tensor of $\FT(A_1\cdot,A_2\cdot,A_3\cdot)$ is
    $(A_1^\top\otimes A_2^\top\otimes A_3^\top)T$ after vectorization. Hence
    \begin{equation}
        \norm{\FT(A_1\cdot,A_2\cdot,A_3\cdot)}_\HS
        \le
        \norm{A_1^\top\otimes A_2^\top\otimes A_3^\top}_\op\norm{\FT}_\HS
        =
        \norm{A_1}_\op\norm{A_2}_\op\norm{A_3}_\op\norm{\FT}_\HS.
    \end{equation}

    We now prove \eqref{eq:hs-compose-trilinear}. Let $L$ again represent
    $\FB$. For each standard basis vector $e_a$, define the scalar bilinear form
    \begin{equation}
        G_a(u_2,u_3)\bydef \inner{U^\top L^\top A e_a}{\BFB(u_2,u_3)}.
    \end{equation}
    Then
    \begin{equation}
        \FC(e_a,u_2,u_3)=G_a(Vu_2,Wu_3).
    \end{equation}
    By \eqref{eq:hs-compose-hyp},
    \begin{equation}
        \norm{G_a}_\HS
        \le K\norm{U^\top L^\top A e_a}
        \le K\norm{U}_\op\norm{L^\top A e_a}.
    \end{equation}
    Applying the already proved bilinear composition bound to $G_a$ gives
    \begin{equation}
        \norm{\FC(e_a,\cdot,\cdot)}_\HS
        \le
        K\norm{U}_\op\norm{V}_\op\norm{W}_\op\norm{L^\top A e_a}.
    \end{equation}
    Summing over $a$ yields
    \begin{equation}
        \begin{aligned}
        \norm{\FC}_\HS^2
        &=\sum_{a=1}^p\norm{\FC(e_a,\cdot,\cdot)}_\HS^2 \\
        &\le K^2\norm{U}_\op^2\norm{V}_\op^2\norm{W}_\op^2
        \sum_{a=1}^p\norm{L^\top A e_a}^2 \\
        &=K^2\norm{U}_\op^2\norm{V}_\op^2\norm{W}_\op^2\norm{L^\top A}_\Frob^2 \\
        &\le K^2\norm{\FB}_\HS^2\norm{A}_\op^2\norm{U}_\op^2\norm{V}_\op^2\norm{W}_\op^2.
        \end{aligned}
    \end{equation}
    Taking square roots proves \eqref{eq:hs-compose-trilinear}.
\end{proof}

\section{Probabilistic Bounds}\label{app:prob}

We collect bounds for the sensing operators, single-column moments, Gaussian cumulants, and leave-one-out differences. The longer proofs follow the statements. Throughout, $\CG_i$ is the leave-one-out sigma-algebra in \eqref{eq:loo-sigma-algebra-def}.

We now prove bounds on the operator norms of the sensing matrices. For a random variable $X$, we write
\begin{equation}
    \norm{X}_{\psi_1} \,\bydef\, \inf \{C \,>\, 0 : \E \exp(|X|/C) \,\le\, 2\}
\end{equation}
for the sub-exponential norm. We use the following result from \cite{1012.0294}.

\begin{lemma}\label{lem:adamczak-opnorm}
    Let $X_1, \ldots, X_N \in \R^p$ be independent centered isotropic random vectors. Assume that for some constants $\psi, K, c > 0$, we have
    \begin{equation}
        \sup_{j \,\le\, N} \sup_{y \,\in\, S^{p-1}} \norm{\inner{X_j}{y}}_{\psi_1} \,\le\, \psi
    \end{equation}
    and
    \begin{equation}
        \P\left(\max_{j\,\le\, N}\frac{\|X_j\|}{\sqrt p}\,>\, K\max\left\{1,\left(\frac{N}{p}\right)^{1/4}\right\}\right)\,\le\, e^{-c\sqrt p}.
    \end{equation}
    Let $A = [X_1, \ldots, X_N] \in \R^{p \times N}$, then
    \begin{equation}
        \P(\norm{A}_\op > C \sqrt{p}) \le 2 e^{-c'\sqrt{p}}
    \end{equation}
    for some constants $C,c'>0$ depending only on $\psi,K$ and an upper bound for $N/p$.
\end{lemma}
\begin{proof}[Proof of Lemma \ref{lem:adamczak-opnorm}]
    This is Theorem 2 of \cite{1012.0294}.
\end{proof}

\begin{lemma}\label{lem:sensing-envelope}
    For any fixed $m \ge 1$, there exists $C_m < \infty$ such that uniformly in $i$,
    \begin{equation}
        \norm{K_i}_{L^m(\op)} \leq C_m, \quad \norm{F_\smi \epsilon}_{L^m} \leq C_m d
    \end{equation}
    Moreover, for $u \in \{f_i, \widetilde f_i\}$,
    \begin{equation}
        \|u\|_{L^m} \leq C_m.
    \end{equation}
\end{lemma}

\begin{lemma}\label{lem:one-column-moments}
    Let $u\in\{f_i,\widetilde f_i\}$. Let $a,b\in\R^p$ and $A\in\R^{p\times p}$ be $\mathcal G_i$-measurable, and let $A_{\mathrm s}=(A+A^\top)/2$.
    Then almost surely,
    \begin{equation}\label{eq:linear-moment-1}
        \E[\inner{u}{a}\mid \mathcal G_i]=0
    \end{equation}
    \begin{equation}\label{eq:linear-moment-2}
        \E[\inner{u}{a}\inner{u}{b}\mid \mathcal G_i]=\frac1p\inner{a}{b}
    \end{equation}
    Moreover, for every integer $m\ge 1$ there exists $C_m<\infty$ such that almost surely,
    \begin{equation}\label{eq:linear-moment-lm}
        \E\left[\abs{\inner{u}{a}}^m\mid \mathcal G_i\right]\le C_m p^{-m/2}\norm{a}^m
    \end{equation}
    \begin{equation}\label{eq:quadratic-moment-lm}
        \E\left[\abs{\inner{u}{Au}-\frac1p\Tr A}^m\mid \mathcal G_i\right]
        \le C_m d^{-m/2}\norm{A_{\mathrm s}}_{\op}^m
    \end{equation}
    Consequently, after averaging over $\mathcal G_i$,
    \begin{equation}\label{eq:linear-moment-uncond}
        \norm{\inner{u}{a}}_{L^m}\le C_m d^{-1}\norm{a}_{L^m}
    \end{equation}
    \begin{equation}\label{eq:quadratic-moment-uncond}
        \left\|\inner{u}{Au}-\frac1p\Tr A\right\|_{L^m}
        \le C_m d^{-1/2}\norm{A_{\mathrm s}}_{L^{2m}(\op)}
    \end{equation}
\end{lemma}

We will use the preceding lemma in the following form.
\begin{corollary}\label{cor:one-column-moments}
For every integer $m\ge 1$ there exists $C_m<\infty$ such that the following hold uniformly in $i$ and in $u\in\{f_i,\widetilde f_i\}$.
\begin{enumerate}[label=(\roman*)]
    \item If $a$ is $\mathcal G_i$-measurable, then
    \begin{equation}
        \norm{\inner{u}{a}}_{L^m}\le C_m d^{-1}\norm{a}_{L^m}.
    \end{equation}
    In particular, if $\norm{a}_{L^m}\le C d^{-1/2}$, then
    \begin{equation}
        \norm{\inner{u}{a}}_{L^m}\le C_m d^{-3/2}.
    \end{equation}
    \item If $A$ is $\mathcal G_i$-measurable and $A_{\mathrm s}=(A+A^\top)/2$, then
    \begin{equation}
        \left\|\inner{u}{Au}-\frac1p\Tr A\right\|_{L^m}\le C_m d^{-1/2}\norm{A_{\mathrm s}}_{L^{2m}(\op)}.
    \end{equation}
    In particular, if $\norm{A_{\mathrm s}}_{L^{2m}(\op)}\le C$, then
    \begin{equation}
        \left\|\inner{u}{Au}-\frac1p\Tr A\right\|_{L^m}\le C_m d^{-1/2},
    \end{equation}
    while if $\norm{A_{\mathrm s}}_{L^{2m}(\op)}\le C d^{-1}$, then
    \begin{equation}
        \left\|\inner{u}{Au}-\frac1p\Tr A\right\|_{L^m}\le C_m d^{-3/2}.
    \end{equation}
\end{enumerate}
\end{corollary}
\begin{proof}[Proof of Corollary \ref{cor:one-column-moments}]
    This follows directly from Lemma~\ref{lem:one-column-moments}.
\end{proof}
Lemma~\ref{lem:one-column-moments} bounds quadratic-form fluctuations in terms of the operator norm. For the linear comparison, we need a bound in terms of the Frobenius norm, which is controlled by the spectral derivative estimates.

\begin{lemma}\label{lem:refined-centered-quad}
    For every integer $m \ge 2$ there exists $C_m < \infty$ such that the following holds uniformly in $i$ and in $u \in \{f_i, \widetilde f_i\}$. If $A \in \R^{p \times p}$ is $\CG_i$-measurable and $A_s = \frac{A + A^\top}{2}$, then we have
    \begin{equation}\label{eq:refined-centered-quad-bound}
        \norm{u^\top A u - \frac 1 p \Tr A}_{L^m} \le C_m d^{-3/2} \norm{A_s}_{L^{2m}(\Frob)}
    \end{equation}
    In particular, if $\norm{A}_{L^{2m}(\Frob)} \le C d^{-1/2}$, then
    \begin{equation}
        \norm{ u^\top A u - \frac 1 p \Tr A}_{L^m} \le C_m d^{-2}
    \end{equation}
\end{lemma}

To work with the rank-one case, we write down the cumulants of a Gaussian quadratic form.
\begin{lemma}\label{lem:gaussian-quad-cumulants}
    Let $x \sim \CN(0, I_d)$. Let $M \in \Sym$ and define
    \begin{equation}
        q_M(x) = x^\top M x - \Tr(M)
    \end{equation}
    For random variables $X_1, \ldots, X_r$ define their joint cumulant by
    \begin{equation}
        \kappa(Z_1, \ldots, Z_r) \,\bydef\, \partial_{t_1} \cdots \partial_{t_r} \log \E \exp\left( \sum_{j \,=\, 1}^t t_j X_j  \right)\Bigg|_{t \,=\, 0}
    \end{equation}
    Then for every integer $ r \ge 2$ and every $M_1, \cdots, M_r \in \Sym$, we have
    \begin{equation}
        \kappa\left(q_{M_1}(x), \ldots, q_{M_r}(x) \right) \,=\, 2^{r-1} \sum_{\sigma \,\in\, \FC_r} \Tr\left( M_{\sigma(1)} \cdots M_{\sigma(r)} \right)
    \end{equation}
    where $\FC_r$ is $S_r$ quotient by the equivalence of differing by a cycle. In particular, we have
    \begin{equation}
        \kappa(q_M(x), q_N(x)) = 2 \Tr(MN)
    \end{equation}
    \begin{equation}
        \kappa\left( q_{M_1}(x), q_{M_2}(x), q_{M_3}(x), q_{M_4}(x) \right) \,=\, 8 \sum_{\sigma \,\in\, \FC_4} \Tr(M_{\sigma(1)}M_{\sigma(2)} M_{\sigma(3)} M_{\sigma(4)})
    \end{equation}
\end{lemma}

\subsection{Proof of Lemma \ref{lem:sensing-envelope}}

\begin{proof}[Proof of Lemma \ref{lem:sensing-envelope}]
    We first prove the $m$th-moment bound for the sensing vector norm. If $u = \widetilde f_i \sim \CN(0, I_p/p)$, then we have $p\|u\|^2 \sim \chi_p^2$. Therefore, $\|u\|_{L^m} \leq C_m$ for every $m \geq 1$. If $u = f_i$, then
    \begin{equation}
        \|u\|^2 \,=\, \frac{1}{d(d + 1)} \norm{x_ix_i^\top - I_d}_\Frob^2 \,=\, \frac{\norm{x_i}^4 - 2 \norm{x_i}^2 + d}{d(d + 1)}.
    \end{equation}
    Observe that $\norm{x_i}^2 \sim \chi_d^2$. Therefore, the RHS has every moment bounded uniformly in $d$. So $\|u\|_{L^m} \leq C_m$. Next, denote the non-leave-one-out columns by
    \begin{equation}
        v_j^{(i)} \bydef \begin{cases}
            \widetilde f_j, & 1 \leq j < i \\
            f_j, & i < j \leq n
        \end{cases}
    \end{equation}
    and let $V_i \bydef [v_j^{(i)}]_{j \neq i} \in \R^{p \times (n - 1)}$. By the definition of $K_i$ in \eqref{eq:Ki-def}, we write $K_i = V_iV_i^\top$. Set $X_j^{(i)} = \sqrt{p}\, v_j^{(i)}$. For the Gaussian columns, we have
    \begin{equation}
        \sup_{y \,\in\, S^{p-1}} \norm{\inner{ X_j^{(i)}}{y}}_{\psi_1} \,\leq\, \|\CN(0,1)\|_{\psi_1} \,\leq\, C
    \end{equation}
    where $\norm{\cdot}_{\psi_1}$ denotes the sub-exponential norm. In particular, $\frac{\norm{X_j^{(i)}}}{\sqrt{p}}$ has sub-exponential tails with constants independent of $d$. For the RO columns, if $y = \svecop(M)$ with $\norm{M}_\Frob = 1$, then
    \begin{equation}
        \inner{X_j^{(i)}}{y} = \frac{x_j^\top M x_j - \Tr M}{\sqrt{2}}.
    \end{equation}
    Writing $M = U \Lambda U^\top$, this is a linear combination of the centered variables $(\xi_r^2 - 1)$ with coefficients the eigenvalues of $M$. Since $\sum_r \lambda_r^2 = \norm{M}_\Frob^2 = 1$, standard bounds for Gaussian quadratic forms (Hanson--Wright) give
    \begin{equation}
        \sup_{y \in S^{p-1}} \norm{\inner{X_j^{(i)}}{y}}_{\psi_1} \leq C.
    \end{equation}
    Moreover,
    \begin{equation}
        \frac{\norm{X_j^{(i)}}^2}{p} = \norm{v_j^{(i)}}^2,
    \end{equation}
    and for the rank-one columns this equals
    \begin{equation}
        \frac{\norm{x_j}^4 - 2 \norm{x_j}^2 + d}{d(d + 1)}.
    \end{equation}
    So for the Gaussian columns, we have from sub-gaussianity
    \begin{equation}
        \P\left(\norm{X^{(i)}_j} \,>\, 2\sqrt{p} \right) \,\leq\, 2e^{-cp}, \quad j \,<\, i,
    \end{equation}
    and for the rank-one columns, we have from the expression above and a $\chi_d^2$ tail bound (which is sub-exponential)
    \begin{equation}
        \P\left(\norm{X^{(i)}_j} \,>\, 2\sqrt{p} \right) \,\leq\, 2 e^{-c\sqrt{p}}, \quad j \,>\, i.
    \end{equation}
    Since $n\leq C_0p$, a union bound over the $n$ columns shows that the boundedness condition from \cite[Eq. 3]{1012.0294} holds uniformly in $i$ for $V_i$ with a deterministic constant $K$,
    \begin{equation}\label{eq:boundedness-event}
        \P \left(
        \max_{j\ne i}\frac{\norm{X_j^{(i)}}}{\sqrt p} > K\max\left\{1,\left(\frac{n-1}{p}\right)^{1/4}\right\}
    \right) \leq e^{-c\sqrt{p}}.
    \end{equation}
    Therefore, the independent columns of $V_i$ satisfy the conditions of Lemma \ref{lem:adamczak-opnorm} \cite[Thm. 2]{1012.0294} with constants independent of $i$.
    With $A_i \bydef \sqrt{p} V_i$, we have
    \begin{equation}\label{eq:operator-bound-prob}
        \P\left(\|A_i\|_\op > C \sqrt{\max\{p, n\}} + 6 \max_{j\neq i}\norm{X_j^{(i)}} \right) \leq e^{-c\sqrt{p}}.
    \end{equation}
    From \eqref{eq:boundedness-event} and \eqref{eq:operator-bound-prob}, we have
    \begin{equation}
        \P\left(\norm{A_i}_\op > C \sqrt{p} \right) \leq 2 e^{-c\sqrt{p}}.
    \end{equation}
    This is equivalent to $\P(\norm{K_i}_\op > C) \leq 2 e^{-c\sqrt{p}}$.
    To obtain moment bounds from this probability estimate, we also control $\|K_i\|_\op$ on the exceptional event. Since $K_i$ is positive semidefinite,
    \begin{equation}
        \norm{K_i}_{\op}\le \Tr(K_i)
        =\,\sum_{j\,<\,i}\norm{\widetilde f_j}^2+\sum_{j\,>\,i}\norm{f_j}^2.
    \end{equation}
    By the one-column bounds proved above and Minkowski,
    \begin{equation}
        \big\|\Tr(K_i)\big\|_{L^{2m}}
        \le \sum_{j<i}\norm{\widetilde f_j}^2_{L^{2m}}
        +\sum_{j>i}\norm{f_j}^2_{L^{2m}}
        \le C_m n
        \le C_m d^2.
    \end{equation}
    Let $E_i=\{\norm{K_i}_{\op}\le C\}$. Then by Cauchy--Schwarz
    \begin{equation}
        \E\norm{K_i}_{\op}^m
        \le C^m+\E\big[\Tr(K_i)^m\bone_{E_i^c}\big]
        \le C^m+\big\|\Tr(K_i)\big\|_{L^{2m}}^m\P(E_i^c)^{1/2}
        \le C_m,
    \end{equation}
    for all sufficiently large $p$. So
    \begin{equation}
    \norm{K_i}_{L^m(\op)}\le C_m.
    \end{equation}

    Conditioned on $F_\smi$, the vector $F_\smi\epsilon$ is Gaussian with covariance $K_i$. Let $(\lambda_1, \ldots, \lambda_r$ be the eigenvalues of $K_i$. Hence
    \begin{equation}
        \begin{aligned}
        \E\left[\norm{F_\smi\epsilon}^{2m}\mid F_\smi\right]
        &= \E\left[ \left( \sum_{j = 1}^r \lambda_j g_j^2\right)^m\right], \qquad g_j \iids \CN(0, 1) \\
        &=\, \sum_{\alpha_1 + \cdots + \alpha_r \,=\, m} \binom{m}{\alpha_1, \ldots, \alpha_r} \prod_{j \,=\, 1}^t \lambda_j^{\alpha_j} \E[g_j^{2\alpha_j}]
        \end{aligned}
    \end{equation}
    But $\prod_{j = 1}^r \E[g_j^{2\alpha_j}] \le C_m$. Therefore,
    \begin{equation}
         \E\left[\norm{F_\smi\epsilon}^{2m}\mid F_\smi\right] \,\le\, C_m \sum_{\alpha_1+ \ldots+ \alpha_r \,=\, m} \binom{m}{\alpha_1, \ldots, \alpha_r} \prod_{j \,=\, 1}^r \lambda_j^{\alpha_j} \,=\, C_m \left( \sum_{j \,=\, 1}^r \lambda_j \right)^m \,=\, C_m \Tr(K_i)^m
    \end{equation}
    Using again
    \begin{equation}
        \Tr(K_i)\,=\,\sum_{j\,<\,i}\norm{\widetilde f_j}^2+\sum_{j\,>\,i}\norm{f_j}^2
    \end{equation}
    and the bound $n\leq C_0p=\CO(d^2)$, we obtain $\|\Tr(K_i)\|_{L^m}\le C_m d^2$. Taking expectations gives
    \begin{equation}
        \norm{F_\smi\epsilon}_{L^m}\le C_m d.
    \end{equation}
    This completes the proof.
\end{proof}

\subsection{Proof of Lemma \ref{lem:one-column-moments}}

\begin{proof}[Proof of Lemma \ref{lem:one-column-moments}]
    Since $u$ is independent of $\mathcal G_i$, it is enough to prove the deterministic version.
    \noindent \emph{Gaussian column.}
    If $u=\widetilde f_i$, then $u\sim N(0,I_p/p)$. The identities \eqref{eq:linear-moment-1}--\eqref{eq:linear-moment-2} are immediate. The bound \eqref{eq:linear-moment-lm} follows from the fact that $\inner{u}{a}$ is centered Gaussian with variance $p^{-1}\norm a^2$. For \eqref{eq:quadratic-moment-lm}, only the symmetric part of $A$ matters:
    \begin{equation}
        \inner{u}{Au}=\inner{u}{A_{\mathrm s}u}
    \end{equation}
    where $A_s = \frac{A + A^\top}{2}$.

    The centered quadratic form
    \begin{equation}
        \inner{u}{A_{\mathrm s}u}-\frac1p\Tr A_{\mathrm s}
    \end{equation}
    has variance
    \begin{equation}
        \Var\left(\inner{u}{A_{\mathrm s}u}\right)
        =\frac{2}{p^2}\Tr(A_{\mathrm s}^2)
        \le \frac{2}{p}\norm{A_{\mathrm s}}_{\op}^2
        \le C d^{-2}\norm{A_{\mathrm s}}_{\op}^2
    \end{equation}
    Gaussian hypercontractivity gives
    \begin{equation}
        \left\|\inner{u}{Au}-\frac1p\Tr A\right\|_{L^m}\le C_m d^{-1}\norm{A_{\mathrm s}}_{\op},
    \end{equation}
    which is even better than \eqref{eq:quadratic-moment-lm}.
    \noindent \emph{Rank-one column.}
    Now let $u=f_i$. Write
    \begin{equation}
        Y\bydef \frac{x_ix_i^\top-I_d}{\sqrt{d(d+1)}}\in\Sym,
        \qquad
        u=\svecop(Y)
        \qquad
        x_i\sim N(0,I_d)
    \end{equation}
    Let $M=\mat(a)\in\Sym$. Then
    \begin{equation}\label{eq:linear-form-r1}
        \inner{u}{a}
        =\Tr(YM)
        =\frac{x_i^\top Mx_i-\Tr M}{\sqrt{d(d+1)}}.
    \end{equation}
    This is a centered Gaussian quadratic form, so \eqref{eq:linear-moment-1} holds. Also,
    \begin{equation}
        \E\left[\inner{u}{a}\inner{u}{b}\right]
        =\frac{2}{d(d+1)}\Tr(\mat(a)\mat(b))
        =\frac1p\inner{a}{b},
    \end{equation}
    which is \eqref{eq:linear-moment-2}. Since
    \begin{equation}
        \Var(x_i^\top Mx_i-\Tr M)=2\norm{M}_{\Frob}^2=2\norm{a}^2,
    \end{equation}
    hypercontractivity for Gaussian quadratic forms gives \eqref{eq:linear-moment-lm}.

    It remains to prove \eqref{eq:quadratic-moment-lm}. Let $\SA:\Sym\to\Sym$ be the self-adjoint operator
    \begin{equation}
        \SA(H)\bydef \mat(A_{\mathrm s}\svecop(H)).
    \end{equation}
    Because $\svecop$ is an isometry,
    \begin{equation}
        \norm{\SA}_{\Frob\to\Frob}=\norm{A_{\mathrm s}}_{\op}.
    \end{equation}
    Using $u=\svecop(Y)$,
    \begin{equation}
        \inner{u}{Au}=\inner{Y}{\SA Y}.
    \end{equation}
    By the covariance identity \eqref{eq:column-covariance},
    \begin{equation}
        \E\inner{Y}{\SA Y}
        =\frac1p\Tr(A_{\mathrm s})
        =\frac1p\Tr A.
    \end{equation}

    Write $Y_0=x_ix_i^\top-I_d$, so that $Y=Y_0/\sqrt{d(d+1)}$, and define
    \begin{equation}
        g(x_i)\bydef \frac{1}{d(d+1)}\inner{Y_0}{\SA Y_0}.
    \end{equation}
    Then
    \begin{equation}
        g(x_i)=\inner{u}{Au}.
    \end{equation}
    A direct differentiation gives, for every $h\in\R^d$,
    \begin{equation}
        Dg(x_i)[h]
        =\frac{4}{d(d+1)}\inner{h}{\SA(Y_0)x_i},
    \end{equation}
    hence
    \begin{equation}
        \norm{\nabla g(x_i)}
        \le \frac{4}{d(d+1)}\norm{\SA}_{F\to F}\norm{Y_0}_{\Frob}\norm{x_i}
        \le \frac{4}{d(d+1)}\norm{A_{\mathrm s}}_{\op}\norm{Y_0}_{\Frob}\norm{x_i}.
    \end{equation}
    Let $S=\norm{x_i}^2\sim \chi_d^2$. Since
    \begin{equation}
        \norm{Y_0}_{\Frob}^2=S^2-2S+d,
    \end{equation}
    we obtain
    \begin{equation}
        \E\big[\norm{Y_0}_{\Frob}^2\norm{x_i}^2\big]
        =\E[S^3-2S^2+dS]
        =d^3+5d^2+4d
        \le C d^3.
    \end{equation}
    The Gaussian Poincar\'e inequality therefore gives
    \begin{equation}
        \Var(g)\le \E\norm{\nabla g}^2\le C d^{-1}\norm{A_{\mathrm s}}_{\op}^2.
    \end{equation}
    Since $g-\E g$ is a polynomial of degree at most four in $x_i$, Gaussian hypercontractivity yields
    \begin{equation}
        \left\|\inner{u}{Au}-\frac1p\Tr A\right\|_{L^m}
        \le C_m \|g-\E g\|_{L^2}
        \le C_m d^{-1/2}\norm{A_{\mathrm s}}_{\op}.
    \end{equation}
    This proves \eqref{eq:quadratic-moment-lm}. The unconditional estimates \eqref{eq:linear-moment-uncond}--\eqref{eq:quadratic-moment-uncond} follow by taking expectations and using $p\asymp d^2$.
\end{proof}

\subsection{Proof of Lemma \ref{lem:gaussian-quad-cumulants}}

\begin{proof}[Proof of Lemma \ref{lem:gaussian-quad-cumulants}]
    Fix $r \ge 2$, we write
    \begin{equation}
        A(t) \,\bydef\, \sum_{j \,=\, 1}^r t_j M_j, \quad t \,=\, (t_1, \ldots, t_r) \,\in\, \R^r
    \end{equation}
    Choose $t$ such that $\norm{A(t)}_\op < 1/2$. We thus have
    \begin{equation}
        \E \exp\left( \sum_{j=1}^r t_j q_{M_j}(x) \right) = \exp\left(-\Tr(A(t))\right) \det(I_d - 2A(t))^{-1/2}
    \end{equation}
    The cumulant-generating function has the Taylor expansion
    \begin{equation}
        \begin{aligned}
            K(t) &\bydef\, \log \E\exp\left( \sum_{j\,=\,1}^r t_j q_{M_j}(x) \right) \,=\, -\Tr A(t) - \frac 1 2 \Tr \log(I_d - 2A(t)) \\
            &=\, -\Tr A(t) + \frac 1 2 \sum_{m\,=\,1}^\infty \frac{2^m}{m}\Tr(A(t)^m) \,=\, \sum_{m \,=\, 2}^\infty \frac{2^{m-1}}{m} \Tr(A(t)^m)
        \end{aligned}
    \end{equation}
    But the joint cumulants are given by the mixed derivative
    \begin{equation}
        \kappa\left(q_{M_1}(x), \ldots, q_{M_r}(x) \right) = \partial_{t_1} \cdots \partial_{t_r} K(t) \ \Big|_{t=0}
    \end{equation}
    But observe that $\frac{2^{m-1}}{m} \Tr(A(t)^m)$ is a polynomial of degree $m$ in $t$. Therefore, only the term with $m = r$ contributes to the $r$th joint cumulant. We thus have
    \begin{equation}
        \kappa\left(q_{M_1}(x), \ldots, q_{M_r}(x) \right) \,=\, \frac{2^{r-1}}{r} \sum_{\sigma \,\in\, S_r} \Tr\left( M_{\sigma(1)} \cdots M_{\sigma(r)} \right)
    \end{equation}
    Since trace is invariant under cyclic permutations, every cyclic class in $S_r$ contributes the same value to the sum and contains $r - 1$ permutations. Grouping the sum over $S_r$ into cyclic classes therefore gives the stated formula.
\end{proof}

\subsection{Proof of Lemma \ref{lem:refined-centered-quad}}

\begin{proof}[Proof of Lemma \ref{lem:refined-centered-quad}]
    Since $u$ is independent of $\CG_i$, it is enough to prove the deterministic version.

    \noindent \emph{Gaussian column.}
    If $u = \widetilde f_i$, then $u \sim \CN(0, I_p/p)$. Note that
    \begin{equation}
        \inner{u}{Au} - \frac 1 p \Tr A = \inner{u}{A_{\mathrm s}u} - \frac 1 p \Tr A_{\mathrm s}.
    \end{equation}
    This is a centered Gaussian quadratic form with variance
    \begin{equation}
        \Var(\inner{u}{A_{\mathrm s}u}) \,=\, \frac{2}{p^2} \Tr(A_{\mathrm s}^2) \,=\, \frac{2}{p^2} \norm{A_{\mathrm s}}_{\Frob}^2.
    \end{equation}
    Gaussian hypercontractivity therefore gives
    \begin{equation}
        \norm{\inner{u}{Au} - \frac 1 p \Tr A}_{L^m} \,\le\, C_m p^{-1} \norm{A_{\mathrm s}}_{\Frob} \,\le\, C_m d^{-2} \norm{A_{\mathrm s}}_{\Frob},
    \end{equation}
    which is stronger than the claimed bound.

    \noindent \emph{Rank-one column.}
    Now let $u = f_i$. Write
    \begin{equation}
        Y \bydef \frac{xx^\top - I_d}{\sqrt{d(d+1)}} \in \Sym,
        \qquad
        u = \svecop(Y),
        \qquad
        x \sim \CN(0, I_d),
    \end{equation}
    and let $Y_0 = xx^\top - I_d$, so that $Y = Y_0 / \sqrt{d(d+1)}$. Recall the self-adjoint operator $\SA : \Sym \to \Sym$ defined by
    \begin{equation}
        \SA(H) \bydef \mat(A_{\mathrm s} \svecop(H)).
    \end{equation}
    Choose a Frobenius-orthonormal eigenbasis $(E_\alpha)_{\alpha = 1}^p$ of $\Sym$ such that
    \begin{equation}
        \SA(E_\alpha) = \lambda_\alpha E_\alpha,
        \qquad \alpha = 1, \ldots, p.
    \end{equation}
    Since $\svecop \circ \SA \circ \mat = A_{\mathrm s}$, the eigenvalues of $\SA$ are the eigenvalues of $A_{\mathrm s}$. Therefore,
    \begin{equation}
        \sum_{\alpha = 1}^p \lambda_\alpha^2 = \norm{A_{\mathrm s}}_{\Frob}^2,
        \qquad
        \sum_{\alpha \,=\, 1}^p \lambda_\alpha \,=\, \Tr A_{\mathrm s} \,=\, \Tr A.
    \end{equation}

    For each $\alpha$, define
    \begin{equation}
        q_\alpha \bydef \inner{Y_0}{E_\alpha} = x^\top E_\alpha x - \Tr(E_\alpha).
    \end{equation}
    By Lemma \ref{lem:gaussian-quad-cumulants} with $r = 2$,
    \begin{equation}
        \E[q_\alpha q_\beta] = 2 \Tr(E_\alpha E_\beta) = 2 \delta_{\alpha\beta}.
    \end{equation}
    In particular, $\E[q_\alpha^2] = 2$ for every $\alpha$. Since $Y_0 = \sum_{\alpha = 1}^p q_\alpha E_\alpha$, we have
    \begin{equation}
        \inner{u}{Au} \,=\, \inner{Y}{\SA Y} \,=\, \frac{1}{d(d+1)} \sum_{\alpha \,=\, 1}^p \lambda_\alpha q_\alpha^2.
    \end{equation}
    Using $p = d(d+1)/2$, we also get
    \begin{equation}
        \frac 1 p \Tr A \,=\, \frac{2}{d(d+1)} \sum_{\alpha \,=\, 1}^p \lambda_\alpha \,=\, \frac{1}{d(d+1)} \sum_{\alpha \,=\, 1}^p \lambda_\alpha \E[q_\alpha^2].
    \end{equation}
    Hence
    \begin{equation}
        \inner{u}{Au} - \frac 1 p \Tr A = \frac{1}{d(d+1)} g(x),
        \qquad
        g(x) \bydef \sum_{\alpha = 1}^p \lambda_\alpha (q_\alpha^2 - 2).
    \end{equation}
    The random variable $g(x)$ is a polynomial of degree four in the Gaussian vector $x$, so Gaussian hypercontractivity gives
    \begin{equation}
        \norm{g(x)}_{L^m} \le C_m \norm{g(x)}_{L^2}.
    \end{equation}

    It remains to bound $\Var(g)$. Since each $q_\alpha$ is centered, the moment-cumulant formula gives
    \begin{equation}
        \Cov(q_\alpha^2, q_\beta^2) = 2\kappa_2(q_\alpha, q_\beta)^2 + \kappa_4(q_\alpha, q_\alpha, q_\beta, q_\beta).
    \end{equation}
    By the previous display,
    \begin{equation}
        2\kappa_2(q_\alpha, q_\beta)^2 = 8\delta_{\alpha\beta}.
    \end{equation}
    By Lemma \ref{lem:gaussian-quad-cumulants} with $r = 4$, the fourth cumulant is a finite sum of traces of cyclic words of length four in $E_\alpha$ and $E_\beta$. Every such word is either of the form $E_\alpha^2 E_\beta^2$ or $E_\alpha E_\beta E_\alpha E_\beta$. Therefore,
    \begin{equation}
        \left|\kappa_4(q_\alpha, q_\alpha, q_\beta, q_\beta)\right|
        \le C \left(\Tr(E_\alpha^2 E_\beta^2) + \left|\Tr(E_\alpha E_\beta E_\alpha E_\beta)\right|\right).
    \end{equation}
    Since $E_\alpha$ and $E_\beta$ are symmetric, Cauchy--Schwarz gives
    \begin{equation}
        \left|\Tr(E_\alpha E_\beta E_\alpha E_\beta)\right|
        = \left|\Tr((E_\beta E_\alpha)(E_\alpha E_\beta))\right|
        \le \norm{E_\beta E_\alpha}_{\Frob} \norm{E_\alpha E_\beta}_{\Frob}
        = \Tr(E_\alpha^2 E_\beta^2).
    \end{equation}
    Hence
    \begin{equation}
        \left|\Cov(q_\alpha^2, q_\beta^2)\right| \le C \left(\delta_{\alpha\beta} + \Tr(E_\alpha^2 E_\beta^2)\right).
    \end{equation}
    Expanding the variance of $g$ gives
    \begin{equation}
        \Var(g) \,\le\, C \sum_{\alpha \,=\, 1}^p \lambda_\alpha^2 + C \sum_{\alpha, \beta \,=\, 1}^p |\lambda_\alpha \lambda_\beta| \Tr(E_\alpha^2 E_\beta^2).
    \end{equation}

    To control the double sum, note first that $\sum_{\beta = 1}^p E_\beta^2$ does not depend on the choice of Frobenius-orthonormal basis of $\Sym$. To see this, if $(\widetilde E_\beta)_{\beta = 1}^p$ is another such basis and $\widetilde E_\beta = \sum_{\gamma = 1}^p O_{\beta\gamma} E_\gamma$ for some $O \in O(p)$, then
    \begin{equation}
        \sum_{\beta \,=\, 1}^p \widetilde E_\beta^2 \,=\, \sum_{\beta, \gamma, \eta \,=\, 1}^p O_{\beta\gamma} O_{\beta\eta} E_\gamma E_\eta \,=\, \sum_{\gamma \,=\, 1}^p E_\gamma^2.
    \end{equation}
    Therefore we may compute this sum using the standard symmetric basis
    \begin{equation}
        F_{aa} = e_a e_a^\top,
        \qquad
        F_{ab} = \frac{e_a e_b^\top + e_b e_a^\top}{\sqrt 2},
        \qquad 1 \,\le\, a \,<\, b \,\le\, d.
    \end{equation}
    A direct calculation gives
    \begin{equation}
        \sum_{\beta = 1}^p E_\beta^2 = \frac{d+1}{2} I_d.
    \end{equation}

    Now set $T_{\alpha\beta} \bydef \Tr(E_\alpha^2 E_\beta^2) \ge 0$. By Cauchy--Schwarz,
    \begin{equation}
        \begin{aligned}
            \sum_{\alpha, \beta = 1}^p |\lambda_\alpha \lambda_\beta| T_{\alpha\beta}
            &\le \left(\sum_{\alpha, \beta = 1}^p \lambda_\alpha^2 T_{\alpha\beta}\right)^{1/2}
            \left(\sum_{\alpha, \beta \,=\, 1}^p \lambda_\beta^2 T_{\alpha\beta}\right)^{1/2}\,=\, \sum_{\alpha, \beta \,=\, 1}^p \lambda_\alpha^2 T_{\alpha\beta} \\
            &=\, \sum_{\alpha \,=\, 1}^p \lambda_\alpha^2 \Tr\left(E_\alpha^2 \sum_{\beta \,=\, 1}^p E_\beta^2\right) \,=\, \frac{d+1}{2} \sum_{\alpha \,=\, 1}^p \lambda_\alpha^2.
        \end{aligned}
    \end{equation}
    We obtain
    \begin{equation}
        \Var(g) \,\le\, C d \sum_{\alpha \,=\, 1}^p \lambda_\alpha^2 \,=\, C d \norm{A_{\mathrm s}}_{\Frob}^2.
    \end{equation}
    Therefore,
    \begin{equation}
        \norm{\inner{u}{Au} - \frac 1 p \Tr A}_{L^m}
        \le \frac{C_m}{d(d+1)} d^{1/2} \norm{A_{\mathrm s}}_{\Frob}
        \le C_m d^{-3/2} \norm{A_{\mathrm s}}_{\Frob}.
    \end{equation}

    This proves the deterministic estimate. Averaging over $\CG_i$, we get
    \begin{equation}
        \norm{\inner{u}{Au} - \frac 1 p \Tr A}_{L^m}
        \le C_m d^{-3/2} \norm{A_{\mathrm s}}_{L^m(\Frob)}
        \le C_m d^{-3/2} \norm{A}_{L^{2m}(\Frob)},
    \end{equation}
    which proves the first claim. The second follows from $\norm{A_{\mathrm s}}_{L^{2m}(\Frob)} \le \norm{A}_{L^{2m}(\Frob)}$.
\end{proof}

\subsection{Proof of Proposition \ref{prop:delta-moments}}

\begin{proof}[Proof of Proposition \ref{prop:delta-moments}]
    For $t = 0$, we clearly have $q^0(u) = 0$ because the iterates start at the same place. By \eqref{eq:w-exact} of Proposition \ref{prop:exact-loo}, we have
    \begin{equation}
        w^0(u) = u \left(\kappa_0 \inner{u}{{\bf s}^* - \widehat{\bf s}^0} + \nu_0 \epsilon_i \right)
    \end{equation}
    But by Assumption \ref{ass:signal-init}, by triangle inequality,
    \begin{equation}
        \norm{{\bf s}^* - \widehat{\bf s}^0} \leq C \sqrt{d}
    \end{equation}
    From \eqref{eq:linear-moment-uncond} of Lemma \ref{lem:one-column-moments} with $a = {\bf s}^* - \widehat{\bf s}^0$, we have
    \begin{equation}
        \norm{\inner{u}{{\bf s}^* - \widehat{\bf s}^0}}_{L^{2m}} \,\leq\, C_m d^{-1} \norm{{\bf s}^* - \widehat{\bf s}^0}_{L^{2m}} \,\leq\, C_m d^{-1/2}
    \end{equation}
    Since we also have $\E \norm{u}^m \leq C_m$ and $\nu_0 \epsilon_i = \CO_{L^m}(d^{-1/2})$, we have
    \begin{equation}
        \norm{w^0(u)}_{L^m} \leq C_m d^{-1/2}
    \end{equation}
    Therefore, $\norm{m^0(u)}_{L^m} \leq C_m d^{-1/2}$. Then by \eqref{eq:q-linearized} and Proposition \ref{prop:exact-loo}, we have
    \begin{equation}
        \norm{q^1(u)}_{L^m} \,\leq\, \norm{L^0_i}_\op \norm{m^0(u)}_{L^m} + \norm{\varepsilon^0_i(u)}_{L^m} \,\leq\, C_m d^{-1/2} + C_m \norm{m^0(u)}_{L^{2m}}^2 \,\leq\, C_m d^{-1/2}
    \end{equation}
    where we used Lemma \ref{lem:dk-bounds} for $\norm{L^0_i}_\op$.

    The same estimate gives $\norm{\varepsilon_i^0(u)}_{L^m}\le C_m d^{-1}$. We proceed by induction. Assume that for some $t$ and every $m$,
    \begin{equation}
        \max_{s \,\le\, t} \norm{q^s(u)}_{L^m} + \max_{s \,\le\, t - 1} \norm{w^s(u)}_{L^m} \,\leq\, C_m d^{-1/2}
    \end{equation}
    Then by \eqref{eq:beta-exact} of Proposition \ref{prop:exact-loo}, we have
    \begin{equation}
        \beta^t_i(u) = \kappa_t \inner{u}{{\bf s}^* - \widehat{\bf s}^t_\smi} - \kappa_t \inner{u}{q^t(u)} + \nu_t \epsilon_i
    \end{equation}
    Applying \eqref{eq:linear-moment-uncond} from Lemma \ref{lem:one-column-moments}, now with $a = {\bf s}^* - \widehat{\bf s}^t_\smi$, using Lemma \ref{lem:base-det-bounds}, we have
    \begin{equation}
        \norm{\inner{u}{{\bf s}^* - \widehat{\bf s}^t_\smi}}_{L^{2m}} \,\le\, C_m d^{-1} \norm{{\bf s}^* - \widehat{\bf s}^t_\smi}_{L^{2m}} \,\le\, C_m d^{-1/2}
    \end{equation}
    Now $ \kappa_t \inner{u}{q^t(u)}$ is bounded by Cauchy--Schwarz
    \begin{equation}
        \norm{u}_{L^{2m}} \norm{q^t(u)}_{L^{2m}} \le C_m d^{-1/2}
    \end{equation}
    So $\norm{\beta^t_i(u)}_{L^m} \le C_m d^{-1/2}$. Therefore, by \eqref{eq:w-exact} of Proposition \ref{prop:exact-loo}, Lemma \ref{lem:sensing-envelope}, and the induction hypothesis,
    \begin{equation}
        \begin{aligned}
            \norm{w^t(u)}_{L^m} &\le\, |\kappa_t| \norm{K_i}_{L^{2m}(\op)} \norm{q^t(u)}_{L^{2m}} + |\rho_t| \norm{w^{t - 1}(u)}_{L^m} + \norm{u}_{L^{2m}} \norm{\beta^t_i(u)}_{L^{2m}} \,\leq\, C_m d^{-1/2}
        \end{aligned}
    \end{equation}
    Therefore, we have $\norm{ m^t(u)}_{L^m} \le C_m d^{-1/2}$. By \eqref{eq:q-linearized} of Proposition \ref{prop:exact-loo}, we have
    \begin{equation}
        \norm{q^{t + 1}(u)}_{L^m} \le C_m d^{-1/2}+ C_m\norm{m^t(u)}_{L^{2m}}^2 \le C_m d^{-1/2}
    \end{equation}
    Lastly,
    \begin{equation}
        \norm{\varepsilon^t_i(u)}_{L^m} \le C \norm{m^t(u)}_{L^{2m}}^2 \le C_m d^{-1}
    \end{equation}
    So by induction, the claim follows.
\end{proof}

\subsection{Verification of Gaussian State Evolution}\label{app:gaussian-se-verification}
We derive Proposition~\ref{prop:gaussian-se} from the nonseparable Gaussian AMP theorem of Berthier, Montanari, and Nguyen \cite[Theorem~1]{BerthierRaphael2020Sefa}. We check the coordinate scaling, covariance recursion, and observable test functions for the matrix formulation in Section~\ref{subsec:se-transfer}.

\begin{proof}[Proof of Proposition \ref{prop:gaussian-se}]
    Let $\widetilde F=[\widetilde f_1,\ldots,\widetilde f_n]\in\R^{p\times n}$ be the Gaussian feature matrix of \eqref{eq:gaussian-surrogate-features}, and set
    \begin{equation}\label{eq:se-proof-scaled-variables}
        A_d\bydef\sqrt{\frac pn}\,\widetilde F^\top,
        \qquad
        \lambda_d\bydef\sqrt{(d+1)\alpha_d},
        \qquad
        \theta^\star\bydef\lambda_d\svecop(S^\star),
        \qquad
        \widehat\theta^t\bydef\lambda_d\svecop(\widehat S_{\mathrm G}^t).
    \end{equation}
    The entries of $A_d$ are independent $\CN(0,1/n)$ variables, $A_d\theta^\star=\CA_{\mathrm G}(S^\star)$, and $y_{\mathrm G}=A_d\theta^\star+\epsilon$. Define
    \begin{equation*}
        \eta_{t,d}(x)\bydef
        \lambda_d\svecop\left(\sigma_t\left(\mmat(x/\lambda_d)\right)\right),
        \qquad
        \mathfrak r^t\bydef\gamma_t r^t_{\mathrm G},
        \qquad
        \overline b_t\bydef b_t/\gamma_{t-1}\quad(t\geq1).
    \end{equation*}
    Using $c_t/\gamma_t=p/(n(d+1))$, the Gaussian recursion \eqref{eq:AMP-recursion} becomes
    \begin{equation}\label{eq:se-proof-standard-amp}
        \mathfrak r^t
        =y_{\mathrm G}-A_d\widehat\theta^t+\overline b_t\mathfrak r^{t-1},
        \qquad
        \widehat\theta^{t+1}
        =\eta_{t,d}\big(\widehat\theta^t+A_d^\top\mathfrak r^t\big),
    \end{equation}
    with $\mathfrak r^{-1}=0$ and $\overline b_0=0$.

    We first match \eqref{eq:se-proof-standard-amp}, including its sign convention, to the asymmetric AMP recursion of \cite[Equations~(27)--(29)]{BerthierRaphael2020Sefa}. Let $\eta_{-1,d}(x)\bydef\widehat\theta^0$ and define
    \begin{equation}\label{eq:se-proof-change-variables}
        e_t(u)\bydef\eta_{t-1,d}(\theta^\star-u)-\theta^\star,
        \qquad
        g_t(v)\bydef v-\epsilon,
    \end{equation}
    together with $u^0=0$ and
    \begin{equation}\label{eq:se-proof-bmn-variables}
        u^{t+1}\bydef\theta^\star-\big(\widehat\theta^t+A_d^\top\mathfrak r^t\big),
        \qquad
        v^t\bydef\epsilon-\mathfrak r^t.
    \end{equation}
    Then $e_t(u^t)=\widehat\theta^t-\theta^\star$ and $g_t(v^t)=-\mathfrak r^t$, and \eqref{eq:se-proof-standard-amp} is equivalent to
    \begin{equation}\label{eq:se-proof-bmn-recursion}
        u^{t+1}=A_d^\top g_t(v^t)-e_t(u^t),
        \qquad
        v^t=A_de_t(u^t)+\overline b_tg_{t-1}(v^{t-1}),
    \end{equation}
    with $g_{-1}=0$. Thus the output-side Onsager coefficient in \cite{BerthierRaphael2020Sefa} is $d_t=1$, while its input-side coefficient is $b_t^{\mathrm{BMN}}=-\overline b_t$.

    The spectral map induced by $\sigma_t$ is Lipschitz in Frobenius norm with constant $\norm{\sigma_t'}_\infty$. Since $\svecop$ is an isometry, the maps $\eta_{t,d}$, and hence $e_t$, are uniformly Lipschitz in Euclidean norm; the maps $g_t$ are $1$-Lipschitz. Moreover,
    \begin{equation}\label{eq:se-proof-divergence-identification}
        \div\eta_{t,d}(x)
        =\CD_t\big(\mmat(x/\lambda_d)\big).
    \end{equation}
    The map $e_0$ is constant, and for $t\geq1$,
    \begin{equation*}
        \frac1n\E\div e_t(Z)
        =-\frac1n\E\CD_{t-1}(\overline R^{t-1})
        =-\frac{b_t}{\gamma_{t-1}}
        =-\overline b_t,
    \end{equation*}
    where the centered Gaussian input may be replaced by its negative without changing its covariance. This is precisely the Onsager coefficient in \eqref{eq:se-proof-bmn-recursion}; the sign is therefore accounted for explicitly.

    It remains to identify the covariance recursion. Let $(U^0,\ldots,U^t)$ be centered jointly Gaussian vectors in $\R^p$ with
    \begin{equation*}
        \E[U^s(U^r)^\top]=Q_{s,r}I_p,
        \qquad
        W^r\bydef\mmat(U^r)/\lambda_d.
    \end{equation*}
    Because off-diagonal coordinates of $\mmat(U^r)$ are divided by $\sqrt2$, the matrices $(W^0,\ldots,W^t)$ have exactly the entrywise covariance \eqref{eq:se-noise-covariance}. Furthermore,
    \begin{equation}\label{eq:se-proof-normalization-identity}
        \frac{\lambda_d^2}{n}=\frac{d+1}{p}=\frac2d.
    \end{equation}
    We condition on $S^\star$, $\widehat S^0$, and the probability-one event $n^{-1}\norm{\epsilon}^2\to\Delta$. If $(Z_1,Z_2)$ is a centered Gaussian pair in $\R^n$ with cross-covariance $K_{12}I_n$, then
    \begin{equation}\label{eq:se-proof-output-overlap}
        \frac1n\E\langle g_s(Z_1),g_t(Z_2)\rangle
        =K_{12}+\frac1n\norm{\epsilon}^2
        \longrightarrow K_{12}+\Delta.
    \end{equation}
    Write $\Sigma$ and $\mathsf T$ for the two covariance arrays in \cite[Equations~(25)--(26)]{BerthierRaphael2020Sefa}. We claim, inductively, that
    \begin{equation}\label{eq:se-proof-covariance-identification}
        \Sigma_{s,t}=C_{s,t},
        \qquad
        \mathsf T_{s+1,t+1}=Q_{s,t}.
    \end{equation}
    Indeed, \eqref{eq:se-proof-normalization-identity} gives
    \begin{equation*}
        \Sigma_{0,0}
        =\lim_{d\to\infty}\frac1n\norm{\widehat\theta^0-\theta^\star}^2
        =C_{0,0},
    \end{equation*}
    and \eqref{eq:se-proof-output-overlap} then gives $\mathsf T_{1,1}=\Delta+C_{0,0}=Q_{0,0}$. If \eqref{eq:se-proof-covariance-identification} holds through time $t$, the Gaussian input to $e_{t+1}$ has covariance $Q$ and
    \begin{equation*}
        \frac1n\E\langle e_s(Z_s),e_{t+1}(Z_{t+1})\rangle
        =\frac2d\E\Tr\big((\overline S^s-S^\star)(\overline S^{t+1}-S^\star)\big)
        \longrightarrow C_{s,t+1},
    \end{equation*}
    where \eqref{eq:se-proof-normalization-identity} was used. Assumption~\ref{ass:se-regularity} supplies these limits for every covariance pair. Applying \eqref{eq:se-proof-output-overlap} once more gives $\mathsf T_{s+1,t+2}=C_{s,t+1}+\Delta=Q_{s,t+1}$, completing the induction. The remaining norm and overlap hypotheses of \cite[Theorem~1]{BerthierRaphael2020Sefa} follow from Assumptions~\ref{ass:signal-init} and \ref{ass:se-regularity} and the Lipschitz bounds above. That theorem therefore yields joint state evolution for the scaled residuals and pre-denoising matrices in \eqref{eq:se-proof-standard-amp}.

    We finally verify that the three observables are admissible test functions. For $\theta,\theta'\in\R^p$, set
    \begin{equation*}
        A\bydef\sigma_t\big(\mmat(\theta/\lambda_d)\big),
        \qquad
        A'\bydef\sigma_t\big(\mmat(\theta'/\lambda_d)\big).
    \end{equation*}
    Then
    \begin{equation}\label{eq:se-proof-spectral-lipschitz}
        \norm{A-A'}_\Frob
        \leq\frac{\norm{\sigma_t'}_\infty}{\lambda_d}\norm{\theta-\theta'}.
    \end{equation}
    By Assumption~\ref{ass:signal-init}, \eqref{eq:se-proof-spectral-lipschitz}, and Cauchy--Schwarz,
    \begin{equation*}
        \frac1d\abs{\Tr((A-A')S^\star)}
        \leq\frac{\norm{S^\star}_\Frob\norm{\sigma_t'}_\infty}{d\lambda_d}\norm{\theta-\theta'}
        \leq\frac{C}{\sqrt n}\norm{\theta-\theta'},
    \end{equation*}
    where the last inequality uses \eqref{eq:se-proof-normalization-identity}. If $B$ is either $\widehat S^0$ or the output of a bounded spectral denoiser, then $\norm{B}_\Frob\leq C\sqrt d$, and the same calculation with $S^\star$ replaced by $B$ applies to each argument of a cross-time overlap. Thus the signal and cross-time overlaps are uniformly Lipschitz functions of the scaled AMP coordinates under the normalized Euclidean metric of \cite{BerthierRaphael2020Sefa}. For the spectral statistic, order the eigenvalues of $A$ and $A'$ consistently. Lidskii's inequality \cite[Section~III.4]{Bhatia1997} and the trace-norm bound give
    \begin{equation*}
        \sum_{j=1}^d\abs{\lambda_j(A)-\lambda_j(A')}
        \leq\norm{A-A'}_{S_1}
        \leq\sqrt d\norm{A-A'}_\Frob.
    \end{equation*}
    Since $\mu$ is Lipschitz on $I_\ast$,
    \begin{equation*}
        \frac1d\abs{\Tr\mu(A)-\Tr\mu(A')}
        \leq\frac{\operatorname{Lip}_{I_\ast}(\mu)}{\sqrt d}\norm{A-A'}_\Frob
        \leq\frac{C}{\sqrt n}\norm{\theta-\theta'}.
    \end{equation*}
    Hence all three observables are uniformly pseudo-Lipschitz test functions. Their state-evolution expectations are exactly the predictions in \eqref{eq:se-predictions}; at $t=0$ the observables are deterministic and already equal to their predictions. This proves \eqref{eq:gaussian-se-observables}.
\end{proof}

\section{Effect of a Spectral Spike}\label{app:numerical-spike}

A rank-one measurement of a spike has nonzero skewness. We vary the spike strength through the family \eqref{eq:numerical-signal-interpolation}: $a=0$ gives a bounded spectrum, while $a=1$ gives a pure spike. Along the spike direction, a single rank-one response contains $a(G^2-1)$. Its third cumulant is $8a^3$, whereas the covariance-matched Gaussian response is symmetric and has zero third cumulant. The two projections have the same mean and variance, but different skewness.

We examine this effect in the pre-denoising matrix $R^0_{\sfe}$ from \eqref{eq:AMP-recursion}. Write its eigenvalues in increasing order as $\lambda_1(R^0_{\sfe})\leq\cdots\leq\lambda_d(R^0_{\sfe})$, and define
\begin{equation*}
    m_{3,\sfe}^{\mathrm{bulk}}
    \bydef \frac1d\sum_{j=1}^{d-1}\lambda_j(R^0_{\sfe})^3.
\end{equation*}
We remove the largest eigenvalue for every $a$, including $a=0$, so that the same statistic is compared throughout. This separates the movement of the bulk spectrum from the outlier created by the spike. A direct expansion at the first AMP step gives
\begin{equation}\label{eq:numerical-mechanism-prediction}
    \sqrt d\,
    \big(\E m_{3,\ro}^{\mathrm{bulk}}-
    \E m_{3,\mathrm G}^{\mathrm{bulk}}\big)
    =\frac{4a^3}{\alpha^2}+o(1).
\end{equation}

\begin{figure}
    \centering
    \includegraphics[width=\textwidth]{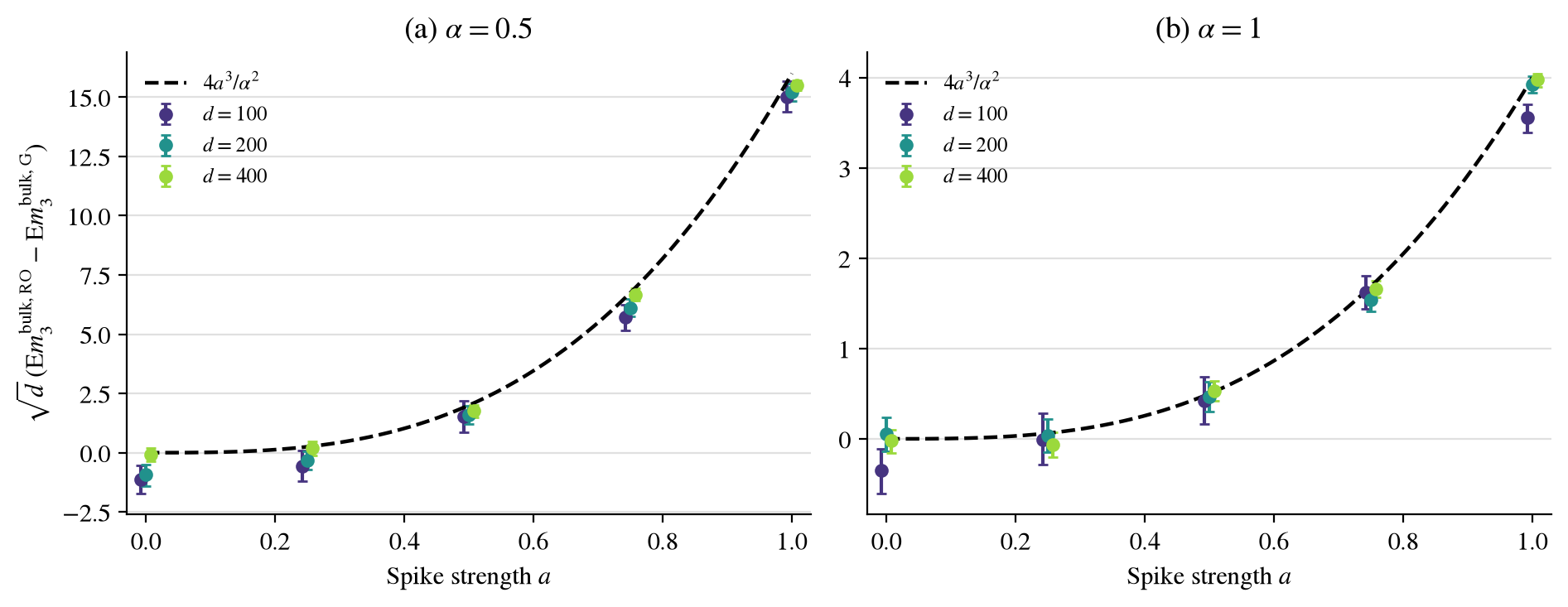}
    \caption{Effect of a spike on the first-step bulk spectrum. For the signals \eqref{eq:numerical-signal-interpolation}, the points show the scaled difference of the bulk third moments of $R^0$ over 200 paired trials, with $\Delta=0.5$ and $d\in\{100,200,400\}$. The dashed curves are the prediction $4a^3/\alpha^2$ in \eqref{eq:numerical-mechanism-prediction}. The agreement improves with dimension and is clearest at $a\in\{0.75,1\}$.}
    \label{fig:numerical-spike-mechanism}
\end{figure}

Figure~\ref{fig:numerical-spike-mechanism} shows an order-$d^{-1/2}$ difference in the bulk third moment before denoising. The skewness of the rank-one response survives covariance matching and affects the bulk spectrum even after the outlier is removed. This diagnostic lies outside Theorem~\ref{thm:final-universality}, which concerns observables of the bounded AMP iterates.

\FloatBarrier
\bibliography{references}
\bibliographystyle{plain}

\end{document}